\documentclass[a4paper,11pt]{article} 

\usepackage{amsmath, amssymb, amsthm}
\usepackage[mathscr]{euscript}     
\usepackage{qip} 
\usepackage[section]{thmenvironments} 
\usepackage{layoutcommands} 

\usepackage{underscore} 
\usepackage[shortcuts]{extdash} 
\usepackage{authblk} 
\usepackage{appendix} 
\usepackage{microtype} 
\usepackage{enumitem} 
\setlist{nolistsep}
\usepackage[margin=10pt,font=small,labelfont=bf,labelsep=endash,hypcap=true]{caption}
\usepackage[margin=10pt,font=small,labelfont=bf,labelformat=parens,labelsep=space,hypcap=false,list=true]{subcaption}
\usepackage{tikz} 
\usetikzlibrary{shapes} 

\tikzset{sArrow/.style={->,>=stealth,thick}}
\tikzset{arrowLabel/.style={auto}}
\tikzset{largeResource/.style={draw,thick,minimum width=1.618*2cm,minimum height=2cm}}
\tikzset{lrnode/.style={minimum width=1.36*2cm,minimum height=.2cm}}
\tikzset{llrnode/.style={minimum width=.2cm,minimum height=1.5cm}}
\tikzset{tlrnode/.style={minimum width=.2cm,minimum height=.5cm}}
\tikzset{thinResource/.style={draw,thick,minimum width=1.618*2cm,minimum height=1cm}}
\tikzset{filter/.style={draw,thick,minimum width=1.618cm,minimum height=1cm}}
\tikzset{lineFilter/.style={draw,ultra thick,minimum width=.8cm,inner sep=0}}
\tikzset{protocol/.style={draw,thick,minimum width=1.545cm,minimum height=2.5cm}}
\tikzset{pnode/.style={minimum width=1cm,minimum height=.5cm}}
\tikzset{protocolLong/.style={draw,thick,minimum height=1cm,minimum width=2.8cm}}
\tikzset{simulator/.style={draw,thick,minimum width=1.618*2cm,minimum height=1.7cm}}
\tikzset{snode/.style={minimum width=1.1cm,minimum height=1.2cm}}
\tikzset{largesnode/.style={minimum width=2.7cm,minimum height=.7cm}}
\tikzset{innersim/.style={minimum width=.83cm,minimum height=1.3cm}}
\tikzset{innersnode/.style={minimum width=.4cm,minimum height=.3cm}}
\tikzset{textbox/.style={draw,text width=2.6cm,text centered,minimum height=.6cm}}

\usepackage[linktoc=all]{hyperref} 
\usepackage[open,numbered]{bookmark} 
\usepackage{breakurl} 
\usepackage{cite} 
\usepackage{hyperlinks} 
\usepackage{eprint} 

\numberwithin{figure}{section}

\newcommand{\auth}{\text{auth}}

\newcommand{\qkd}{\text{qkd}}
\newcommand{\otp}{\text{otp}}
\newcommand{\corr}{\text{cor}}
\newcommand{\secr}{\text{sec}}

\newcommand{\pabort}{p_{\text{abort}}}

\newcommand{\distinguishOp}{p_{\operatorname{distinguish}}}
\newcommand{\sdistinguish}[1]{\ensuremath{\distinguishOp(#1)}}
\newcommand{\ldistinguish}[1]{\ensuremath{\distinguishOp\left(#1\right)}}
\newcommand{\distinguish}[1]{\if@display\ldistinguish{#1}\else\sdistinguish{#1}\fi}

\newcommand{\ddistinguishOp}{p^{\fD}_{\operatorname{distinguish}}}
\newcommand{\sddistinguish}[1]{\ensuremath{\ddistinguishOp(#1)}}
\newcommand{\lddistinguish}[1]{\ensuremath{\ddistinguishOp\left(#1\right)}}
\newcommand{\ddistinguish}[1]{\if@display\lddistinguish{#1}\else\sddistinguish{#1}\fi}

\title{Cryptographic security of quantum key distribution}

\author{Christopher Portmann\email{chportma@phys.ethz.ch}\ }
\author{Renato Renner\email{renner@phys.ethz.ch}}

\affil{Institute for Theoretical Physics, ETH Zurich, 8093 Zurich, Switzerland.}

\date{\today}

\begin{document}

\pdfbookmark[1]{Title page}{titlepage}

\maketitle

\pagenumbering{roman}
\thispagestyle{empty}

\begin{abstract}
  This work is intended as an introduction to cryptographic security
  and a motivation for the widely used Quantum Key Distribution (QKD)
  security definition. We review the notion of security necessary for
  a protocol to be usable in a larger cryptographic context, i.e., for
  it to remain secure when composed with other secure protocols. We
  then derive the corresponding security criterion for QKD. We provide
  several examples of QKD composed in sequence and parallel with
  different cryptographic schemes to illustrate how the error of a
  composed protocol is the sum of the errors of the individual
  protocols. We also discuss the operational interpretations of the
  distance metric used to quantify these errors.
\end{abstract}



\clearpage
\phantomsection
\pdfbookmark[1]{\contentsname}{sec:toc}
\tableofcontents
\newpage
\phantomsection
\pdfbookmark[1]{\listfigurename}{sec:lof}
\listoffigures
\clearpage

\pagenumbering{arabic}

\section{Introduction}
\label{sec:intro}

\subsection{Background}
\label{sec:intro.background}

The first Quantum Key Distribution (QKD) protocols were proposed
independently by Bennett and Brassard~\cite{BB84} in 1984 \---
inspired by early work on quantum money by Wiesner~\cite{Wie83} \---
and by Ekert~\cite{Eke91} in 1991. The original papers discussed
security in the presence of an eavesdropper that could perform only
limited operations on the quantum channel. The first security proofs
that considered an unbounded adversary were given more than a decade
later~\cite{May96,BBBMR00,SP00,May01,BBBMR06}. Another decade after
the first such proof, K\"onig et al.~\cite{KRBM07} showed that the
security criterion used was insufficient: even though it guarantees
that an eavesdropper cannot guess the key, this only holds if the key
is never used. If part of the key is revealed to the eavesdropper \---
for example, by using it to encrypt a message known to her \--- the
rest becomes insecure. A new security criterion for QKD was
introduced, along with a new proof of security for
BB84~\cite{RK05,BHLMO05,Ren05}. It was argued that $\rho_{KE}$, the
joint state of the final key ($K$) and quantum information gathered by
an eavesdropper ($E$), must be close to an ideal key, $\tau_K$, that
is perfectly uniform and independent from the adversary's information
$\rho_E$: \begin{equation} \label{eq:d} (1-\pabort) D \left(
    \rho_{KE},\tau_K \tensor \rho_E \right) \leq \eps \
  ,\end{equation} where $\pabort$ is the probability that the protocol
aborts,\footnote{In \cite{Ren05}, \eqnref{eq:d} was introduced with a
  subnormalized state $\rho_{KE}$, with $\trace{\rho_{KE}}=1-\pabort$,
  instead of explicitly writing the factor $(1-\pabort)$. The two
  formulations are however mathematically equivalent.}
$D(\cdot,\cdot)$ is the trace distance\footnote{This metric is defined
  and discussed in detail in \appendixref{app:op}.} and $\eps \in
[0,1]$ is a (small) real number.\footnote{Another formulation of this
  security criterion, $(1-\pabort) \min_{\sigma_E} D \left(
    \rho_{KE},\tau_K \tensor \sigma_E \right) \leq \eps$, has also
  been proposed in the literature. We discuss this alternative
  in \appendixref{app:alternative}.}

The type of security flaw suffered by the early QKD security criteria
is well known in classical cryptography. It was addressed
independently by Pfitzmann and Waidner~\cite{PW00,PW01,BPW04,BPW07}
and Canetti~\cite{Can01,CDPW07,Can13}, who introduced general
frameworks to define cryptographic security, which they dubbed
\emph{reactive simulatability} and \emph{universal composability},
respectively. These frameworks were adapted to quantum cryptography by
Ben-Or and Mayers~\cite{BM04} and Unruh~\cite{Unr04,Unr10}, and the
security of QKD was discussed within these frameworks by Ben-Or et
al.~\cite{BHLMO05} and M\"uller-Quade and
Renner~\cite{MR09}. Recently, Maurer and Renner~\cite{MR11} introduced
a new cryptographic security framework, \emph{Abstract Cryptography}
(AC), which both simplifies and generalizes previous frameworks, and
applies equally to the classical and quantum settings.

The core idea of all these security frameworks is to prove that the
functionality constructed by the real protocol is indistinguishable
from the functionality of an ideal resource that fulfills in a perfect
way whatever task is expected of the cryptographic protocol \--- in
the case of QKD, this ideal resource provides the two players with a
perfect key, unknown to the adversary. If this ideal system is
indistinguishable from the real one, then one can be substituted for
the other in any context. Players who run a QKD protocol can thus
treat the resulting key as if it were perfect, which trivially implies
that it can be safely used and composed arbitrarily with other
(secure) protocols.

\subsection{Contributions}
\label{sec:intro.contributions}

Since the security criterion of \eqnref{eq:d} provides the
aforementioned compositional guarantees, it is widely used in the QKD
literature and generally introduced as the correct security definition
(see, e.g., the QKD review paper \cite{SBCDLP09}). A more detailed
explanation as to why this is the case is however usually omitted due
to the highly involved security frameworks. Even the technical
works~\cite{RK05,BHLMO05,Ren05,MR09} that introduced and discuss
\eqnref{eq:d} do not provide a self contained justification of this
security notion. The current paper aims to fill in this gap by
revisiting the security of QKD using the AC framework.

Our goals are twofold. Firstly, we provide an introduction to
cryptographic security. We do not discuss the AC framework in detail,
but explain the main ideas underlying cryptographic security and
illustrate protocol composition with many examples. Secondly, we use
this framework to show how \eqnref{eq:d} can be \emph{derived}. We
also provide in \appendixref{app:op} an extensive discussion of the
interpretation and operational meaning of the trace distance used in
\eqnref{eq:d}.

\subsection{Abstract cryptography}
\label{sec:intro.ac}

The traditional approach to defining security~\cite{PW00,PW01,Can01}
can be seen as bottom\-/up. One first defines (at a low level) a
computational model (e.g., a Turing machine). One then defines how the
machines communicate (e.g., by writing to and reading from shared
tapes) and some form of scheduling. Next, one can define notions of
complexity and efficiency. Finally, the security of a cryptosystem can
be defined.

Abstract cryptography (AC) on the other hand uses a top\-/down
approach. In order to state definitions and develop a theory, one
starts from the other end, the highest possible level of abstraction
\--- the composition of abstract systems \--- and proceeds downwards,
introducing in each new lower level only the minimal necessary
specializations. The (in)distinguishability of the real and ideal
systems is defined as a metric on abstract systems, which, at a lower
level, can be chosen to capture the distinguishing power of a
computationally bounded or unbounded environment. The abstract systems
are instantiated with, e.g., a synchronous or asynchronous network of
(abstract) machines. These machines can be instantiated with either
classical or quantum processes.

One may give the analogous example of group theory, which is used to
describe matrix multiplication. In the bottom\-/up approach, one would
start explaining how matrices are multiplied, and then based on this
find properties of the matrix multiplication. In contrast to this, the
top\-/down approach would correspond to first defining the (abstract)
multiplication group and prove theorems already on this level. The
matrix multiplication would then be introduced as a special case of
the multiplicative group. This simplifies greatly the framework by
avoiding unnecessary specificities from lower levels, and does not
hard code a computation or communication model (e.g., classical or
quantum, synchronous or asynchronous) in the security
framework. 

\subsection{Structure of this paper}
\label{sec:intro.structure}

In \secref{sec:ac} we start by introducing a simplified version of the
AC framework~\cite{MR11}, which is sufficient for the specific
adversarial structure relevant to QKD, namely honest Alice and Bob,
and dishonest Eve. In \secref{sec:qkd} we model the real and ideal
systems of a generic QKD protocol, and plug it in the AC security
framework, obtaining a security definition for QKD. In
\secref{sec:security} we then prove that this can be reduced to
\eqnref{eq:d}.\footnote{More precisely, the security definition of QKD
  is reduced to a combination of two criteria, secrecy (captured by
  \eqnref{eq:d}) and correctness.} In \secref{sec:ex} we illustrate
the composition of protocols in AC with examples of QKD composed in
various settings. We emphasize that this section does not prove that
the QKD security criterion is composable \--- the proof of this
follows from the generic proof that the AC framework is
composable~\cite{MR11} \--- but illustrates how the security of
composed protocols results from the security of individual protocols
and the triangle inequality.  Further examples can be found
in \appendixref{app:ex.auth}, where we model the security of
authentication and compose it with QKD, resulting in a key expansion
protocol. We also provide a substantial review of the trace distance
and its operational interpretations in \appendixref{app:op}. In
particular, we prove that it corresponds to the probability a
distinguisher has of correctly guessing whether it is interacting with
the real or ideal QKD system \--- the measure used in the AC framework
\--- and discuss how to interpet this. An overview of the other
appendices is given on \pref{app}.

\section{Cryptographic security}
\label{sec:ac}

A central element in modeling security is that of \emph{resources}
\--- resources used in protocols and resources constructed by
protocols. For example, a QKD protocol constructs a functionality
which shares a secret key between two players. This functionality is a
resource, which can be used by other protocols, e.g., to encrypt a
message. To construct this secret key resource, a QKD protocol
typically uses two other resources, an authentic classical
channel\footnote{An authentic channel guarantees that the message
  received comes from the legitimate sender, and has not been tampered
  with or generated by an adversary.} and an insecure quantum
channel. The authentic channel resource can in turn be constructed
from an insecure channel resource and a
password\footnoteremember{fn:password}{A short key $K$ with
  min\-/entropy $H_{\infty}(K) = \Omega(\log |\cK|)$ is sufficient for
  authentication~\cite{RW03}, where $\cK$ is the key alphabet \---
  i.e., having $H_{\infty}(K)$ linear in the key length $\log |\cK|$ is
  sufficient. We refer to such a weak key as a
  \emph{password}.}~\cite{RW03}. Composing the authentication protocol
with the QKD protocol results in a scheme which constructs a secret
key from a password and insecure channels. Part of the resulting
secret key can be used in further rounds of authentication and QKD to
produce even more secret key. This is illustrated in
\figref{fig:construction}. 

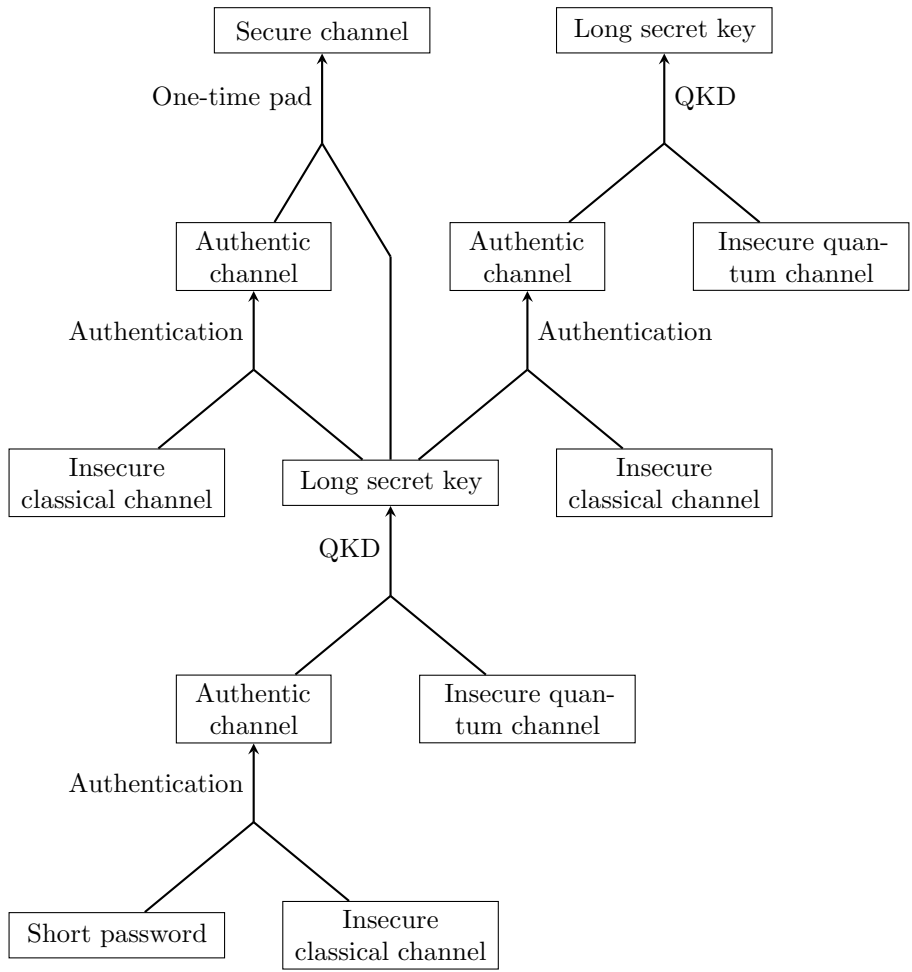
\begin{figure}[htbp]
\begin{centering}

\begin{tikzpicture}\small

\def\d{1.8}

\node[textbox] (a1) at (0*\d,0) {Short password};
\node[textbox] (a3) at (2*\d,0) {Insecure classical channel};
\node[draw,text width=1.8cm,text centered] (b2) at (1*\d,3) {Authentic channel};
\node (n2) at (1*\d,1.5) {};
\node[textbox] (b4) at (3*\d,3) {Insecure quantum channel};
\node[textbox] (c3) at (2*\d,6) {Long secret key};
\node (o3) at (2*\d,4.5) {};
\node[textbox] (c1) at (0*\d,6) {Insecure classical channel};
\node[textbox] (c5) at (4*\d,6) {Insecure classical channel};
\node[draw,text width=1.8cm,text centered] (d2) at (1*\d,9) {Authentic channel};
\node[draw,text width=1.8cm,text centered] (d4) at (3*\d,9) {Authentic channel};
\node (p2) at (1*\d,7.5) {};
\node (p4) at (3*\d,7.5) {};
\node[textbox] (e25) at (1.5*\d,12) {Secure channel};
\node (q3) at (2*\d,9) {};
\node (r25) at (1.5*\d,10.5) {};
\node[textbox] (d6) at (5*\d,9) {Insecure quantum channel};
\node[textbox] (e5) at (4*\d,12) {Long secret key};
\node (r5) at (4*\d,10.5) {};

\draw[thick] (a1) to (n2.center);
\draw[thick] (a3) to (n2.center);
\draw[sArrow] (n2.center) to node[auto] {Authentication} (b2);
\draw[thick] (b2) to (o3.center);
\draw[thick] (b4) to (o3.center);
\draw[sArrow] (o3.center) to node[auto] {QKD} (c3);
\draw[thick] (c1) to (p2.center);
\draw[thick] (c3) to (p2.center);
\draw[sArrow] (p2.center) to node[auto] {Authentication} (d2);
\draw[thick] (c3) to (p4.center);
\draw[thick] (c5) to (p4.center);
\draw[sArrow] (p4.center) to node[auto,swap] {Authentication} (d4);
\draw[thick] (c3) to (q3.center);
\draw[thick] (d2) to (r25.center);
\draw[thick] (q3.center) to (r25.center);
\draw[sArrow] (r25.center) to node[auto] {One-time pad} (e25);
\draw[thick] (d4) to (r5.center);
\draw[thick] (d6) to (r5.center);
\draw[sArrow] (r5.center) to node[auto,swap] {QKD} (e5);

\end{tikzpicture}

\end{centering}
\caption[Recursive construction of
resources]{\label{fig:construction}A cryptographic protocol uses
  (weak) resources to construct other (stronger) resources. These
  resources are depicted in the boxes, and the arrows are
  protocols. Each box is a one-time-use resource, so the same resource
  appears in multiple boxes if different protocols require it. The
  long secret key resource in the center of the figure is split in
  three shorter keys, and each protocol uses one of these keys.}
\end{figure}

For any cryptographic task one can define an ideal resource which
fullfils this task in a perfect way. A protocol is then considered
secure if the real resource actually constructed is indistinguishable
from a system running the ideal
resource.\footnoteremember{fn:relative}{Note that we use the notions
  \emph{real} and \emph{ideal} in a relative sense: the ideal resource
  that we wish to construct with one protocol might be considered a
  real resource available to another protocol.} This notion of
security based on distinguishing real and ideal systems is explained
informally in \secref{sec:ac.view}. It is then illustrated with the
one-time pad\footnote{The \emph{one-time pad} is an encryption scheme
  that XORs every bit of a message $x$ with a bit of a key $k$, and
  transmits the resulting ciphertext $y = x \xor k$ to the
  receiver. The message, which can be decrypted by performing the
  reverse operation $x = y \xor k$, is hidden from any player who
  intercepts the ciphertext $y$ but has no knowledge of the key $k$.}
in \secref{sec:ac.otp}. In \secref{sec:ac.definition} we give a formal
security definition in the Abstract Cryptography (AC) framework for
the special case of three party protocols with honest Alice and Bob,
and dishonest Eve. Finally, in \secref{sec:ac.interpretation} we
discuss how the metric used to quantify the (in)distinguishability
between the real and ideal settings should be
interpreted.

\subsection{Real-world ideal-world paradigm}
\label{sec:ac.view}

Cryptography aims at providing security guarantees in the presence of
an \emph{adversary}. And traditionally, security has been defined with
respect to the information gathered by this adversary \--- but, as we
shall see, this can be insufficient to achieve the desired security
guarantees.
A typical example of this is the security criterion used in early
papers on QKD, e.g., \cite{May96,BBBMR00,SP00,May01}. Let $K$ be the
secret key produced by a run of a QKD protocol, and $Y$ be a random
variable obtained by an adversary attacking the scheme and measuring
her quantum system $E$. It can be argued that the key is unknown to
the adversary if she gains only negligible information about it, i.e.,
if for all attacks and measurements of the resulting quantum system,
\begin{equation}
  \label{eq:localqkd}
  I(K;Y) \approx 0 \ ,
\end{equation} where $I(K;Y)$ is the mutual information\footnote{This
  information measure, the maximum mutual information over all
  measurements of the quantum system, is called
  \emph{accessible information}.} between $K$ and $Y$.

However, even if a key obtained from a protocol satisfying
\eqnref{eq:localqkd} is used in a perfectly secure encryption scheme
like the one-time pad, it can leak information about the
message. K\"onig et al.~\cite{KRBM07} give such an example: they find
a quantum state $\rho_{KE}$ which satisfies \eqnref{eq:localqkd}, but
which cannot be used to encrypt a message partly known to an
adversary. They show that if the key is split in two, $K=K_1K_2$, and
the adversary delays measuring her system $E$ until the first part,
$K_1$, is revealed to her \--- e.g., because a known message was
encrypted by the one-time pad with $K_1$ \--- she can obtain
information about the rest of the key. More precisely, they prove that
for this state $\rho_{K_1K_2E}$, \[I(K_2;Y') \gg 0 \ ,\] where $Y'$ is
a random variable obtained by a measurement of the joint state
$\rho_{K_1E}$ consisting of the partial key $K_1$ and the quantum
information $E$ gathered during the QKD protocol.\footnote{This
  phenomenon is called \emph{information locking}
  \cite{DHLST04,Win14}.} Even though the key obtained from the QKD
protocol is approximately uniform and independent from the adversary's
information $Y$, it is unusable in a cryptographic context, and
another approach than the adversarial viewpoint is necessary for
defining cryptographic security.

This new approach was proposed independently by Canetti~\cite{Can01}
and Pfitzmann and Waidner~\cite{PW00,PW01} for classical
cryptography. The gist of their global security paradigm lies in
measuring how well some \emph{real} protocol can be distinguished from
some \emph{ideal} system that fullfils the task in an ideal way, and
is often referred to as the ``real\-/world ideal\-/world''
paradigm.\footnote{As already noted in \footnoteref{fn:relative}, we
  use the notions \emph{real} and \emph{ideal} in a relative sense.}

To do this, the notion of an adversary is dropped in favor of a
\emph{distinguisher}. Apart from having the capabilities of the
adversary, this distinguisher also encompasses any protocol that is
run before, after, and during the protocol being analyzed. The role of
the distinguisher is to capture ``the rest of the world'', everything
that exists outside of the honest players and the resources they
share. A distinguisher is defined as an entity that can choose the
inputs of the honest players (that might come from a previously run
protocol), receives their outputs (that could be used in a subsequent
protocol), and simultaneously fullfils the role of the adversary,
possibly eavesdropping on the communication channels and tampering
with messages. This distinguisher is given a black box access to
either the real or an ideal system, and must decide with which of the
two it is interacting. A protocol is then considered secure if the
real system constructed is indistinguishable from the ideal one. This
is illustrated in \figref{fig:distinguisher}.

\begin{figure}[htb]
\begin{centering}

\begin{tikzpicture}\small

\def\v{.6}
\def\u{2.27} 
\def\w{-1.9} 
\def\b{7cm}

\node[lrnode] (r1) at (0,\v) {};
\node[lrnode] (r2) at (0,0) {};
\node[lrnode] (r3) at (0,-\v) {};
\node[llrnode] (rr1) at (-\v,0) {};
\node[llrnode] (rr2) at (0,0) {};
\node[llrnode] (rr3) at (\v,0) {};
\node[largeResource] (R) at (0,0) {\footnotesize Real system};

\draw[thick] (-1.618-1.15,1) -- ++(.75,0) -- ++(0,-2.4)  --
++(1.618*2+.8,0)  -- ++(0,2.4) -- ++(.75,0) -- ++(0,-3.4) --
++(-1.618*2-2.3,0) -- cycle;

\node at (0,\w) {\footnotesize Distinguisher};
\node[tlrnode] (dd1) at (-\v,\w) {}; 
\node[tlrnode] (dd2) at (0,\w) {}; 
\node[tlrnode] (dd3) at (\v,\w) {};
\node[inner sep=0] (d1) at (-\u,\v) {};
\node[inner sep=0] (d2) at (-\u,0) {};
\node[inner sep=0] (d3) at (-\u,-\v) {};
\node[inner sep=0] (d4) at (\u,\v) {};
\node[inner sep=0] (d5) at (\u,0) {};
\node[inner sep=0] (d6) at (\u,-\v) {};
\node[inner sep=0] (d0) at (0,\w-.5-.7) {};

\draw[sArrow] (rr1) to (dd1);
\draw[sArrow] (dd2) to (rr2);
\draw[sArrow] (rr3) to (dd3);
\draw[sArrow] (d1) to (r1);
\draw[sArrow] (r2) to (d2);
\draw[sArrow] (d3) to (r3);
\draw[sArrow] (r1) to (d4);
\draw[sArrow] (d5) to (r2);
\draw[sArrow] (r3) to (d6);
\draw[sArrow] (dd2) to node[auto,pos=.6] {\footnotesize $0,1$} (d0);


\node[xshift=\b,lrnode] (i1) at (0,\v) {};
\node[xshift=\b,lrnode] (i2) at (0,0) {};
\node[xshift=\b,lrnode] (i3) at (0,-\v) {};
\node[xshift=\b,llrnode] (ii1) at (-\v,0) {};
\node[xshift=\b,llrnode] (ii2) at (0,0) {};
\node[xshift=\b,llrnode] (ii3) at (\v,0) {};
\node[xshift=\b,largeResource] (I) at (0,0) {\footnotesize Ideal system};

\draw[xshift=\b,thick] (-1.618-1.15,1) -- ++(.75,0) -- ++(0,-2.4)  --
++(1.618*2+.8,0)  -- ++(0,2.4) -- ++(.75,0) -- ++(0,-3.4) --
++(-1.618*2-2.3,0) -- cycle;

\node[xshift=\b] at (0,\w) {\footnotesize Distinguisher};
\node[xshift=\b,tlrnode] (ee1) at (-\v,\w) {}; 
\node[xshift=\b,tlrnode] (ee2) at (0,\w) {}; 
\node[xshift=\b,tlrnode] (ee3) at (\v,\w) {};
\node[xshift=\b,inner sep=0] (e1) at (-\u,\v) {};
\node[xshift=\b,inner sep=0] (e2) at (-\u,0) {};
\node[xshift=\b,inner sep=0] (e3) at (-\u,-\v) {};
\node[xshift=\b,inner sep=0] (e4) at (\u,\v) {};
\node[xshift=\b,inner sep=0] (e5) at (\u,0) {};
\node[xshift=\b,inner sep=0] (e6) at (\u,-\v) {};
\node[xshift=\b,inner sep=0] (e0) at (0,\w-.5-.7) {};

\draw[sArrow] (ii1) to (ee1);
\draw[sArrow] (ee2) to (ii2);
\draw[sArrow] (ii3) to (ee3);
\draw[sArrow] (e1) to (i1);
\draw[sArrow] (i2) to (e2);
\draw[sArrow] (e3) to (i3);
\draw[sArrow] (i1) to (e4);
\draw[sArrow] (e5) to (i2);
\draw[sArrow] (i3) to (e6);
\draw[sArrow] (ee2) to node[auto,pos=.6] {\footnotesize $0,1$} (e0);

\end{tikzpicture}

\end{centering}
\caption[Distinguishing systems]{\label{fig:distinguisher}A
  distinguisher has a complete description of two systems, and is
  given a black\-/box access to one of the two. After interacting with
  the system, it must guess which one it is holding.}
\end{figure}
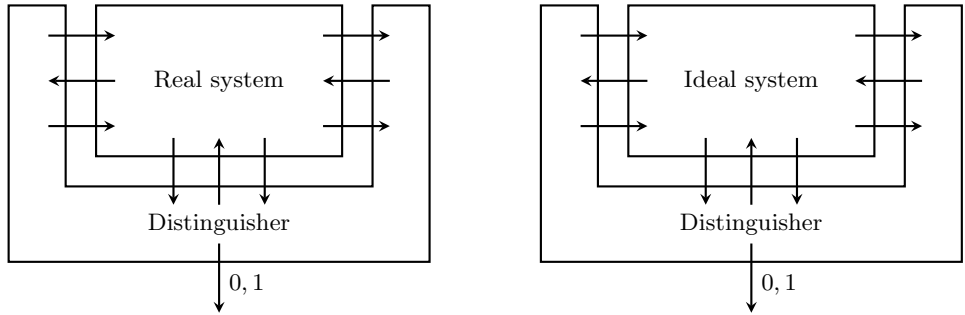

In the case of QKD, this means that the distinguisher does not only
obtain the system $E$ of the eavesdropper, but also receives the final
key $K$ generated by Alice and Bob. In the real world, this key is
potentially correlated to $E$, and in an ideal system, $K$ is
uniformly random and independent from $E$. The distinguisher can then
run the attack of K\"onig et al.~\cite{KRBM07} to distinguish between
the real and ideal systems: if $Y'$, the result of the measurement of
$K_1$ and $E$ is correlated to $K_2$, it knows that it was given the
real system, otherwise it must have the ideal one. This specific
attack is illustrated in more detail in \secref{sec:ex.leak}.

\subsection{Example: one-time pad}
\label{sec:ac.otp}

In this section, we illustrate with the one-time pad how security is
defined in the real\-/world ideal\-/world paradigm. The one-time pad
protocol uses a secret key $k$ to encrypt a message $x$ as $y
\coloneqq x \xor k$. The ciphertext $y$ is then sent on an authentic
channel to the receiver, who decrypts it, obtaining $x = y \xor
k$. $y$ is however also leaked to the adversary that is eavesdropping
on the authentic channel. This is depicted in \figref{fig:otp.real}.

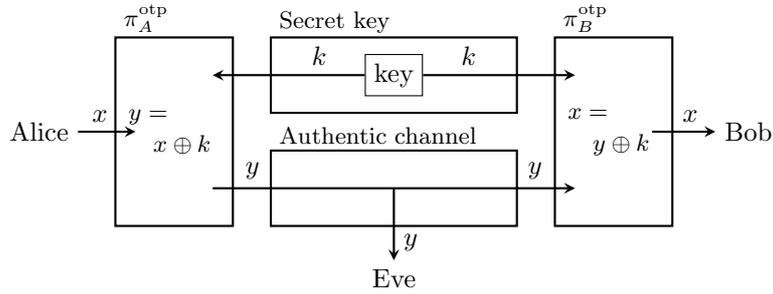
\begin{figure}[htb]
\begin{centering}

\begin{tikzpicture}\small

\def\t{4.663} 
\def\u{2.89} 
\def\v{.75}

\node[pnode] (a1) at (-\u,\v) {};
\node[pnode] (a2) at (-\u,0) {};
\node[pnode] (a3) at (-\u,-\v) {};
\node[protocol,text width=1.2cm] (a) at (-\u,0) {\footnotesize $y
  =$\\$\quad x \xor k$};
\node[yshift=-2,above right] at (a.north west) {\footnotesize
  $\pi^{\otp}_A$};
\node (alice) at (-\t,0) {Alice};

\node[pnode] (b1) at (\u,\v) {};
\node[pnode] (b2) at (\u,0) {};
\node[pnode] (b3) at (\u,-\v) {};
\node[protocol,text width=1.2cm] (b) at (\u,0) {\footnotesize $x
  =$\\$\quad y \xor k$};
\node[yshift=-2,above right] at (b.north west) {\footnotesize $\pi^{\otp}_B$};
\node (bob) at (\t,0) {Bob};

\node[thinResource] (keyBox) at (0,\v) {};
\node[draw] (key) at (0,\v) {key};
\node[yshift=-2,above right] at (keyBox.north west) {\footnotesize Secret key};
\node[thinResource] (channel) at (0,-\v) {};
\node[yshift=-1.5,above right] at (channel.north west) {\footnotesize
  Authentic channel};
\node (eve) at (0,-1.95) {Eve};
\node (junc) at (eve |- a3) {};

\draw[sArrow] (key) to node[auto,swap,pos=.3] {$k$} (a1);
\draw[sArrow] (key) to node[auto,pos=.3] {$k$} (b1);

\draw[sArrow] (alice) to node[auto,pos=.38] {$x$} (a2);
\draw[sArrow] (b2) to node[auto,pos=.6] {$x$} (bob);

\draw[sArrow] (a3) to node[pos=.11,auto] {$y$} node[pos=.89,auto] {$y$} (b3);
\draw[sArrow] (junc.center) to node[pos=.75,auto] {$y$} (eve);

\end{tikzpicture}

\end{centering}
\caption[Real one-time pad system]{\label{fig:otp.real}The real
  one-time pad system \--- Alice has access to the left interface, Bob
  to the right interface and Eve to the lower interface \--- consists
  of the one-time pad protocol $(\pi^{\otp}_A,\pi^{\otp}_B)$, and the
  secret key and authentic channel resources. The combination of these
  resources and protocol constructs a system that takes a message $x$
  at Alice's interface, outputs a ciphertext $y$ at Eve's interface
  and the original message $x$ at Bob's interface.}
\end{figure}

The one-time pad protocol thus uses two resources, a secret key and an
authentic channel. The resource we wish to construct with this
encryption scheme is a secure channel: a resource which transmits a
message $x$ from the sender to the receiver, and leaks only
information about the message size $|x|$ at the adversary's interface,
but not the contents of the message. This is illustrated in
\figref{fig:otp.ideal.resource}.

\begin{figure}[htb]
\begin{centering}

\begin{tikzpicture}\small

\node[largeResource] (keyBox) at (0,0) {};
\node (alice) at (-2.6,0) {Alice};
\node (bob) at (2.6,0) {Bob};
\node (eve) at (0,-1.7) {Eve};
\node (ajunc) at (eve.north |- alice) {};

\draw[sArrow] (alice) to node[pos=.06,auto] {$x$} node[pos=.94,auto] {$x$} (bob);
\draw[dotted,sArrow] (ajunc.center) to node[pos=.85,auto,swap] {$|x|$} (eve);

\end{tikzpicture}

\end{centering}
\caption[Secure channel]{\label{fig:otp.ideal.resource}A secure
  channel from Alice to Bob leaks only the message size at Eve's interface.}
\end{figure}
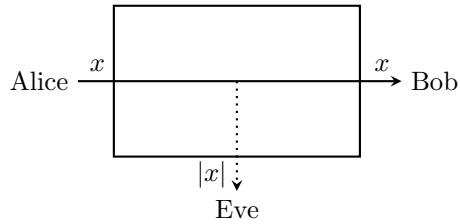

Since an ideal resource ``magically'' solves the cryptographic task
considered, e.g., by producing perfect secret keys or transmitting a
message directly from Alice to Bob, the adversary's interface of the
ideal resource is usually quite different from her interface of the
real system, which gives her access to the resources used. For the
one-time pad, the real system from \figref{fig:otp.real} outputs a
string $y$ at Eve's interface, but the ideal secure channel from
\figref{fig:otp.ideal.resource} outputs an integer, $|x|$. To make the
comparison between real and ideal systems possible, we define the
ideal system to consist of the ideal resource as well as a
\emph{simulator} plugged into the adversary's interface of the ideal
resource, that recreates the communication occurring in the real
system. For the one-time pad, this simulator must generate a
ciphertext $y$ given the message length $|x|$. This is simply done by
generating a random string of the appropriate length, as depicted in
\figref{fig:otp.ideal}. Note that putting such a simulator between the
ideal resource and the adversary can only weaken her, since any
operation performed by the simulator could equivalently be performed
by an adversary connected directly to the interface of the ideal
resource.

\begin{figure}[htb]
\begin{centering}

\begin{tikzpicture}\small

\def\t{2.368} 
\def\u{-.75}
\def\v{.75}
\def\w{-1.75} 

\node[thinResource] (channel) at (0,\v) {};
\node[yshift=-1.5,above right] at (channel.north west) {\footnotesize
  Secure channel};
\node (alice) at (-2.6,\v) {Alice};
\node (bob) at (2.6,\v) {Bob};

\node[protocolLong] (sim) at (0,\u) {};
\node[xshift=1.5,below left] at (sim.north west) {$\sigma^{\otp}_E$};
\node[draw] (rand) at (0,\u) {\footnotesize Random string};

\draw[sArrow] (alice) to node[pos=.06,auto] {$x$}
node[pos=.93,auto] {$x$} (bob);
\draw[sArrow,dotted] (0,\v) to node[pos=.6,auto] {$|x|$} (rand);

\node (eve) at (0,\w-.2) {Eve};
\draw[sArrow] (rand) to node[pos=.65,auto] {$y$} (eve);

\end{tikzpicture}

\end{centering}
\caption[Ideal one-time pad system]{\label{fig:otp.ideal}The ideal
  one-time pad system \--- Alice has access to the left interface, Bob
  to the right interface and Eve to the lower interface \--- consists
  of the ideal secure channel and a simulator $\sigma^{\otp}_E$ that
  generates a random string $y$ of length $|x|$.}
\end{figure}
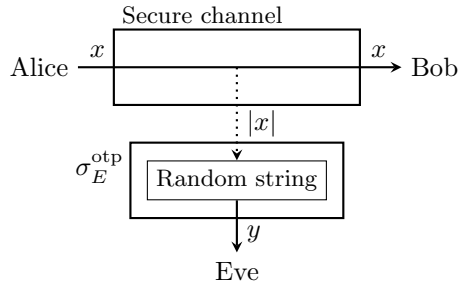

To prove that the one-time pad constructs a secure channel from an
authentic channel and a secret key, we view the real and ideal
one-time pad systems of \figref{fig:otp.real} and
\figref{fig:otp.ideal} as black boxes, and need to show that no
distinguisher can tell with which of the two it has been
connected. For both black boxes, if the distinguisher inputs $x$ at Alice's
interface, the same string $x$ is output at Bob's interface and a
uniformly random string of length $|x|$ is output at Eve's
interface. The two systems are thus completely indistinguishable \--- if
the distinguisher were to take a guess, it would be right with
probability exactly $1/2$ \--- and we say that the one-time pad has
perfect security.

If two systems are indistinguishable, they can be used interchangeably
in any setting. For example, let some protocol $\pi'$ be proven secure
if Alice and Bob are connected by a secure channel. Since the one-time
pad constructs such a channel, it can be used in lieu of the secure
channel, and composed with $\pi'$. Or equivalently, the
contrapositive: if composing the one-time pad and $\pi'$ were to leak
some vital information, which would not happen with a secure channel,
a distinguisher that is either given the real or ideal system could
run $\pi'$ internally and check whether this leak occurs to know with
which of the two it is interacting.





\subsection{General security definition}
\label{sec:ac.definition}

The previous sections introduced the concepts of resources, protocols
and simulator in an informal manner. In the AC framework these
elements are defined in an abstract way. For example, a resource is an
abstract system that is shared between all players and provides each
one with an interface that allows in- and outputs. AC does not define
the internal workings of a resource. It postulates axioms that these
abstract systems must fulfill \--- e.g. there must exist a metric and
a parallel composition operator on the space of resources \--- and is
valid for any instantiation which respects these axioms. In the group
theory analogy introduced in \secref{sec:intro.ac}, these axioms
correspond to the group axioms (closure, associativity, identity and
invertibility). Any set and operation that respects these group axioms
is an instantiation of a group, and any theorem proven for groups
applies to this instantiation.

Thus, AC defines cryptographic security for abstract systems which
fulfill certain basic properties. In the following we briefly sketch
what these are. Note that examples \--- such as the model of the
one-time pad given in Figures~\ref{fig:otp.real} and
\ref{fig:otp.ideal} \--- necessarily assume some instantiation of the
abstract systems. Since we consider only simple examples in this work,
we do not provide formal generic definitions of these lower levels,
and refer to the discussions in \cite{MR11,Mau12,DFPR14} on how this
can be modeled.

\paragraph{Resource.} An \emph{$\cI$-resource} is an (abstract) system
with interfaces specified by a set $\cI$ (e.g., $\cI =
\{A,B,E\}$). Each interface $i \in \cI$ is accessible to a user $i$
and provides her or him with certain controls (the possibility of
reading outputs and providing inputs). Resources are equipped with a
parallel composition operator, $\|$ , that maps two resources to
another resource.

\paragraph{Converter.} To transform one resource into another, we use
\emph{converters}. These are (abstract) systems with two interfaces,
an \emph{inside} interface and an \emph{outside} interface. The inside
interface connects to an interface of a resource, and the outside
interface becomes the new interface of the constructed resource. We
write either $\alpha_i \aR$ or $\aR\alpha_i$ to denote the new
resource with the converter $\alpha$ connected at the interface $i$ of
$\aR$,\footnote{There is no mathematical difference between $\alpha_i
  \aR$ and $\aR\alpha_i$. It sometimes simplifies the notation to have
  the converters for some players written on the right of the resource
  and the ones for other players on the left, instead of all on the
  same side, hence the two notations.} and $\alpha\aR$ or $\aR\alpha$
for a set of converters $\alpha = \{\alpha_i\}_i$, for which it is
clear to which interface they connect.

A protocol is a set of converters (one for every honest player) and a
simulator is also a converter. Another type of converter that we need
is a \emph{filter}, which we often denote by $\sharp$ or
$\lozenge$. When placed over a dishonest player's interface, a filter
prevents access to the corresponding controls and emulates an honest
behavior.

Serial and parallel composition of converters is defined as follows:
\begin{equation} \label{eq:axioms.order} (\alpha\beta)_i \aR \coloneqq
  \alpha_i (\beta_i \aR) \qquad \text{and} \qquad (\alpha \| \beta)_i
  (\aR\|\aS) \coloneqq (\alpha_i \aR) \| (\beta_i \aS) \ . \end{equation}

\paragraph{Filtered resource.} A pair of a resource $\aR$ and a filter
$\sharp$ together specify the (reactive) behavior of a system both
when no adversary is present \--- with the filter plugged in the
adversarial interface, $\aR\sharp_E$ \--- and in the case of a
cheating player that removes the filter and has full access to her
interface of $\aR$. We call such a pair $(\aR,\sharp)$ a
\emph{filtered resource}, and usually denote it by $\aR_\sharp$.

\paragraph{Metric.} There must exist a pseudo\-/metric
$d(\cdot,\cdot)$ on the space of resources, i.e., for any three
resources $\aR,\aS,\aT$, it satisfies the following
conditions:\footnote{If additionally $d(\aR,\aS) = 0 \implies
  \aR=\aS$, then $d$ is a metric.}
\begin{align} \text{(identity)} & &
  d\left(\aR,\aR\right) & = 0 \ , \label{eq:pm.ref} \\
  \text{(symmetry)} & & d\left(\aR,\aS\right) & =
  d\left(\aS,\aR\right) \ , \label{eq:pm.sym} \\
  \text{(triangle inequality)} & & d\left(\aR,\aS\right)
  & \leq d\left(\aR,\aT\right) +
  d\left(\aT,\aS\right) \ .
  \label{eq:pm.tri} \end{align} Furthermore, this pseudo\-/metric must
be non-increasing under composition with resources and converters: for
any converter $\alpha$ and resources $\aR,\aS,\aT$, we require
\begin{equation} \label{eq:axioms.nonincrease} d(\alpha\aR,\alpha\aS)
  \leq d(\aR,\aS) \qquad \text{and} \qquad d(\aR\|\aT,\aS\|\aT) \leq
  d(\aR,\aS) \ . \end{equation}

\vspace{\baselineskip}

We are now ready to define the security of a cryptographic
protocol. We do so in the three player setting, for honest Alice and
Bob, and dishonest Eve. Thus, in the following, all resources have
three interfaces, denoted $A$, $B$ and $E$, and we only consider
honest behaviors (given by a protocol $(\pi_A,\pi_B)$) at the $A$ and
$B$\=/interfaces, but arbitrary behavior at the $E$\=/interface. We
refer to \cite{MR11} for the general case, when arbitrary players can
be dishonest.

\begin{deff}[Cryptographic security~\cite{MR11}] \label{def:security}
  Let $\pi_{AB} = (\pi_A,\pi_B)$ be a protocol and $\aR_\sharp =
  (\aR,\sharp)$ and $\aS_\lozenge = (\aS,\lozenge)$ denote two
  filtered resources.  We say that \emph{$\pi_{AB}$ constructs
    $\aS_{\lozenge}$ from $\aR_\sharp$ within $\eps$}, which we write
  $\aR_\sharp \xrightarrow{\pi,\eps} \aS_\lozenge$, if the two
  following conditions hold:
\begin{enumerate}[label=\roman*), ref=\roman*]
\item \label{eq:def.cor} We have
  \[d(\pi_{AB}\aR\sharp_E,\aS\lozenge_E) \leq \eps \ .\]
\item \label{eq:def.sec} There exists a converter $\sigma_E$ \---
  which we call simulator \--- such that
  \[  d(\pi_{AB}\aR,\cS\sigma_E) \leq \eps \ .\]
\end{enumerate}
If it is clear from the context what filtered resources $\aR_\sharp$
and $\aS_\lozenge$ are meant, we simply say that $\pi_{AB}$ is
$\eps$\=/secure.
\end{deff}

The first of these two conditions measures how close the constructed
resource is to the ideal resource in the case where no malicious
player is intervening, which we call
\emph{availability}.\footnote{This is sometimes referred to as the
  \emph{correctness} of the protocol in the cryptographic
  literature. But in QKD, correctness has another meaning \--- namely
  the probability that Alice and Bob end up with different keys when
  Eve is active. Instead, the term \emph{robustness} is traditionally
  used to denote the performance of a QKD protocol under honest
  (noisy) conditions. We refer to \secref{sec:security.rob} for a
  discussion of the relation between availability and robustness.} The
second condition captures \emph{security} in the presence of an
adversary. These two equations are illustrated in
\figref{fig:security}.

\begin{figure}[htb]
\begin{subfigure}[b]{\textwidth}
\begin{centering}

\begin{tikzpicture}[scale=.8]\small

\def\t{4.422} 
\def\u{2.9} 
\def\v{.6}
\def\w{1.8}

\def\a{5.2}
\def\b{8.3}
\def\tb{2.368} 
\def\vb{.4}
\def\ub{.75}
\def\wb{2.05}

\node[scale=.8,pnode] (a1) at (-\u,\v) {};
\node[scale=.8,pnode] (a2) at (-\u,0) {};
\node[scale=.8,pnode] (a3) at (-\u,-\v) {};
\node[scale=.8,protocol] (a) at (-\u,0) {};
\node[yshift=-2,above right] at (a.north west) {\footnotesize
  $\pi_A$};
\node (alice1) at (-\t,\v) {};
\node (alice2) at (-\t,0) {};
\node (alice3) at (-\t,-\v) {};

\node[scale=.8,pnode] (b1) at (\u,\v) {};
\node[scale=.8,pnode] (b2) at (\u,0) {};
\node[scale=.8,pnode] (b3) at (\u,-\v) {};
\node[scale=.8,protocol] (b) at (\u,0) {};
\node[yshift=-2,above right] at (b.north west) {\footnotesize $\pi_B$};
\node (bob1) at (\t,\v) {};
\node (bob2) at (\t,0) {};
\node (bob3) at (\t,-\v) {};

\node[scale=.8,lrnode] (r1) at (0,\v) {};
\node[scale=.8,lrnode] (r2) at (0,0) {};
\node[scale=.8,lrnode] (r3) at (0,-\v) {};
\node[scale=.8,llrnode] (rr1) at (-\v,0) {};
\node[scale=.8,llrnode] (rr2) at (0,0) {};
\node[scale=.8,llrnode] (rr3) at (\v,0) {};
\node[scale=.8,largeResource] (R) at (0,0) {};
\node[yshift=-2,above right] at (R.north west) {\footnotesize $\aR$};

\node (eve1) at (-\v,-\w) {};
\node (eve2) at (0,-\w) {};
\node (eve3) at (\v,-\w) {};

\node[scale=.8,filter] (fi) at (0,-1.9) {};
\node[xshift=2,below left] at (fi.north west) {\footnotesize
  $\sharp_E$};

\draw[sArrow] (alice1) to (a1);
\draw[sArrow] (a2) to (alice2);
\draw[sArrow] (alice3) to (a3);

\draw[sArrow] (a1) to (r1);
\draw[sArrow] (r2) to (a2);
\draw[sArrow] (a3) to (r3);

\draw[sArrow] (r1) to (b1);
\draw[sArrow] (b2) to (r2);
\draw[sArrow] (r3) to (b3);

\draw[sArrow] (b1) to (bob1);
\draw[sArrow] (bob2) to (b2);
\draw[sArrow] (b3) to (bob3);

\draw[sArrow] (rr1) to (eve1);
\draw[sArrow] (eve2) to (rr2);
\draw[sArrow] (rr3) to (eve3);

\node at (\a,0) {\Large $\close{\eps}$};

\node[scale=.8,lrnode] (s1) at (\b,\vb+\ub) {};
\node[scale=.8,lrnode] (s2) at (\b,\ub) {};
\node[scale=.8,lrnode] (s3) at (\b,-\vb+\ub) {};
\node[scale=.8,tlrnode] (ss1) at (\b-\v,\ub) {};
\node[scale=.8,tlrnode] (ss2) at (\b,\ub) {};
\node[scale=.8,tlrnode] (ss3) at (\b+\v,\ub) {};
\node[scale=.8,thinResource] (S) at (\b,\ub) {};
\node[yshift=-2,above right] at (S.north west) {\footnotesize $\aS$};

\node[scale=.8,tlrnode] (t1) at (\b-\v,-\ub) {};
\node[scale=.8,tlrnode] (t2) at (\b,-\ub) {};
\node[scale=.8,tlrnode] (t3) at (\b+\v,-\ub) {};
\node[scale=.8,filter] (fil) at (\b,-\ub) {};
\node[xshift=-2,below right] at (fil.north east) {\footnotesize $\lozenge_E$};

\node (cate1) at (\b-\tb,\vb+\ub) {};
\node (cate2) at (\b-\tb,\ub) {};
\node (cate3) at (\b-\tb,-\vb+\ub) {};

\node (dave1) at (\b+\tb,\vb+\ub) {};
\node (dave2) at (\b+\tb,\ub) {};
\node (dave3) at (\b+\tb,-\vb+\ub) {};

\draw[sArrow] (cate1) to (s1);
\draw[sArrow] (s2) to (cate2);
\draw[sArrow] (cate3) to (s3);

\draw[sArrow] (s1) to (dave1);
\draw[sArrow] (dave2) to (s2);
\draw[sArrow] (s3) to (dave3);

\draw[sArrow] (ss1) to (t1);
\draw[sArrow] (t2) to (ss2);
\draw[sArrow] (ss3) to (t3);

\end{tikzpicture}

\end{centering}
\caption[Availability (no
adversary)]{\label{fig:security.availability}Condition
  \eqref{eq:def.cor} from \defref{def:security}. If Eve's interfaces
  are blocked by filters emulating honest behavior, the functionality
  constructed by the protocol should be indistinguishable from the
  ideal resource.}
\end{subfigure}

\vspace{12pt}

\begin{subfigure}[b]{\textwidth}
\begin{centering}

\begin{tikzpicture}[scale=.8]\small

\def\t{4.422} 
\def\u{2.9} 
\def\v{.6}
\def\w{1.8}

\def\a{5.2}
\def\b{8.3}
\def\tb{2.368} 
\def\vb{.4}
\def\ub{.75}
\def\wb{2.05}

\node[scale=.8,pnode] (a1) at (-\u,\v) {};
\node[scale=.8,pnode] (a2) at (-\u,0) {};
\node[scale=.8,pnode] (a3) at (-\u,-\v) {};
\node[scale=.8,protocol] (a) at (-\u,0) {};
\node[yshift=-2,above right] at (a.north west) {\footnotesize
  $\pi_A$};
\node (alice1) at (-\t,\v) {};
\node (alice2) at (-\t,0) {};
\node (alice3) at (-\t,-\v) {};

\node[scale=.8,pnode] (b1) at (\u,\v) {};
\node[scale=.8,pnode] (b2) at (\u,0) {};
\node[scale=.8,pnode] (b3) at (\u,-\v) {};
\node[scale=.8,protocol] (b) at (\u,0) {};
\node[yshift=-2,above right] at (b.north west) {\footnotesize $\pi_B$};
\node (bob1) at (\t,\v) {};
\node (bob2) at (\t,0) {};
\node (bob3) at (\t,-\v) {};

\node[scale=.8,lrnode] (r1) at (0,\v) {};
\node[scale=.8,lrnode] (r2) at (0,0) {};
\node[scale=.8,lrnode] (r3) at (0,-\v) {};
\node[scale=.8,llrnode] (rr1) at (-\v,0) {};
\node[scale=.8,llrnode] (rr2) at (0,0) {};
\node[scale=.8,llrnode] (rr3) at (\v,0) {};
\node[scale=.8,largeResource] (R) at (0,0) {};
\node[yshift=-2,above right] at (R.north west) {\footnotesize $\aR$};

\node (eve1) at (-\v,-\w) {};
\node (eve2) at (0,-\w) {};
\node (eve3) at (\v,-\w) {};

\draw[sArrow] (alice1) to (a1);
\draw[sArrow] (a2) to (alice2);
\draw[sArrow] (alice3) to (a3);

\draw[sArrow] (a1) to (r1);
\draw[sArrow] (r2) to (a2);
\draw[sArrow] (a3) to (r3);

\draw[sArrow] (r1) to (b1);
\draw[sArrow] (b2) to (r2);
\draw[sArrow] (r3) to (b3);

\draw[sArrow] (b1) to (bob1);
\draw[sArrow] (bob2) to (b2);
\draw[sArrow] (b3) to (bob3);

\draw[sArrow] (rr1) to (eve1);
\draw[sArrow] (eve2) to (rr2);
\draw[sArrow] (rr3) to (eve3);

\node at (\a,0) {\Large $\close{\eps}$};

\node[scale=.8,lrnode] (s1) at (\b,\vb+\ub) {};
\node[scale=.8,lrnode] (s2) at (\b,\ub) {};
\node[scale=.8,lrnode] (s3) at (\b,-\vb+\ub) {};
\node[scale=.8,tlrnode] (ss1) at (\b-\v,\ub) {};
\node[scale=.8,tlrnode] (ss2) at (\b,\ub) {};
\node[scale=.8,tlrnode] (ss3) at (\b+\v,\ub) {};
\node[scale=.8,thinResource] (S) at (\b,\ub) {};
\node[yshift=-2,above right] at (S.north west) {\footnotesize $\aS$};

\node[scale=.8,tlrnode] (t1) at (\b-\v,-\ub) {};
\node[scale=.8,tlrnode] (t2) at (\b,-\ub) {};
\node[scale=.8,tlrnode] (t3) at (\b+\v,-\ub) {};
\node[scale=.8,protocolLong] (sim) at (\b,-\ub) {};
\node[xshift=-2,below right] at (sim.north east) {\footnotesize $\sigma_E$};

\node (cate1) at (\b-\tb,\vb+\ub) {};
\node (cate2) at (\b-\tb,\ub) {};
\node (cate3) at (\b-\tb,-\vb+\ub) {};

\node (dave1) at (\b+\tb,\vb+\ub) {};
\node (dave2) at (\b+\tb,\ub) {};
\node (dave3) at (\b+\tb,-\vb+\ub) {};

\node (finn1) at (\b-\v,-\wb) {};
\node (finn2) at (\b,-\wb) {};
\node (finn3) at (\b+\v,-\wb) {};

\draw[sArrow] (cate1) to (s1);
\draw[sArrow] (s2) to (cate2);
\draw[sArrow] (cate3) to (s3);

\draw[sArrow] (s1) to (dave1);
\draw[sArrow] (dave2) to (s2);
\draw[sArrow] (s3) to (dave3);

\draw[sArrow] (ss1) to (t1);
\draw[sArrow] (t2) to (ss2);
\draw[sArrow] (ss3) to (t3);

\draw[sArrow] (t1) to (finn1);
\draw[sArrow] (finn2) to (t2);
\draw[sArrow] (t3) to (finn3);

\end{tikzpicture}

\end{centering}
\caption[Security in the presence of an
adversary]{\label{fig:security.active}Condition \eqref{eq:def.sec}
  from \defref{def:security}. If Eve accesses her cheating interface
  of $\aR$, the resulting system must be simulatable in the ideal
  world by a converter $\sigma_E$ that only accesses Eve's interface
  of the ideal resource $\aS$.}
\end{subfigure}

\caption[Cryptographic security]{\label{fig:security}A protocol
  $(\pi_A,\pi_B)$ constructs $\aS_\lozenge$ from $\aR_\sharp$ within
  $\eps$ if the two conditions illustrated in this figure hold. The
  sequences of arrows at the interfaces between the objects represent
  (arbitrary) rounds of communication.}
\end{figure}
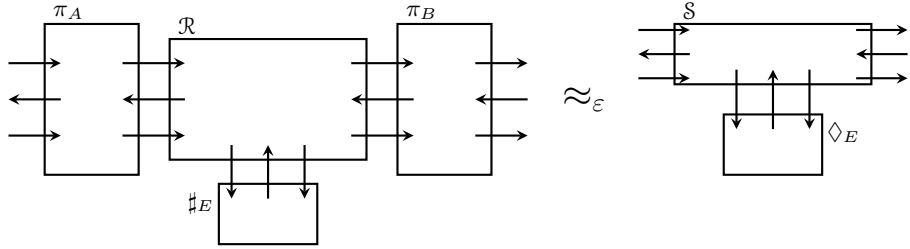
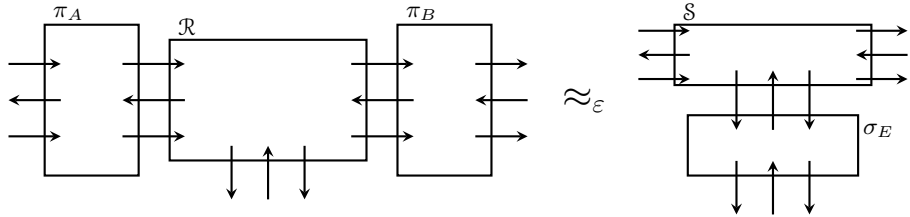

It follows from the AC framework~\cite{MR11} that if two protocols
$\pi$ and $\pi'$ are $\eps$- and $\eps'$\=/secure, the composition of
the two is $(\eps+\eps')$\=/secure. We illustrate this with several
examples in \secref{sec:ex} and \appendixref{app:ex.auth}, and sketch
a generic proof in \appendixref{app:generic}.

\subsection{The distinguishing metric}
\label{sec:ac.interpretation}

The usual pseudo\-/metric used to define security in the real\-/world
ideal\-/world paradigm is the \emph{distinguishing advantage}, defined
as follows. If a distinguisher $\fD$ can guess correctly with
probability $\ddistinguish{\aR,\aS}$ with which of two systems $\aR$
and $\aS$ it is interacting, we define its advantage as
\begin{equation}
  \label{eq:adv} 
    d^{\fD}\left(\aR,\aS\right) \coloneqq 2\ddistinguish{\aR,\aS} - 1\ .
\end{equation}
Changing the power of the distinguisher $\fD$ (e.g.,
computationally bounded or unbounded) results in different metrics and
different levels of security. In this work we are interested only in
information\-/theoretic security, we therefore consider only a
computationally unbounded distinguisher, and drop the superscript
$\fD$. We write
\[
  d(\aR,\aS) \leq \eps \qquad \text{or} \qquad \aR \close{\eps} \aS \ ,
\]
if two systems $\aR$ and $\aS$ can be distinguished with advantage at
most $\eps$, and in the following, the distance between two resources
always refers to the distinguishing advantage of an unbounded
distinguisher. A more extensive discussion of distinguishers is given
in \appendixref{app:moreAC.dist}.

Although any pseudo\-/metric which satisfies the basic axioms can be
used in \defref{def:security}, the distinguishing advantage is of
particular importance, because it has an operational definition \---
the advantage a distinguisher has in guessing whether it is
interacting with the real or ideal system. If the distinguisher
notices a difference between the two, then something in the real
setting did not behave ideally. This can be loosely interpreted as a
failure occurring. If the distinguisher can guess correctly with
probability $1$ with which system it is interacting, a failure must
occur systematically. If it can only guess correctly with probability
$1/2$, no failure occurs at all. If it can guess correctly with
probability $p$, this can be seen as a failure occurring with
probability $\eps = 2p-1$. The distinguishing advantage can thus be
interpreted as the probability that a failure occurs in the real
protocol.\footnote{A formal derivation of this interpretation is given
  in \appendixref{app:op.failure} for the trace distance \--- the
  distinguishing advantage between two quantum states.} And in any
practical implementation, the value $\eps$ can be chosen accordingly.

A bound on the security of a protocol does however not tell us how
``bad'' this failure is. For example, a key distribution protocol which
produces a perfectly uniform key, but with probability $\eps$ Alice
and Bob end up with different keys, is $\eps$\=/secure. Likewise, a
protocol which gives $1$ bit of the key to Eve with probability
$\eps$, but is perfect otherwise, and another protocol which gives the
entire key to Eve with probability $\eps$, but is perfect otherwise,
are both $\eps$\=/secure as well. One could argue that leaking the
entire key is worse than leaking one bit, which is worse than not
leaking anything but generating mismatching keys, and this should be
reflected in the level of security of the protocol. However, leaking
one bit can be as bad as leaking the entire key if only one bit of the
message is vital, and this happens to be the bit obtained by
Eve. Having mismatching keys and therefore misinterpreting a message
could have more dire consequences than leaking the message to Eve. How
bad a failure is depends on the use of the protocol, and since the
purpose of cryptographic security is to make a security statement that
is valid for all contexts, bounding the probability that a failure
occurs is the best it can do.

Since such a security bound gives no idea of the gravity of a failure
\--- a faulty QKD protocol might not only leak the current key, but
all future keys as well if the current key is used to authenticate
messages in future rounds \--- the probability $\eps$ of a failure
occurring must be chosen small enough that the accumulation of all
possible failure probabilities over a lifetime is still small
enough. For example, if an implementation of a QKD protocol produces a
key at a rate of $1$ Mbit/s with a failure per bit of $10^{-24}$, then
this protocol can be run for the age of the universe and still have an
accumulated failure strictly less than $1$.

\section{Quantum key distribution}
\label{sec:qkd}

In order to apply the general AC security definition to QKD, we need
to specify the ideal key filtered resource, which we do in
\secref{sec:qkd.ideal}. Likewise, we specify in
\secref{sec:qkd.protocol} the real QKD system consisting of the
protocol, an authentic classical channel and an insecure quantum
channel. Plugging these systems in \defref{def:security}, we obtain in
\secref{sec:qkd.security} the security criteria for QKD.


\subsection{Ideal key}
\label{sec:qkd.ideal}

The goal of a key distribution protocol is to generate a secret key
shared between two players. One can represent such a resource by a
box, one end of which is in Alice's lab, and another in Bob's. It
provides each of them with a secret key of a given length, but does
not give Eve any information about the key. This is illustrated in
\figref{fig:qkd.resource.simple}, and is the key resource we used in the
one-time pad construction (\figref{fig:otp.real}).

\begin{figure}[htb]
  \subcaptionbox[Simple ideal key]{\label{fig:qkd.resource.simple}A resource that always
    gives a key $k$ to Alice and Bob, and nothing to Eve.}[.5\textwidth][c]{
\begin{tikzpicture}\small

\def\u{0} 

\node[largeResource] (keyBox) at (0,0) {};
\node (alice) at (-2.5,\u) {Alice};
\node (bob) at (2.5,\u) {Bob};
\node (eve) at (0,-1.7) {Eve};
\node[draw] (key) at (0,0) {key};

\draw[sArrow] (key) to node[pos=.55,auto,swap] {$k$} (alice);
\draw[sArrow] (key) to node[pos=.55,auto] {$k$} (bob);

\end{tikzpicture}
}  \subcaptionbox[Ideal key with switch]{\label{fig:qkd.resource.switch}A resource that allows Eve
to decide if Alice and Bob get a key $k$ or an error $\bot$.}[.5\textwidth][c]{
\begin{tikzpicture}\small

\def\u{.236} 

\node[largeResource] (keyBox) at (0,0) {};
\node (alice) at (-2.5,\u) {Alice};
\node (bob) at (2.5,\u) {Bob};
\node (eve) at (0,-1.7) {Eve};
\node[draw] (key) at (.8,\u/2-.5) {key};
\node (junc) at (-.5,0 |- key.center) {};

\draw[sArrow,<->] (alice) to node[pos=.2,auto] {$k,\bot$} node[pos=.8,auto] {$k,\bot$} (bob);
\draw[thick] (junc.center |- 0,\u) to (junc.center) to node[pos=.666] (handle) {} +(160:-.8);
\draw[thick] (.3,0 |- junc.center) to (key);
\draw[double] (eve) to node[pos=.2,auto] {$0,1$} (handle.center);

\end{tikzpicture}
}

\vspace{12pt}

\subcaptionbox[Ideal key \& filter]{\label{fig:qkd.resource.filter}The resource from
  \figref{fig:qkd.resource.switch} with a filter $\lozenge_E$, modeling the case
  with no adversary.}[.5\textwidth][c]{

\begin{tikzpicture}\small

\def\t{2.368} 
\def\u{-1.85} 
\def\w{-3.2} 
\def\v{.118} 
\def\s{.91}

\node[thinResource] (keyBox) at (0,0) {};
\node[draw] (key) at (0,\v/2-.25) {key};
\node (junc) at (-1.4,0 |- key.center) {};
\node[yshift=-1.5,above right] at (keyBox.north west) {\footnotesize
  Secret key $\aK$};
\node (alice) at (-\t,\v) {};
\node (bob) at (\t,\v) {};

\draw[sArrow,<->] (alice.center) to node[pos=.08,auto] {$k,\bot$} node[pos=.92,auto] {$k,\bot$} (bob.center);
\draw[thick] (junc.center |- 0,\v) to (junc.center) to node[pos=.472] (handle) {} +(160:-.8);
\draw[thick] (-.6,0 |- junc.center) to (key);

\node[thinResource] (sim) at (0,\u+.35) {};
\node[xshift=1.5,below left] at (sim.north west) {\footnotesize
  $\lozenge_E$};
\node[innersnode] (a1) at (handle |- 0,\u+.35) {};

\draw[double] (a1) to node[pos=.55,auto,swap] {$0,1$} (handle.center);

\end{tikzpicture}

\vspace{24pt}

}
\subcaptionbox[Ideal key \& simulator]{\label{fig:qkd.resource.sim}The resource from
  \figref{fig:qkd.resource.switch} with a
  simulator $\sigma_E$.}[.5\textwidth][c]{ 
\begin{tikzpicture}\small

\def\t{2.368} 
\def\u{-1.85} 
\def\w{-3.2} 
\def\v{.118} 
\def\s{1.05}

\node[thinResource] (keyBox) at (0,0) {};
\node[draw] (key) at (0,\v/2-.25) {key};
\node (junc) at (-1.4,0 |- key.center) {};
\node[yshift=-1.5,above right] at (keyBox.north west) {\footnotesize
  Secret key $\aK$};
\node (alice) at (-\t,\v) {};
\node (bob) at (\t,\v) {};

\draw[sArrow,<->] (alice.center) to node[pos=.08,auto] {$k,\bot$} node[pos=.92,auto] {$k,\bot$} (bob.center);
\draw[thick] (junc.center |- 0,\v) to (junc.center) to node[pos=.472] (handle) {} +(160:-.8);
\draw[thick] (-.6,0 |- junc.center) to (key);

\node[simulator] (sim) at (0,\u) {};
\node[xshift=1.5,below left] at (sim.north west) {\footnotesize
  $\sigma_E$};
\node[innersnode] (a1) at (-\s,\u+.35) {};
\node[innersnode] (a2) at (-\s,\u-.35) {};
\node[innersnode] (b2) at (\s,\u-.35) {};
\node[innersnode] (c2) at (0,\u-.35) {};

\node (evel) at (-\s,\w) {};
\node (evec) at (0,\w) {};
\node (ever) at (\s,\w) {};

\draw[double] (a1) to node[pos=.55,auto,swap] {$0,1$} (handle.center);
\draw[sArrow] (a2) to (evel.center);
\draw[sArrow] (evec.center) to (c2);
\draw[sArrow] (b2) to (ever.center);

\end{tikzpicture}
}
\caption[Secret key resources]{\label{fig:qkd.resource}Some depictions
  of shared secret key resources, with filter and simulator converters
  in the last two.}
\end{figure}

However, if we wish to realize such a functionality with QKD, there is
a caveat: an eavesdropper can always prevent any real QKD protocol
from generating a key by cutting or jumbling the communication lines
between Alice and Bob, and this must be reflected in the definition of
the ideal resource. This box thus also has an interface accessible to
Eve, which provides her with a switch that, when pressed, prevents the
box from generating this key. We depict this in
\figref{fig:qkd.resource.switch}.

If modeled with the secret key resource of
\figref{fig:qkd.resource.switch}, the one-time pad is trivially secure
conditioned on Eve preventing a key from being distributed \--- in
this case, Alice and Bob do not have a key and do not run the one-time
pad. The security of the one-time pad is thus reduced to the case
where a key is generated, which corresponds to
\figref{fig:qkd.resource.simple} and is the situation analyzed in
\secref{sec:ac.otp}.

If no adversary is present, a filter covers Eve's interface of the
resource, making it inaccessible to the distinguisher. This filter
emulates the honest behavior that one expects in the case of a
non\-/malicious noisy channel. For a protocol and noisy channel that
together produce a key with probability $1-\delta$, the filter should
flip the switch on the $E$\=/interface of the ideal key with
probability $\delta$. This is illustrated in
\figref{fig:qkd.resource.filter}, and discussed in more detail in
\secref{sec:security.rob}.

\begin{rem}[Adaptive key length] \label{rem:adaptive} For a protocol
  to construct the shared secret key resource of
  \figref{fig:qkd.resource.switch}, it must either abort or produce a
  key of a fixed length. A more practical protocol could adapt the
  secret key length to the noise level on the quantum channel. This
  provides the adversary with the functionality to control the key
  length (not only whether it gets generated or not), and can be
  modeled by allowing the key length to be input at Eve's interface of
  the ideal key resource.\end{rem}

\subsection{Real protocol}
\label{sec:qkd.protocol}

To construct the secret key resource of
\figref{fig:qkd.resource.switch}, a QKD protocol uses some other
resources: a two-way authentic classical channel and an insecure
quantum channel. An authentic channel faithfully transmits messages
between Alice and Bob, but provides Eve with a copy as well.
An insecure channel is completely under the control of Eve, she can
apply any operation allowed by physics to the message on the
channel. If Eve does not intervene, some noise might still be present
on the channel, which is modeled by a filter that prevents Eve from
reading the message, but introduces honest noise instead.
Since an authentic channel can be constructed from an insecure channel
and a short shared secret key,\footnote{In fact, a short
  \emph{non\-/uniform} key is sufficient for
  authentication~\cite{RW03}, see \footnoteref{fn:password}.} QKD is
sometimes referred to as a \emph{key expansion} protocol.\footnote{We
  model QKD this way in \appendixref{app:ex.auth.qkd}.}

A QKD protocol typically has three phases: quantum state distribution,
error estimation and classical post\-/processing (for a detailed
review of QKD see \cite{SBCDLP09}). In the first, Alice sends some
quantum states on the insecure channel to Bob, who measures them upon
reception, obtaining a classical string. In the error estimation
phase, they communicate on the (two-way) authentic classical channel
to sample some bits at random positions in the string and estimate the
noise on the quantum channel by comparing these values to what Bob
should have obtained. If the noise level is above a certain threshold,
they abort the protocol and output an error message. If the noise is
low enough, they move on to the third phase, and make use of the
authentic channel to perform error correction and privacy
amplification on their respective strings, resulting in keys $k_A$ and
$k_B$ (which, ideally, should be equal). We sketch this in
\figref{fig:qkd.real}.

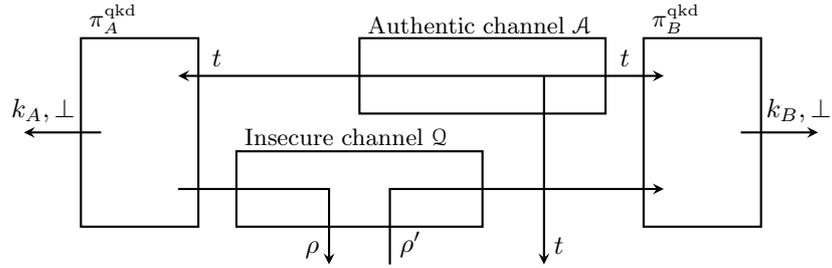
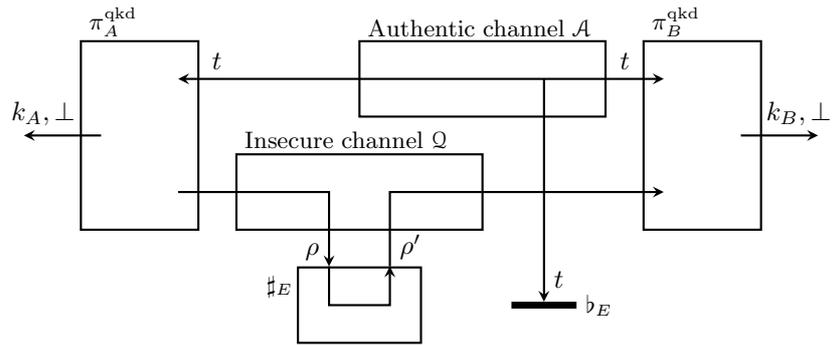
\begin{figure}[htbp]
\begin{subfigure}[b]{\textwidth}
\begin{centering}

\begin{tikzpicture}\small

\def\t{5.222} 
\def\u{3.7} 
\def\v{.75}
\def\w{.809}

\node[pnode] (a1) at (-\u,\v) {};
\node[pnode] (a2) at (-\u,0) {};
\node[pnode] (a3) at (-\u,-\v) {};
\node[protocol] (a) at (-\u,0) {};
\node[yshift=-2,above right] at (a.north west) {\footnotesize
  $\pi^{\qkd}_A$};
\node (alice) at (-\t,0) {};

\node[pnode] (b1) at (\u,\v) {};
\node[pnode] (b2) at (\u,0) {};
\node[pnode] (b3) at (\u,-\v) {};
\node[protocol] (b) at (\u,0) {};
\node[yshift=-2,above right] at (b.north west) {\footnotesize $\pi^{\qkd}_B$};
\node (bob) at (\t,0) {};

\node[thinResource] (cch) at (\w,\v) {};
\node[yshift=-2,above right] at (cch.north west) {\footnotesize
  Authentic channel  $\aA$};
\node[thinResource] (qch) at (-\w,-\v) {};
\node[yshift=-1.5,above right] at (qch.north west) {\footnotesize
  Insecure channel $\aQ$};
\node (eveq1) at (-\w-.4,-1.75) {};
\node (junc1) at (eveq1 |- a3) {};
\node (eveq2) at (-\w+.4,-1.75) {};
\node (junc2) at (eveq2 |- a3) {};
\node (evec) at (\w+\w,-1.75) {};
\node (junc3) at (evec |- b1) {};

\draw[sArrow,<->] (a1) to node[auto,pos=.08] {$t$} node[auto,pos=.92] {$t$}  (b1);
\draw[sArrow] (junc3.center) to node[auto,pos=.9] {$t$} (evec.center);

\draw[sArrow] (a2) to node[auto,pos=.75,swap] {$k_{A},\bot$} (alice.center);
\draw[sArrow] (b2) to node[auto,pos=.75] {$k_{B},\bot$} (bob.center);

\draw[sArrow] (a3) to (junc1.center) to node[pos=.8,auto,swap] {$\rho$} (eveq1.center);
\draw[sArrow] (eveq2.center) to node[pos=.264,auto,swap] {$\rho'$} (junc2.center) to (b3);

\end{tikzpicture}

\end{centering}
\caption[With adversary]{\label{fig:qkd.real.adv}When Eve is present, her interface
  gives her complete controle of the insecure channel and allows her
  to read the messages on the authentic channel.}
\end{subfigure}

\vspace{12pt}

\begin{subfigure}[b]{\textwidth}
\begin{centering}

\begin{tikzpicture}\small

\def\t{5.222} 
\def\u{3.7} 
\def\v{.75}
\def\w{.809}

\node[pnode] (a1) at (-\u,\v) {};
\node[pnode] (a2) at (-\u,0) {};
\node[pnode] (a3) at (-\u,-\v) {};
\node[protocol] (a) at (-\u,0) {};
\node[yshift=-2,above right] at (a.north west) {\footnotesize
  $\pi^{\qkd}_A$};
\node (alice) at (-\t,0) {};

\node[pnode] (b1) at (\u,\v) {};
\node[pnode] (b2) at (\u,0) {};
\node[pnode] (b3) at (\u,-\v) {};
\node[protocol] (b) at (\u,0) {};
\node[yshift=-2,above right] at (b.north west) {\footnotesize $\pi^{\qkd}_B$};
\node (bob) at (\t,0) {};

\node[thinResource] (cch) at (\w,\v) {};
\node[yshift=-2,above right] at (cch.north west) {\footnotesize
  Authentic channel $\aA$};
\node[thinResource] (qch) at (-\w,-\v) {};
\node[yshift=-1.5,above right] at (qch.north west) {\footnotesize
  Insecure channel $\aQ$};

\node[filter] (qchf) at (-\w,-3*\v) {};
\node[xshift=2,below left] at (qchf.north west) {\footnotesize $\sharp_E$};
\node[lineFilter] (cchf) at (2*\w,-3*\v) {};
\node[right] at (cchf.east) {\footnotesize $\flat_E$};

\node[xshift=-.4cm] (qchfl) at (qchf.center) {};
\node[xshift=.4cm] (qchfr) at (qchf.center) {};
\node (qchl) at (qchfl |- qchf.north) {};
\node (qchr) at (qchfr |- qchf.north) {};
\node (junc1) at (qchfl |- a3) {};
\node (junc2) at (qchfr |- a3) {};

\node (junc3) at (cchf.center |- b1) {};

\draw[sArrow,<->] (a1) to node[auto,pos=.08] {$t$} node[auto,pos=.92] {$t$}  (b1);
\draw[sArrow] (junc3.center) to node[auto,pos=.9] {$t$} (cchf);

\draw[sArrow] (a2) to node[auto,pos=.75,swap] {$k_{A},\bot$} (alice.center);
\draw[sArrow] (b2) to node[auto,pos=.75] {$k_{B},\bot$} (bob.center);

\draw[sArrow] (a3) to (junc1.center) to node[pos=.8,auto,swap] {$\rho$} (qchl.center);
\draw[sArrow] (qchr.center) to node[pos=.264,auto,swap] {$\rho'$} (junc2.center) to (b3);

\draw[sArrow] (qchl.center) to (qchfl.center) to (qchfr.center) to (qchr.center);

\end{tikzpicture}

\end{centering}
\caption[Without adversary]{\label{fig:qkd.real.filter}When no eavesdropper is present,
  filters forward Alice's quantum messages to Bob and block the
  authentic channel's output at the $E$\=/interface. The filter
  $\sharp_E$ might produce non\-/malicious noise that modifies $\rho$
  and models a (honest) noisy channel.}
\end{subfigure}

\caption[QKD system]{\label{fig:qkd.real}The real QKD system \---
  Alice has access to the left interface, Bob to the right interface
  and Eve to the lower interface \--- consists of the protocol
  $(\pi^{\qkd}_A,\pi^{\qkd}_B)$, the insecure quantum channel $\aQ$
  and two-way authentic classical channel $\aA$. Alice and Bob abort
  if the insecure channel is too noisy, i.e., if $\rho'$ is not
  similar enough to $\rho$ to obtain a secret key of the desired
  length. They run the classical post\-/processing over the authentic
  channel, obtaining keys $k_A$ and $k_B$. The message $t$ depicted on
  the two-way authentic channel represents the entire classical
  transcript of the classical post\-/processing.}
\end{figure}

\begin{rem}[Source of entanglement] \label{rem:entanglement} In this
  work we use an insecure quantum channel from Alice to Bob to
  construct the shared secret key resource. An alternative resource
  that is frequently used in QKD instead of this insecure channel, is
  a source of entangled states under the control of Eve. The source
  sends half of an entangled state to Alice and another half to
  Bob. It can be modeled similarly to the insecure channel depicted in
  \figref{fig:qkd.real}, but with the first arrow reversed: the states
  are sent from Eve to Alice and from Eve to Bob. \end{rem}

\subsection{Security}
\label{sec:qkd.security}

Let $(\pi^\qkd_A,\pi^\qkd_B)$ be the QKD protocol. Let $\aQ$ and $\aA$
be the insecure quantum channel and authentic classical channel,
respectively, with their filters $\sharp_E$ and $\flat_E$. Let $\aK$
denote the secret key resource of \figref{fig:qkd.resource.switch} and
let $\lozenge_E$ be its filter. Applying \defref{def:security}, we
find that $(\pi_A^{\qkd},\pi_B^{\qkd})$ constructs $\aK_\lozenge$ from
$\aQ_\sharp$ and $\aA_\flat$ within $\eps$ if
\begin{align}
  \pi_A^{\qkd}\pi_B^{\qkd}(\aQ \| \aA) (\sharp_E \| \flat_E )&
  \close{\eps} \aK
  \lozenge_E \label{eq:qkd.robust} \\
  \intertext{and} \exists \sigma_E, \quad \pi_A^{\qkd}\pi_B^{\qkd}(\aQ
  \| \aA) & \close{\eps} \aK \sigma_E \ . \label{eq:qkd.security}
\end{align}
The left- and right-hand sides of \eqnref{eq:qkd.robust} are
illustrated in Figures~\ref{fig:qkd.real.filter} and
\ref{fig:qkd.resource.filter}, and the left- and right-hand sides of
\eqnref{eq:qkd.security} are illustrated in
Figures~\ref{fig:qkd.real.adv} and \ref{fig:qkd.resource.sim}. These
two conditions are decomposed into simpler criteria in
\secref{sec:security}.

\section{Security reduction}
\label{sec:security}

By applying the general AC security definition to QKD, we obtained two
criteria, \eqnsref{eq:qkd.robust} and \eqref{eq:qkd.security},
capturing availability and security, respectively. In this section we
derive \eqnref{eq:d}, the trace distance criterion discussed in the
introduction, from \eqnref{eq:qkd.security}. We first show in
\secref{sec:security.dist} that the distinguishing advantage used in
the previous sections reduces to the trace distance between the
quantum states gathered by the distinguisher interacting with the real
and ideal systems. Then in \secref{sec:security.simulator}, we fix the
simulator $\sigma_E$ from the ideal system. In
\secref{sec:security.simple} we decompose the resulting security
criterion into a combination of \emph{secrecy} \--- \eqnref{eq:d} \---
and \emph{correctness} \--- the probability that Alice's and Bob's
keys differ. In the last section, \ref{sec:security.rob}, we consider
the security condition of \eqnref{eq:qkd.robust}, which captures
whether, in the absence of a malicious adversary, the protocol behaves
as specified by the ideal resource and corresponding filter.  We show
how this condition can be used to model the \emph{robustness} of the
protocol \--- the probability that the protocol aborts with
non\-/malicious noise.

\subsection{Trace distance}
\label{sec:security.dist}

The security criteria given in \eqnsref{eq:qkd.robust} and
\eqref{eq:qkd.security} are defined in terms of the distinguishing
advantage between resources. To simplify these equations, we rewrite
them in terms of the trace distance, $D(\cdot,\cdot)$. A formal
definition of this metric is given
in \appendixref{app:op.definitions}, along with a discussion of how to
interpret it in the rest of \appendixref{app:op}. We start with the
simpler case of \eqnref{eq:qkd.robust} in the next paragraph, then
deal with \eqnref{eq:qkd.security} after that.

The two resources on the left- and right-hand sides of
\eqnref{eq:qkd.robust} simply output classical strings (a key or error
message) at Alice and Bob's interfaces. Let these pairs of strings be
given by the joint probability distributions $P_{AB}$ and
$\tilde{P}_{AB}$. The distinguishing advantage between these systems
is thus simply the distinguishing advantage between these probability
distributions \--- a distinguisher is given a pair of strings sampled
according to either $P_{AB}$ or $\tilde{P}_{AB}$ and has to guess from
which distribution it was sampled \--- i.e., \[ d\left(
  \pi_A^{\qkd}\pi_B^{\qkd}(\aQ \| \aA) (\sharp_E \| \flat_E ),
  \aK\lozenge_E \right) = d(P_{AB},\tilde{P}_{AB}) \ . \] The
distinguishing advantage between two probability distributions is
equal to their total variation distance\footnote{The total variation
  distance between two probability distributions is equivalent to the
  trace distance between the corresponding (diagonal) quantum
  states. We use the same notation for both metrics, $D(\cdot,\cdot)$,
  since the former is a special case of the latter.} \--- which we
prove in in \appendixref{app:op.distadv} \--- i.e.,
$d(P_{AB},\tilde{P}_{AB})= D(P_{AB},\tilde{P}_{AB})$.  Putting the two
together we get
\[ d\left( \pi_A^{\qkd}\pi_B^{\qkd}(\aQ \| \aA) (\sharp_E \| \flat_E )
  , \aK\lozenge_E \right) = D(P_{AB},\tilde{P}_{AB}) \ ,\] where
$P_{AB}$ and $\tilde{P}_{AB}$ are the distributions of the strings
output by the real and ideal systems, respectively.

The resources on the left- and right-hand sides of
\eqnref{eq:qkd.security} are slightly more complex. They first output
a state $\varphi_C$ at the $E$\=/interface, namely the quantum states
prepared by Alice, which she sends on the insecure quantum
channel. Without loss of generality, the distinguisher now applies any
map $\cE : \lo{C} \to \lo{CE'}$ allowed by quantum physics to this
state, obtaining $\rho_{CE'} = \cE(\varphi_C)$ and puts the $C$
register back on the insecure channel for Bob, keeping the part in
$E'$. Finally, the systems output some keys (or error messages) at the
$A$ and $B$\=/interfaces, and a transcript of the post\-/processing at
the $E$\=/interface. Let $\rho^{\cE}_{ABE}$ denote the tripartite
state held by a distinguisher interacting with the real system, and
let $\tilde{\rho}^{\cE}_{ABE}$ denote the state held after interacting
with the ideal system, where the registers $A$ and $B$ contain the
final keys or error messages, and the register $E$ holds both the state
$\rho_{E'}$ obtained from tampering with the quantum channel and the
post-processing transcript. Distinguishing between these two systems
thus reduces to maximizing over the distinguisher strategies (the
choice of $\cE$) and distinguishing between the resulting states,
$\rho^{\cE}_{ABE}$ and $\tilde{\rho}^{\cE}_{ABE}$:
\[ d\left( \pi_A^{\qkd}\pi_B^{\qkd}(\aQ \| \aA) , \aK \sigma_E \right)
= \max_{\cE} d\left( \rho^{\cE}_{ABE},\tilde{\rho}^{\cE}_{ABE} \right)
\ . \] The advantage a distinguisher has in guessing whether it holds
the state $\rho^{\cE}_{ABE}$ or $\tilde{\rho}^{\cE}_{ABE}$ is given by
the trace distance between these states, i.e., \[d\left(
  \rho^{\cE}_{ABE},\tilde{\rho}^{\cE}_{ABE} \right) =
D\left(\rho^{\cE}_{ABE},\tilde{\rho}^{\cE}_{ABE} \right) \ . \] This
was first proven by Helstrom~\cite{Hel76}. For completeness, we
provide a proof in \appendixref{app:op.distadv},
\thmref{thm:op.distinguishing}.

The distinguishing advantage between the real and ideal systems of
\eqnref{eq:qkd.security} thus reduces to the trace distance between
the quantum states gathered by the distinguisher. In the following, we
usually omit ${\cE}$ where it is clear that we are maximizing over the
distinguisher strategies, and simply express the security criterion
as \begin{equation} \label{eq:qkd.security.1}
  D(\rho_{ABE},\tilde{\rho}_{ABE}) \leq \eps \ , \end{equation} where
$\rho_{ABE}$ and $\tilde{\rho}_{ABE}$ are the quantum states gathered
by the distinguisher interacting with the real and ideal systems,
respectively.

\subsection{Simulator}
\label{sec:security.simulator}

In the real setting (\figref{fig:qkd.real.adv}), Eve has full control
over the quantum channel and obtains the entire classical transcript
of the protocol. So for the real and ideal settings to be
indistinguishable, a simulator $\sigma^{\qkd}_E$ must generate the
same communication as in the real setting. This can be done by
internally running Alice's and Bob's protocol
$(\pi^{\qkd}_A,\pi^{\qkd}_B)$, producing the same messages at Eve's
interface as the real system. However, instead of letting this
(simulated) protocol decide the value of the key as in the real
setting, the simulator only checks whether they actually produce a key
or an error message, and presses the switch on the secret key resource
accordingly. We illustrate this in \figref{fig:qkd.ideal}.


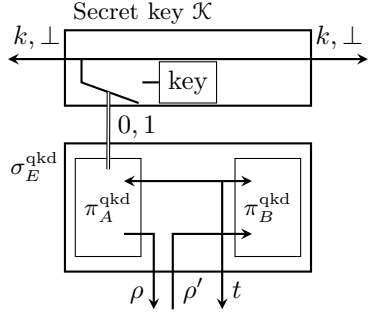
\begin{figure}[htb]
\begin{centering}

\begin{tikzpicture}\small

\def\t{2.368} 
\def\u{-1.85} 
\def\w{-3.2} 
\def\v{.118} 
\def\s{1.05}

\node[thinResource] (keyBox) at (0,0) {};
\node[draw] (key) at (0,\v/2-.25) {key};
\node (junc) at (-1.4,0 |- key.center) {};
\node[yshift=-1.5,above right] at (keyBox.north west) {\footnotesize
  Secret key $\aK$};
\node (alice) at (-\t,\v) {};
\node (bob) at (\t,\v) {};

\draw[sArrow,<->] (alice.center) to node[pos=.08,auto] {$k,\bot$} node[pos=.92,auto] {$k,\bot$} (bob.center);
\draw[thick] (junc.center |- 0,\v) to (junc.center) to node[pos=.472] (handle) {} +(160:-.8);
\draw[thick] (-.6,0 |- junc.center) to (key);

\node[simulator] (sim) at (0,\u) {};
\node[xshift=1.5,below left] at (sim.north west) {\footnotesize
  $\sigma^{\qkd}_E$};
\node[innersim,draw] (pleft) at (-\s,\u) {\footnotesize $\pi^{\qkd}_A$};
\node[innersim,draw] (pright) at (\s,\u) {\footnotesize $\pi^{\qkd}_B$};
\node[innersnode] (a1) at (-\s,\u+.35) {};
\node[innersnode] (a2) at (-\s,\u-.35) {};
\node[innersnode] (b1) at (\s,\u+.35) {};
\node[innersnode] (b2) at (\s,\u-.35) {};

\node (evel) at (-.45,\w) {};
\node (juncl) at (evel |- a2) {};
\node (evec) at (-.2,\w) {};
\node (juncc) at (evec |- a2) {};
\node (ever) at (.45,\w) {};
\node (juncr) at (ever |- a1) {};

\draw[double] (a1) to node[pos=.55,auto,swap] {$0,1$} (handle.center);
\draw[sArrow,<->] (a1) to (b1);
\draw[sArrow] (juncr.center) to node[pos=.852,auto] {$t$} (ever.center);
\draw[sArrow] (a2) to (juncl.center) to node[pos=.805,auto,swap] {$\rho$} (evel.center);
\draw[sArrow] (evec.center) to node[pos=.25,auto,swap] {$\rho'$} (juncc.center) to (b2);

\end{tikzpicture}

\end{centering}
\caption[Simulator for QKD]{\label{fig:qkd.ideal}The ideal QKD system \--- Alice
  has access to the left interface, Bob to the right interface and Eve
  to the lower interface \--- consists of the ideal secret key
  resource and a simulator $\sigma^{\qkd}_E$.}
\end{figure}

The security criterion from \eqnref{eq:qkd.security.1} can now be
simplified by noting that with this simulator, the states of the ideal
and real systems are identical when no key is produced. The outputs at
Alice's and Bob's interfaces are classical, elements of the set
$\{\bot\} \cup \cK$, where $\bot$ symbolizes an error and $\cK$ is the
set of possible keys. The states of the real and ideal systems can be
written as
\begin{align*}  \rho_{ABE} & = p^\bot \proj{\bot_A,\bot_B} \tensor
  \rho^\bot_E \\
  & \qquad \qquad + \sum_{k_A,k_B \in \cK} p_{k_A,k_B} \proj{k_A,k_B} \tensor
  \rho^{k_A,k_B}_E \ ,\\
  \tilde{\rho}_{ABE} & = p^\bot \proj{\bot_A,\bot_B}
  \tensor \rho^\bot_E \\ & \qquad \qquad +\frac{1}{|\cK|} \sum_{k \in \cK} \proj{k,k}
  \tensor \sum_{k_A,k_B \in \cK} p_{k_A,k_B} \rho^{k_A,k_B}_E \ .
\end{align*}

Plugging these in \eqnref{eq:qkd.security.1} we get
\begin{equation} \label{eq:qkd.security.2} D\left(
    \rho_{ABE},\tilde{\rho}_{ABE}\right) = (1- p^{\bot})
  D\left(\rho^{\top}_{ABE},\tau_{AB} \tensor \rho^{\top}_{E}\right)
  \leq \eps \ , \end{equation}
where \begin{equation} \label{eq:qkd.security.tmp} \rho^{\top}_{ABE}
  \coloneqq \frac{1}{1- p^{\bot}} \sum_{k_A,k_B \in \cK} p_{k_A,k_B}
  \proj{k_A,k_B} \tensor \rho^{k_A,k_B}_E \end{equation} is the
renormalized state of the system conditioned on not aborting and
$\tau_{AB} \coloneqq \frac{1}{|\cK|} \sum_{k \in \cK} \proj{k,k}$ is a
perfectly uniform shared key.

\subsection{Correctness \& secrecy}
\label{sec:security.simple}

We now break \eqnref{eq:qkd.security.2} down into two components,
often referred to as \emph{correctness} and \emph{secrecy}, and
recover the security definition for QKD introduced in
\cite{RK05,BHLMO05,Ren05}. The correctness of a QKD protocol refers to
the probability that Alice and Bob end up holding different keys. We
say that a protocol is \emph{$\eps_{\corr}$\=/correct} if for all
adversarial strategies,
\begin{equation}
  \label{eq:qkd.cor}
  \Pr \left[ K_A \neq K_B \right] \leq \eps_{\corr} \ ,
\end{equation}
where $K_A$ and $K_B$ are random variables over the alphabet $\cK \cup
\{\bot\}$ describing Alice's and Bob's outputs.\footnote{This can
  equivalently be written as $(1-p^\bot)\Pr \left[ K^\top_A \neq
    K^\top_B \right] \leq \eps_{\corr}$, where $p^\bot$ is the
  probability of aborting and $K^\top_A$ and $K^\top_B$ are Alice and
  Bob's keys conditioned on not aborting.} The secrecy of a QKD
protocol measures how close the final key is to a distribution that is
uniform and independent of the adversary's system. Let $p^\bot$ be the
probability that the protocol aborts, and $\rho^\top_{AE}$ be the
resulting state of the $AE$ subsystems conditioned on not aborting. A
protocol is \emph{$\eps_{\secr}$\=/secret} 
if for all adversarial strategies,
\begin{equation}
  \label{eq:qkd.sec}
  (1-p^\bot) D\left(\rho^\top_{AE},\tau_A \tensor \rho^{\top}_E\right)
  \leq \eps_{\secr} \ ,
\end{equation}
where the distance $D(\cdot,\cdot)$ is the trace distance and $\tau_A$
is the fully mixed state.\footnote{\eqnref{eq:qkd.sec} is a
  reformulation of \eqnref{eq:d}.}

\begin{thm}
  \label{thm:qkd}
  If a QKD protocol is $\eps_{\corr}$\=/correct and
  $\eps_{\secr}$\=/secret, then \eqnref{eq:qkd.security} is satisfied
  for $\eps = \eps_{\corr} + \eps_{\secr}$.
\end{thm}

\begin{proof}
  Let us define $\gamma_{ABE}$ to be a state obtained from
  $\rho^{\top}_{ABE}$ (\eqnref{eq:qkd.security.tmp}) by throwing away
  the $B$ system and replacing it with a copy of $A$, i.e., \[
  \gamma_{ABE} = \frac{1}{1-p^\bot} \sum_{k_A,k_B \in \cK} p_{k_A,k_B}
  \proj{k_A,k_A} \tensor \rho^{k_A,k_B}_E \ .\] From the triangle
  inequality we get \[ D(\rho^\top_{ABE},\tau_{AB} \tensor
  \rho^\top_{E}) \leq D(\rho^\top_{ABE},\gamma_{ABE}) +
  D(\gamma_{ABE},\tau_{AB} \tensor \rho^\top_{E}) \ .\]

Since in the states $\gamma_{ABE}$ and $\tau_{AB} \tensor
\rho^\top_{E}$ the $B$ system is a copy of the $A$ system, it does not
modify the distance. Furthermore, $\trace[B]{\gamma_{ABE}} =
\trace[B]{\rho^{\top}_{ABE}}$. Hence \[D(\gamma_{ABE},\tau_{AB}
\tensor \rho^\top_{E}) = D(\gamma_{AE},\tau_{A} \tensor \rho^\top_{E})
= D(\rho^\top_{AE},\tau_{A} \tensor \rho^\top_{E})  \ .\]

For the other term note that
\begin{align*}
  & D(\rho^\top_{ABE},\gamma_{ABE}) \\
  & \qquad \leq \sum_{k_A,k_B} \frac{p_{k_A,k_B}}{1-p^{\bot}}
  D\left(\proj{k_A,k_B} \tensor \rho^{k_A,k_B}_E,\proj{k_A,k_A} \tensor \rho^{k_A,k_B}_E \right)\\
  & \qquad = \sum_{k_A \neq k_B} \frac{p_{k_A,k_B}}{1-p^{\bot}} = \frac{1}{1-p^{\bot}}\Pr
  \left[ K_A \neq K_B \right] \ .
\end{align*}
Putting the above together with \eqnref{eq:qkd.security.2}, we get
\begin{align*} D(\rho_{ABE},\tilde{\rho}_{ABE}) & = (1-p^\bot)
  D(\rho^\top_{ABE},\tau_{AB} \tensor \rho^\bot_{E}) \\ & \leq \Pr
  \left[ K_A \neq K_B \right] + (1-p^\bot) D(\rho^\top_{AE},\tau_{A}
  \tensor \rho^\top_{E}) \ . \qedhere \end{align*}
\end{proof}

\begin{rem}[Tightness of the security criteria] In \thmref{thm:qkd} we
  prove a bound on the second security condition of
  \defref{def:security} for QKD in terms of the correctness and
  secrecy of the protocol. The converse can also be shown: if
  \eqnref{eq:qkd.security} holds for some $\eps$, then the
  corresponding QKD protocol is both $\eps$\=/correct and
  $2\eps$\=/secret.\footnote{The factor $2$ is a result of the
    \emph{existence} of the simulator $\sigma_E$ in the security
    definition. We cannot exclude that for some specific QKD protocol
    there exists a different simulator $\bar{\sigma}^\qkd_E$ \---
    different from the one used in this proof \--- generating a state
    $\bar{\rho}_E$ when interacting with the distinguisher, such that
    $D\left(\rho^{\top}_{AE},\tau_{A} \tensor
      \bar{\rho}^{\top}_E\right) \leq D\left(\rho^{\top}_{AE},\tau_{A}
      \tensor \rho^{\top}_E\right)$.  However, by the triangle
    inequality we also have that for any $\bar{\rho}_E$,
    $D\left(\rho^{\top}_{AE},\tau_{A} \tensor
      \bar{\rho}^{\top}_E\right) \geq \frac{1}{2}
    D\left(\rho^{\top}_{AE},\tau_{A} \tensor
      \rho^{\top}_E\right)$. Hence the failure $\eps$ of the generic
    simulator used in this proof is at most twice larger than
    optimal.}\end{rem}

\subsection{Robustness}
\label{sec:security.rob}

So far in this section we have discussed the security of a QKD
protocol with respect to a malicious Eve using the second condition
from \defref{def:security} (\eqnref{eq:qkd.security}). A QKD protocol
which always aborts without producing any key trivially satisfies
\eqnref{eq:qkd.security} with $\eps=0$, but is not a useful protocol
at all! The real system must not only be indistinguishable from ideal
when an adversary is present, but also when the adversarial interfaces
are covered by filters emulating honest behavior. This is modeled by
the first condition from \defref{def:security}, namely
\eqnref{eq:qkd.robust} for QKD. If no adversary is tampering with the
quantum channel \--- only natural non\-/malicious noise is present
\--- we expect a secret key to be generated with high
probability. This can be captured by designing the filter $\lozenge_E$
to allow a key to be produced with high probability: if the real
system does not generate a key with the same probability, this
immediately results in a gap noticeable by the distinguisher.

The probability of a key being generated depends on the noise
introduced by the filter $\sharp_E$ covering the adversarial interface
of the insecure quantum channel $\aQ$ in the real system (illustrated
in \figref{fig:qkd.real.filter}). Suppose that this noise is
parametrized by a value $q$, e.g., a depolarizing channel with
probability $q$. For every $q$, the protocol has a probability of
aborting, $\delta$, which is called the \emph{robustness}. Let
$\sharp^q_E$ denote a filter of the channel $\aQ$ that models this
noise, and let $\lozenge^\delta_E$ denote the filter of the ideal key
resource $\aK$, which flips the switch to prevent a key from being
generated with corresponding probability
$\delta$. \eqnref{eq:qkd.robust} thus becomes
\begin{equation} \label{eq:robustness} \pi_A^{\qkd}\pi_B^{\qkd}(\aQ \|
  \aA) (\sharp^q_E \| \flat_E ) \close{\eps} \aK \lozenge^\delta_E \
  ,\end{equation} where varying $q$ and $\delta$ results in a family
of real and ideal systems.

We now prove that in this case the failure $\eps$ from
\eqnref{eq:robustness} is bounded by $\eps_{\corr}+\eps_{\secr}$. Note
that this statement is only useful if the probability of aborting,
$\delta$, is small for reasonable noise models $q$.

\begin{lem} \label{lem:robustness}
If the filters from \eqnref{eq:robustness} are parametrized such that
$\lozenge^\delta_E$ aborts with exactly the same probability as the
protocol $(\pi_A^{\qkd},\pi_B^{\qkd})$ run on the noisy channel
$\aQ\sharp^q_E$, then the availability of the protocol is bounded by
the security, i.e., 
\[ d\left( \pi_A^{\qkd}\pi_B^{\qkd}(\aQ \| \aA) (\sharp^q_E \| \flat_E
  ),\aK \lozenge^\delta_E\right) \leq d\left(
  \pi_A^{\qkd}\pi_B^{\qkd}(\aQ \| \aA),\aK \sigma^{\qkd}_E\right) \
,\] where the simulator $\sigma^{\qkd}_E$ is the one used in the
previous sections, introduced in \secref{sec:security.simulator},
\figref{fig:qkd.ideal}.
\end{lem}

\begin{proof}
  Since $\lozenge^\delta_E$ aborts with exactly the same probability
  as the real system and since $\sigma^{\qkd}_E$ simulates the real
  system, we can substitute $\sigma^{\qkd}_E(\sharp^q_E \| \flat_E )$
  for $\lozenge^\delta_E$. The result then follows, because the
  converter $\sharp^q_E \| \flat_E$ on both the real and ideal systems
  can only decrease their distance (\eqnref{eq:axioms.nonincrease}).
\end{proof}

\section{Examples of composition}
\label{sec:ex}

It is immediate from the AC framework~\cite{MR11} that the composition
of two protocols satisfying \defref{def:security} is still
secure.\footnote{See \appendixref{app:generic} for a proof
  sketch.} In this section we attempt to provide a better feeling for
protocol composition by illustrating it with several examples. We
compose QKD in series and in parallel, and show that \--- as a result
of the triangle inequality and the security of the individual
protocols \--- the corresponding composed real systems are
indistinguishable from the composed ideal systems.

In \secref{sec:ex.leak} we first look at a situation in which part of
the key is known to the adversary. In \secref{sec:ex.otp} we compose
QKD with a one-time pad. And in \secref{sec:ex.qkd} we compose two runs of
a QKD protocol in parallel. We provide a more extensive example of
protocol composition in \appendixref{app:ex.auth}, where we model
the security of authentication and compose it with QKD, resulting in a
key expansion protocol.

To simplify the examples, we only consider security in the presence of
an adversary and ignore the first condition from
\defref{def:security}. For the same reason, when writing up the
security condition with the trace distance, we hard-code the simulator
used in \secref{sec:security} in the security criterion. Furthermore,
as shown in \secref{sec:security.simulator}, conditioned on aborting,
the real and ideal systems of QKD are identical, so the security
criterion can be reduced to the case in which the QKD protocol
terminates with a shared key between Alice and Bob, which happens with
probability $1-\pabort$. With these simplifications, a QKD protocol is
$\eps$\=/secure if
\begin{equation} \label{eq:qkd.security.3} 
(1-\pabort) D\left(\rho_{ABE},\tau_{AB} \tensor \rho_{E}\right) \leq \eps \ ,
\end{equation}
where $\tau_{AB}$ is a perfect shared key and $\rho_{ABE}$ and
$\tau_{AB} \tensor \rho_E$ are the final states, conditioned on
producing a key, that the distinguisher holds after interacting with
the real and ideal systems, respectively.

\subsection{Partially known key}
\label{sec:ex.leak}

The accessible information given in \eqnref{eq:localqkd} is shown to
be insufficient to define security for a QKD protocol by considering a
setting in which part of the key $K$ is available to
Eve~\cite{KRBM07}. This allows her to guess the remaining bits of the
key, which would not have been possible had the key been distributed
using an ideal resource. We analyze exactly this setting here, and
argue that this does not affect the security of a QKD scheme that
satisfies \defref{def:security}.

To model this partial knowledge of the key, let Alice run a protocol
$\pi'_A$ that receives part of the secret key \--- generated either by a
QKD protocol or by an ideal resource \--- and sends it on a channel to
Eve. Plugging this in the real and ideal QKD systems from
Figures~\ref{fig:qkd.real.adv} and \ref{fig:qkd.ideal}, we get
\figref{fig:ex.leak}.

\begin{figure}[htbp]
\begin{subfigure}[b]{\textwidth}
\begin{centering}

\begin{tikzpicture}\small

\def\ta{6.067} 
\def\tb{4.022} 
\def\e{-2.5}
\def\u{2.5}
\def\uu{4.545} 
\def\v{.75}
\def\w{.9}

\node[pnode] (a1) at (-\u,\v) {};
\node[pnode] (a2) at (-\u,0) {};
\node[pnode] (a3) at (-\u,-\v) {};
\node[protocol] (a) at (-\u,0) {};
\node[yshift=-2,above right] at (a.north west) {\footnotesize
  $\pi^{\qkd}_A$};
\node (alice) at (-\ta,\v) {};

\node[pnode] (b1) at (\u,\v) {};
\node[pnode] (b2) at (\u,0) {};
\node[pnode] (b3) at (\u,-\v) {};
\node[protocol] (b) at (\u,0) {};
\node[yshift=-2,above right] at (b.north west) {\footnotesize $\pi^{\qkd}_B$};
\node (bob) at (\tb,0) {};

\node[pnode] (f1) at (-\uu,0) {};
\node[pnode] (f3) at (-\uu,-2*\v) {};
\node[protocol,gray] (f) at (-\uu,-\v) {};
\node[gray,xshift=2,below left] at (f.north west) {\footnotesize
  $\pi'_A$};

\node (eveq1) at (-\w,\e) {};
\node (junc1) at (eveq1 |- a3) {};
\node (eveq2) at (0,\e) {};
\node (junc2) at (eveq2 |- a3) {};
\node (evec) at (\w,\e) {};
\node (junc3) at (evec |- b1) {};
\node (evef) at (-2*\w,\e) {};
\node (junc0) at (evef |- f3) {};

\draw[sArrow,<->] (a1) to (b1);
\draw[sArrow] (junc3.center) to node[auto,pos=.9,swap] {$t$} (evec.center);

\draw[sArrow] (a2) to node[auto,pos=.5,swap] {$k^1_{A}$} (f1);
\draw[sArrow] (b2) to node[auto,pos=.6] {$k_{B}$} (bob.center);

\draw[sArrow] (a3) to (junc1.center) to node[pos=.85,auto,swap] {$\rho$} (eveq1.center);
\draw[sArrow] (eveq2.center) to node[pos=.19,auto] {$\rho'$}
(junc2.center) to (b3);

\draw[sArrow] (a1) to node[auto,pos=.5,swap] {$k^2_A$} (alice.center);
\draw[gray,sArrow] (f3) to (junc0.center) to node[pos=.7,auto,swap] {$k^1_A$} (evef.center);

\end{tikzpicture}

\end{centering}
\caption[QKD protocol]{\label{fig:ex.leak.real}The QKD protocol
  $(\pi_A^\qkd,\pi_B^\qkd)$ generates a pair of keys
  $(k_A,k_B)$. Alice then runs $\pi'_A$, which provides the first part
  of $k_A=k^1_A \| k^2_A$ to Eve. The drawing of the insecure quantum
  channel and authentic classical channel have been removed to
  simplify the figure.}
\end{subfigure}

\vspace{12pt}

\begin{subfigure}[b]{\textwidth}
\begin{centering}

\begin{tikzpicture}\small

\def\ta{4.913} 
\def\tb{2.368} 
\def\u{-1.85} 
\def\w{-3.95} 
\def\v{.118} 
\def\s{1.05}
\def\f{-3.5}
\def\fh{-1.85}

\node[thinResource] (keyBox) at (0,0) {};
\node[draw] (key) at (0,\v/2-.25) {key};
\node (junc) at (-1.4,0 |- key.center) {};
\node[yshift=-1.5,above right] at (keyBox.north west) {\footnotesize
  Secret key};
\node (alice) at (-\ta,\v) {};
\node (bob) at (\tb,\v) {};

\node[draw,thick,gray,minimum width=1.545cm,minimum height=3.2cm] (f) at (\f,\fh) {};
\node[pnode] (f1) at (\f,-.75) {};
\node[pnode] (f3) at (\f,-2.95) {};
\node[gray,xshift=2,below left] at (f.north west) {\footnotesize
  $\pi'_A$};

\draw[sArrow,<->] (alice.center) to node[pos=.2,auto] {$k^2$} node[pos=.94,auto] {$k$} (bob.center);
\draw[thick] (junc.center |- 0,\v) to (junc.center) to node[pos=.472] (handle) {} +(160:-.8);
\draw[thick] (-.6,0 |- junc.center) to (key);

\node[simulator] (sim) at (0,\u) {};
\node[xshift=-1.5,below right] at (sim.north east) {\footnotesize
  $\sigma^{\qkd}_E$};
\node[innersim] (pleft) at (-\s,\u) {};
\node[innersim] (pcenter) at (0,\u) {};
\node[innersim] (pright) at (\s,\u) {};

\node (evell) at (-2*\s,\w) {};
\node (juncll) at (evell |- f3) {};
\node (evel) at (-\s,\w) {};
\node (evec) at (0,\w) {};
\node (ever) at (\s,\w) {};

\node (junctt) at (evell |- alice) {};
\node (junct) at (evell |- f1) {};

\draw[double] (pleft) to node[pos=.55,auto,swap] {$0,1$} (handle.center);
\draw[sArrow] (pright) to node[pos=.8,auto] {$t$} (ever.center);
\draw[sArrow] (pleft) to node[pos=.805,auto] {$\rho$} (evel.center);
\draw[sArrow] (evec.center) to node[pos=.22,auto,swap] {$\rho'$} (pcenter);

\draw[sArrow] (junctt.center) to (junct.center) to node[pos=.4,auto,swap] {$k^1$} (f1);

\draw[gray,sArrow] (f3) to (juncll.center) to node[auto,pos=.7] {$k^1$} (evell.center);

\end{tikzpicture}

\end{centering}
\caption[Ideal key]{\label{fig:ex.leak.ideal}The ideal secret key resource
  generates a key $k=k^1\|k^2$, part of which is provided to Eve by $\pi'_A$. A
  simulator $\sigma^{\qkd}_E$ pads the ideal key resource to generate
  the same communication as in the real setting.}
\end{subfigure}

\caption[Example: partially known secret key]{\label{fig:ex.leak}Alice
  runs a protocol $\pi'_A$ which reveals the first half of her key to
  Eve. In each figure, Alice and Bob have access to the left and right
  interfaces, and Eve to the lower interface. If we remove the parts
  in gray we recover the real and ideal systems of QKD.}
\end{figure}
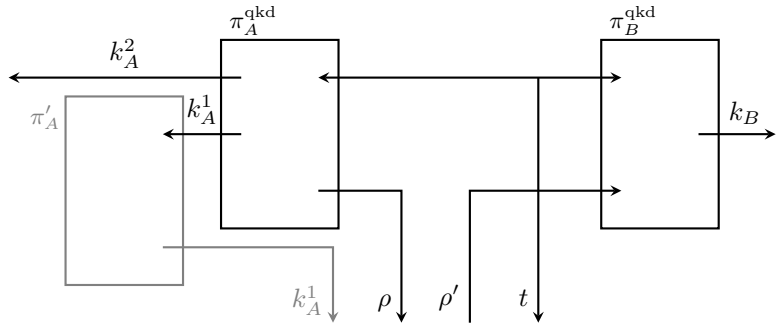
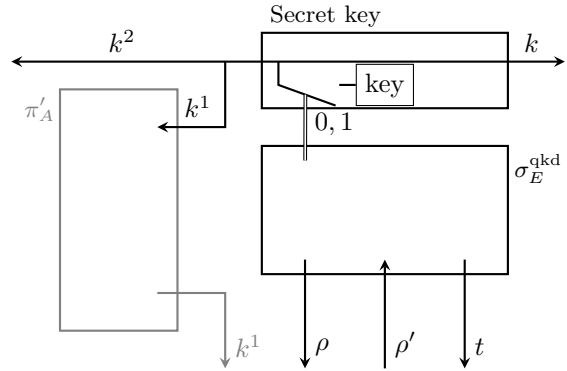

It is immediate from \figref{fig:ex.leak} that $\pi'_A$ cannot
increase the distance between the real and ideal systems and therefore
cannot compromise security: the systems in gray can be run internally by a
distinguisher attempting to guess whether it is interacting with the
real or ideal QKD system, so this case is already bounded by the
security of QKD.

This reasoning is summed up in the following equation, which can be
directly derived from \eqnref{eq:axioms.nonincrease}:
\[ \pi^{\qkd}_A \pi^{\qkd}_B \left( \aQ \| \aA \right) \close{\eps}
\aK \sigma^{\qkd}_E \implies \pi'_A \pi^{\qkd}_A \pi^{\qkd}_B \left(
  \aQ \| \aA \right) \close{\eps} \pi'_A \aK \sigma^{\qkd}_E \ . \]

The same can be obtained from the properties of the trace distance if
we write out explicitly the states gathered by the distinguisher. If
the QKD protocol is $\eps$\=/secure, we have from
\eqnref{eq:qkd.security.3} that
\[(1-\pabort) D(\rho_{ABE},\tau_{AB} \tensor \rho_{E}) \leq \eps,\]
where $\rho_{ABE}$ is the state gathered by a distinguisher
interacting with the real QKD system (\figref{fig:qkd.real.adv}) and
$\tau_{AB} \tensor \rho_{E}$ is the state gathered by interacting with
the ideal system (\figref{fig:qkd.ideal}), conditioned on the protocol
not aborting. A distinguisher interacting with either of the two
systems from \figref{fig:ex.leak} gets extra information at Eve's
interface, namely the first part of Alice's key $k^1$, and only the
second part of that key $k^2$ at Alice's interface. The complete
states gathered by interacting with \figref{fig:ex.leak.real} and
\figref{fig:ex.leak.ideal} are given by $\rho_{A'BE'} =
\rho_{A_2BA_1E}$ and $\tilde{\rho}_{A'BE'} = \tau_{A_2BA_1} \tensor
\rho_E$, respectively, where the orignal system $A=A_1A_2$ containing
Alice's key is split in two, $A'=A_2$ and $E'=A_1E$. These can be
obtained from $\rho_{ABE}$ and $\tau_{AB} \tensor \rho_E$ by a unitary
map which simply permutes the registers. Thus, the trace distance does
not increase. So we have
\begin{equation*} (1-\pabort) D\left(\rho_{A'BE'},\tilde{\rho}_{A'BE'}\right) =
  (1-\pabort) D\left(\rho_{ABE},\tau_{AB} \tensor \rho_E\right) \leq \eps \ .
\end{equation*}

If we analyze the same situation from the perspective of an adversary
that can access only the $E$\=/interface, composing QKD with a
protocol that reveals $k^1$ results in a net gain of information for
this adversary. But as shown above, for a distinguisher that also
receives the outputs of the honest players \--- the generated secret
keys \--- there is no gain. 

\subsection{Sequential composition of key distribution and one-time
  pad}
\label{sec:ex.otp}

If we compose a one-time pad (depicted in \figref{fig:otp.real}) and a
QKD protocol (depicted in \figref{fig:qkd.real.adv}), we obtain
\figref{fig:ex.otp.real}, where the secret key resource used by the
one-time pad is replaced by the QKD protocol. We showed in
\secref{sec:ac.otp} that a one-time pad constructs a secure channel
(\figref{fig:otp.ideal.resource}), which provides Eve with only one
functionality, learning the length of the message. However, this was
if the one-time pad protocol had access to a secret key resource with
a blank $E$\=/interface, as in \figref{fig:qkd.resource.simple}. In
reality, QKD constructs a resource that allows Eve to prevent a key
from being generated, as in \figref{fig:qkd.resource.switch}. It can
easily be shown that with access to this resource, a one-time pad
constructs a secure channel with two controls at Eve's interface: one
for preventing any message from being sent and a second for learning
the length of the message if she did not activate the first. This
resource is illustrated in \figref{fig:ex.otp.ideal}, along with the
appropriate simulator for constructing this resource with a one-time
pad and a QKD protocol: the combination of the two simulators used in
the individual proofs of the one-time pad (\figref{fig:otp.ideal}) and
QKD (\figref{fig:qkd.ideal}).

\begin{figure}[htbp]
\begin{subfigure}[b]{\textwidth}
\begin{centering}

\begin{tikzpicture}\small

\def\t{6.067} 
\def\e{-2.5}
\def\u{2.5}
\def\uu{4.545} 
\def\v{.75}
\def\w{.9}

\node[pnode] (a1) at (-\u,\v) {};
\node[pnode] (a2) at (-\u,0) {};
\node[pnode] (a3) at (-\u,-\v) {};
\node[protocol] (a) at (-\u,0) {};
\node[yshift=-2,above right] at (a.north west) {\footnotesize
  $\pi^{\qkd}_A$};
\node (alice) at (-\t,-\v) {};

\node[pnode] (b1) at (\u,\v) {};
\node[pnode] (b2) at (\u,0) {};
\node[pnode] (b3) at (\u,-\v) {};
\node[protocol] (b) at (\u,0) {};
\node[yshift=-2,above right] at (b.north west) {\footnotesize $\pi^{\qkd}_B$};
\node (bob) at (\t,-\v) {};

\node[pnode] (f1) at (-\uu,0) {};
\node[pnode] (f2) at (-\uu,-\v) {};
\node[pnode] (f3) at (-\uu,-2*\v) {};
\node[protocol,gray,text width=1.2cm] (f) at (-\uu,-\v) {\footnotesize $y
  =$\\$\ x_A \xor k_A$};
\node[gray,yshift=-2,above right] at (f.north west) {\footnotesize
  $\pi^{\otp}_A$};

\node[pnode] (g1) at (\uu,0) {};
\node[pnode] (g2) at (\uu,-\v) {};
\node[pnode] (g3) at (\uu,-2*\v) {};
\node[protocol,gray,text width=1.2cm] (g) at (\uu,-\v) {\footnotesize $x_B
  =$\\$\ \ y \xor k_A$};
\node[gray,yshift=-2,above right] at (g.north west) {\footnotesize
  $\pi^{\otp}_B$};

\node (eveq1) at (-\w,\e) {};
\node (junc1) at (eveq1 |- a3) {};
\node (eveq2) at (0,\e) {};
\node (junc2) at (eveq2 |- a3) {};
\node (evec) at (\w,\e) {};
\node (junc3) at (evec |- b1) {};
\node (evef) at (2*\w,\e) {};
\node (junc4) at (evef |- f3) {};

\draw[sArrow,<->] (a1) to (b1);
\draw[sArrow] (junc3.center) to node[auto,pos=.9,swap] {$t$} (evec.center);

\draw[sArrow] (a2) to node[auto,pos=.5,swap] {$k_{A}$} (f1);
\draw[sArrow] (b2) to node[auto,pos=.5] {$k_{B}$} (g1);

\draw[sArrow] (a3) to (junc1.center) to node[pos=.85,auto,swap] {$\rho$} (eveq1.center);
\draw[sArrow] (eveq2.center) to node[pos=.19,auto] {$\rho'$}
(junc2.center) to (b3);

\draw[gray,sArrow] (alice.center) to node[auto,pos=.4] {$x_A$} (f2);
\draw[gray,sArrow] (g2) to node[auto,pos=.6] {$x_B$} (bob.center);
\draw[gray,sArrow] (f3) to node[pos=.07,auto,swap] {$y$} node[pos=.93,auto,swap] {$y$} (g3);
\draw[gray,sArrow] (junc4.center) to node[pos=.7,auto,swap] {$y$} (evef.center);

\end{tikzpicture}

\end{centering}
\caption[One-time pad \& QKD]{\label{fig:ex.otp.real}The composition
  of a one-time pad and a QKD protocol. The authentic channels and
  insecure quantum channels used have not been depicted as boxes to
  simplify the figure.}
\end{subfigure}

\vspace{12pt}

\begin{subfigure}[b]{\textwidth}
\begin{centering}

\begin{tikzpicture}\small

\def\t{4.913} 
\def\u{-1.85} 
\def\w{-3.95} 
\def\v{.118} 
\def\s{1.05}
\def\f{3.5}
\def\fh{-1.475}
\def\vv{-2.95}

\node[thinResource] (keyBox) at (0,0) {};
\node[draw] (key) at (0,\v/2-.25) {key};
\node (junc) at (-1.4,0 |- key.center) {};
\node[yshift=-1.5,above right] at (keyBox.north west) {\footnotesize
  Secret key $\aK$};
\node (alice) at (-\t,\fh) {};
\node (bob) at (\t,\fh) {};

\node[draw,thick,gray,minimum width=1.545cm,minimum height=3.95cm,text
  width=1.1cm] (f) at (-\f,\fh) {\footnotesize $y =$\\$\quad x \xor k$};
\node[pnode] (f1) at (-\f,\v) {};
\node[pnode] (f2) at (-\f,\fh) {};
\node[pnode] (f3) at (-\f,\vv) {};
\node[gray,yshift=-2,above right] at (f.north west) {\footnotesize
  $\pi^{\otp}_A$};

\node[draw,thick,gray,minimum width=1.545cm,minimum height=3.95cm,text
  width=1.2cm] (g) at (\f,\fh) {\footnotesize $x =$\\$\quad y \xor k$};

\node[pnode] (g1) at (\f,\v) {};
\node[pnode] (g2) at (\f,\fh) {};
\node[pnode] (g3) at (\f,\vv) {};
\node[gray,yshift=-2,above right] at (g.north west) {\footnotesize
  $\pi^{\otp}_B$};

\draw[sArrow,<->] (f1) to node[pos=.15,auto] {$k$} node[pos=.85,auto] {$k$} (g1);
\draw[thick] (junc.center |- 0,\v) to (junc.center) to node[pos=.472] (handle) {} +(160:-.8);
\draw[thick] (-.6,0 |- junc.center) to (key);

\node[simulator,dashed] (sim) at (0,\u) {};
\node[xshift=-1.5,below right] at (sim.north east) {\footnotesize
  $\sigma^{\qkd}_E$};
\node[innersim] (pleft) at (-\s,\u) {};
\node[innersim] (pcenter) at (0,\u) {};
\node[innersim] (pright) at (\s,\u) {};

\node (everr) at (2*\s,\w) {};
\node (juncrr) at (everr |- f3) {};
\node (evel) at (-\s,\w) {};
\node (evec) at (0,\w) {};
\node (ever) at (\s,\w) {};

\draw[double] (pleft) to node[pos=.55,auto,swap] {$0,1$} (handle.center);
\draw[dashed,sArrow] (pright) to node[pos=.8,auto,swap] {$t$} (ever.center);
\draw[dashed,sArrow] (pleft) to node[pos=.805,auto,swap] {$\rho$} (evel.center);
\draw[dashed,sArrow] (evec.center) to node[pos=.22,auto] {$\rho'$} (pcenter);

\draw[gray,sArrow] (f3) to node[auto,pos=.1,swap] {$y$} node[auto,pos=.9,swap] {$y$} (g3);
\draw[gray,sArrow] (juncrr.center) to node[auto,pos=.7,swap] {$y$} (everr.center);
\draw[gray,sArrow] (alice.center) to node[auto,pos=.4] {$x_A$} (f2);
\draw[gray,sArrow] (g2) to node[auto,pos=.6] {$x_B$} (bob.center);

\end{tikzpicture}

\end{centering}
\caption[One-time pad \& ideal key]{\label{fig:ex.otp.hybrid}A hybrid system consisting of a real
one-time pad and ideal secret key resource with simulator.}
\end{subfigure}

\vspace{12pt}

\begin{subfigure}[b]{\textwidth}
\begin{centering}

\begin{tikzpicture}\small

\def\b{7.436cm} 
\def\t{4.468} 
\def\u{-1.85} 
\def\s{1.05}
\def\w{-3.45} 

\node[draw,thick,minimum width=\b,minimum height=1cm] (channel) at (0,0) {};
\node[yshift=-1.5,above right] at (channel.north west) {\footnotesize
  Secure channel $\aS$};
\node (alice) at (-\t,0) {};
\node (bob) at (\t,0) {};

\node[simulator,dashed] (simqkd) at (-2*\s,\u) {};
\node[xshift=1.5,below left] at (simqkd.north west) {\footnotesize
  $\sigma^{\qkd}_E$};
\node[innersim] (pleft) at (-3*\s,\u) {};
\node[innersim] (pcenter) at (-2*\s,\u) {};
\node[innersim] (pright) at (-\s,\u) {};

\node[simulator] (simotp) at (2*\s,\u) {};
\node[xshift=-1.5,below right] at (simotp.north east) {\footnotesize
  $\sigma^{\otp}_E$};
\node[innersim] (qleft) at (\s,\u) {};
\node[innersim] (qcenter) at (2*\s,\u) {};
\node[innersim] (qright) at (3*\s,\u) {};

\node[xshift=-.4cm] (ajunc) at (pcenter |- alice) {};
\draw[thick] (alice) to node[pos=.2,auto] {$x$} (ajunc.center) to node[pos=.528] (handle) {} +(160:-.8);
\draw[double] (pcenter) to node[pos=.35,auto,swap] {$0,1$} (handle.center);
\node[xshift=.4cm] (bjunc) at (pcenter |- bob) {};
\draw[sArrow] (bjunc.center) to node[pos=.96,auto] {$x$} (bob);
\node (leak) at (qcenter |- bob) {};
\draw[dotted,sArrow] (leak.center) to node[pos=.65,auto,swap] {$|x|$} (qcenter);

\node (peve1) at (-3*\s,\w) {};
\node (peve2) at (-2*\s,\w) {};
\node (peve3) at (-\s,\w) {};
\node (qeve) at (2*\s,\w) {};

\draw[dashed,sArrow] (pleft) to node[pos=.6,auto] {$\rho$} (peve1);
\draw[dashed,sArrow] (peve2) to node[pos=.45,auto,swap] {$\rho'$} (pcenter);
\draw[dashed,sArrow] (pright) to node[pos=.6,auto] {$t$} (peve3);
\draw[sArrow] (qcenter) to node[pos=.6,auto] {$y$} (qeve);

\end{tikzpicture}

\end{centering}
\caption[Secure channel]{\label{fig:ex.otp.ideal}The ideal secure channel and
  corresponding composed simulator $\sigma^{\otp}_E\sigma^{\qkd}_E$.}
\end{subfigure}

\caption[Example: serial composition of one-time pad and
QKD]{\label{fig:ex.otp}Steps in the security proof of the sequential
  composition of a one-time pad and QKD protocol. In each figure,
  Alice and Bob have access to the left and right interfaces, and Eve
  to the lower interface. If we remove the gray parts from
  Figures~\ref{fig:ex.otp.real} and \ref{fig:ex.otp.hybrid}, we
  recover the real and ideal systems of QKD. If we remove the dashed
  parts from Figures~\ref{fig:ex.otp.hybrid} and
  \ref{fig:ex.otp.ideal} we recover the real and ideal systems of the
  one-time pad.}
\end{figure}
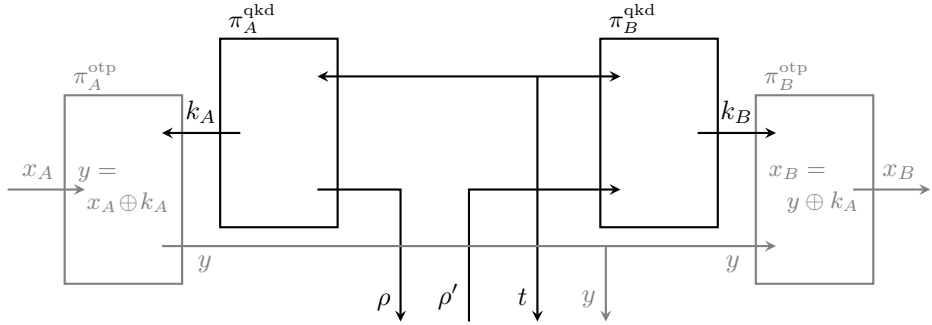
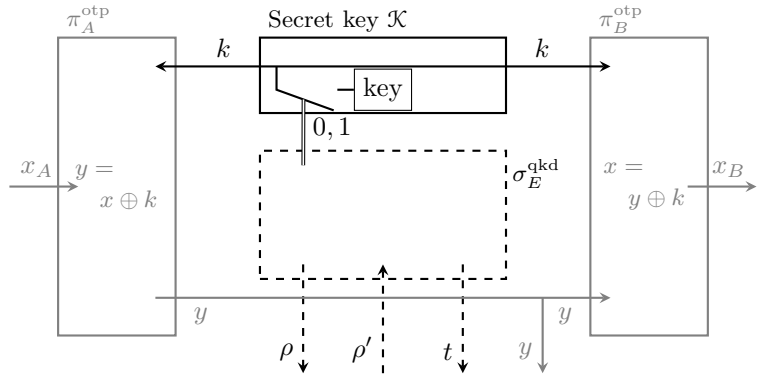
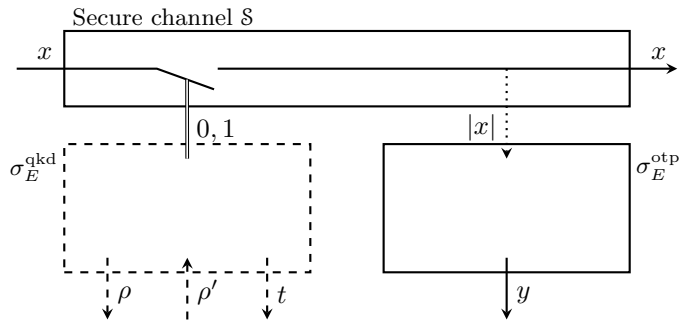

We now wish to show that the combination of an $\eps$\=/secure QKD
protocol and a (perfect) one-time pad results in a combined scheme
that constructs within $\eps$ a secure channel from authentic
classical channels\footnote{The QKD protocol requires a two-way
  authentic channel, whereas the one-time pad needs only a single use
  one-way authentic channel. This distinction is however not relevant
  to the current argument, so we refer to both resources as
  ``authentic channels'' and use the same notation, $\aA$, for each.}
and an insecure quantum channel. To do this, we look at an
intermediary step consisting of the combination of an ideal secret key
resource and a one-time pad, which we illustrate in
\figref{fig:ex.otp.hybrid}. If we remove the gray parts from
Figures~\ref{fig:ex.otp.real} and \ref{fig:ex.otp.hybrid}, we recover
the real and ideal systems of QKD. If the QKD protocol is
$\eps$\=/secure, then the distinguishing advantage between these two
figures can also be at most $\eps$.  Likewise, if we remove the dashed
parts from Figures~\ref{fig:ex.otp.hybrid} and \ref{fig:ex.otp.ideal}
we recover the real and ideal systems of the one-time pad. Since the
one-time pad is perfectly secure, the distinguishing advantage between
these two figures must be $0$. It follows from the triangle inequality
that the composition of an $\eps$\=/secure QKD protocol and a one-time
pad is $\eps$\=/secure.

This reasoning is summed up in the following equation, which can be
directly derived from \eqnsref{eq:axioms.order},
\eqref{eq:axioms.nonincrease} and the triangle inequality
(\eqnref{eq:pm.tri}):
\[ \left.\begin{aligned}
\pi^{\qkd}_A \pi^{\qkd}_B \left( \aQ \| \aA \right) & \close{\eps}
\aK \sigma^{\qkd}_E \\
\pi^{\otp}_A \pi^{\otp}_B \left( \aK \| \aA \right) & = \aS \sigma^{\otp}_E
\end{aligned}\right\}
\implies \pi^{\otp}_A \pi^{\otp}_B\pi^{\qkd}_A \pi^{\qkd}_B \left( \aQ
  \| \aA \| \aA \right) \close{\eps} \aS \sigma^{\otp}_E\sigma^{\qkd}_E \
. \]

The same can be obtained from the properties of the trace distance if
we write out explicitly the states gathered by the
distinguisher. After a run of an $\eps$\=/secure QKD scheme, we know
that
\[(1-\pabort) D(\rho_{ABE},\tau_{AB} \tensor \rho_{E}) \leq \eps \ .\]
The encryption and decryption operations of the one-time pad,
$(\pi^{\otp}_A,\pi^{\otp}_B)$, plugged into
Figures~\ref{fig:ex.otp.real} and \ref{fig:ex.otp.hybrid}, modify the
states $\rho_{AB}$ and $\tau_{AB}$. They correspond to a unitary map
$\cE^{\otp} : \cX \times \cK \times \cK \to \cX \times \cX \times \cX$
which takes the message $x_A$ and Alice's and Bob's keys $k_A$, $k_B$,
and generates the ciphertext and Bob's message while persevering
Alice's message, \[\cE^{\otp} : (x_A,k_A,k_B) \mapsto (x_A,x_A \xor
k_A \xor k_B,x_A \xor k_A) \ .\] A unitary map does not change the
trace distance, so for \[\rho_{X_AX_BY} = \cE^{\otp} (\rho_{X_A}
\tensor \rho_{AB}) \qquad \text{and} \qquad \tau_{X_AX_BY} =
\cE^{\otp} (\rho_{X_A} \tensor \tau_{AB})\] we have
\[(1-\pabort) D(\rho_{X_AX_BYE},\tau_{X_AX_BY} \tensor \rho_E) =
(1-\pabort) D(\rho_{ABE},\tau_{AB} \tensor \rho_{E}) \leq \eps \ ,\] where
$\rho_{X_AX_BYE}$ and $\tau_{X_AX_BY} \tensor \rho_E$ are the states
held be a distinguisher interacting with Figures~\ref{fig:ex.otp.real}
and \ref{fig:ex.otp.hybrid}, respectively.

We also know that the one-time pad perfectly constructs a secure
channel from an authentic channel and a secret key, i.e., if we remove
the simulator $\sigma^{\qkd}_E$ from Figures~\ref{fig:ex.otp.hybrid}
and \ref{fig:ex.otp.ideal}, the corresponding systems are
indistinguishable \--- a distinguisher interacting with them obtains
two states $\tau_{X_AX_BY}$ and $\tau'_{X_AX_BY}$ with
$D(\tau_{X_AX_BY},\tau'_{X_AX_BY}) = 0$. Plugging the simulator
$\sigma^{\qkd}_E$ in Eve's interface simply results in the state
$\rho_E$ being appended to $\tau$ and $\tau'$. The final
state held by the distinguisher is thus $\tau_{X_AX_BY} \tensor
\rho_E$ and $\tau'_{X_AX_BY} \tensor \rho_E$, respectively, with
$D(\tau_{X_AX_BY} \tensor \rho_E,\tau'_{X_AX_BY} \tensor \rho_E) = 0$.

By the triangle inequality, the distance between
Figures~\ref{fig:ex.otp.real} and \ref{fig:ex.otp.ideal} is then
\[ (1-\pabort) D( \rho_{X_AX_BYE} , \tau'_{X_AX_BY} \tensor \rho_E) \leq \eps \
. \]

\subsection{Parallel composition of key distribution with itself}
\label{sec:ex.qkd}

If two QKD protocols are run in parallel, as illustrated in
\figref{fig:ex.qkd.real} the adversary can entwine their respective
messages as she pleases, e.g, parts of the state $\rho$ sent on the
insecure channel by the first protocol can be input into the insecure
channel of the second protocol. We wish to show that even in this
case, the combined protocol is still $2\eps$\=/secure \--- i.e.,
indistinguishable from the parallel compositions of two ideal key
resources and their individual simulators \--- if each QKD protocol is
$\eps$\=/secure. This ideal case is depicted in
\figref{fig:ex.qkd.ideal}.

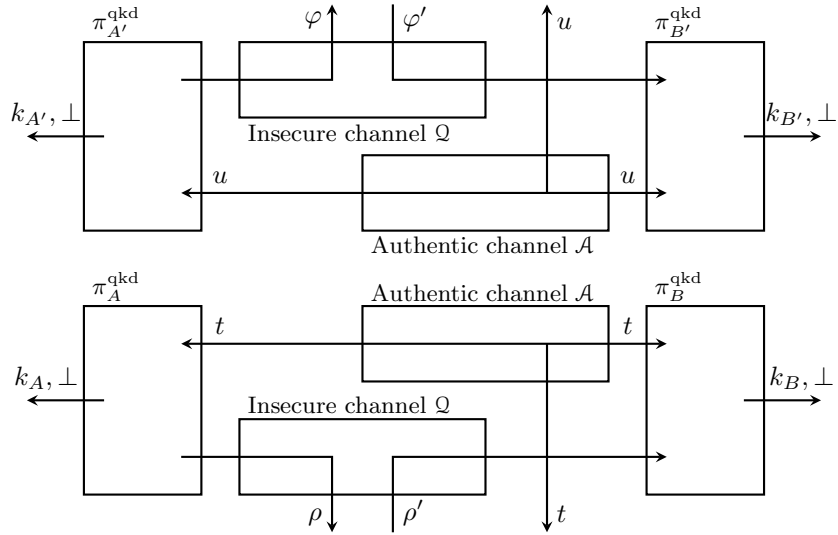
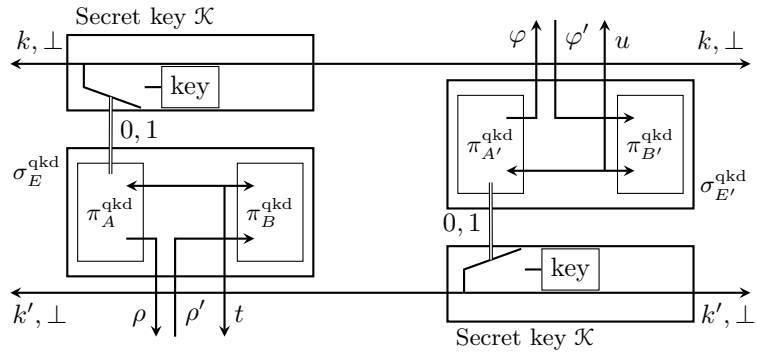
\begin{figure}[htbp]
\begin{subfigure}[b]{\textwidth}
\begin{centering}

\begin{tikzpicture}\small

\def\t{5.222} 
\def\u{3.7} 
\def\v{.75}
\def\w{.809}
\def\s{1.75cm}

\node[yshift=-\s,pnode] (a1) at (-\u,\v) {};
\node[yshift=-\s,pnode] (a2) at (-\u,0) {};
\node[yshift=-\s,pnode] (a3) at (-\u,-\v) {};
\node[yshift=-\s,protocol] (a) at (-\u,0) {};
\node[yshift=-2,above right] at (a.north west) {\footnotesize
  $\pi^{\qkd}_A$};
\node[yshift=-\s] (alice) at (-\t,0) {};

\node[yshift=-\s,pnode] (b1) at (\u,\v) {};
\node[yshift=-\s,pnode] (b2) at (\u,0) {};
\node[yshift=-\s,pnode] (b3) at (\u,-\v) {};
\node[yshift=-\s,protocol] (b) at (\u,0) {};
\node[yshift=-2,above right] at (b.north west) {\footnotesize $\pi^{\qkd}_B$};
\node[yshift=-\s] (bob) at (\t,0) {};

\node[yshift=-\s,thinResource] (cch) at (\w,\v) {};
\node[yshift=-2,above right] at (cch.north west) {\footnotesize
  Authentic channel $\aA$};
\node[yshift=-\s,thinResource] (qch) at (-\w,-\v) {};
\node[yshift=-1.5,above right] at (qch.north west) {\footnotesize
  Insecure channel $\aQ$};
\node[yshift=-\s] (eveq1) at (-\w-.4,-1.75) {};
\node (junc1) at (eveq1 |- a3) {};
\node[yshift=-\s] (eveq2) at (-\w+.4,-1.75) {};
\node (junc2) at (eveq2 |- a3) {};
\node[yshift=-\s] (evec) at (\w+\w,-1.75) {};
\node (junc3) at (evec |- b1) {};

\draw[sArrow,<->] (a1) to node[auto,pos=.08] {$t$} node[auto,pos=.92] {$t$}  (b1);
\draw[sArrow] (junc3.center) to node[auto,pos=.9] {$t$} (evec.center);

\draw[sArrow] (a2) to node[auto,pos=.75,swap] {$k_{A},\bot$} (alice.center);
\draw[sArrow] (b2) to node[auto,pos=.75] {$k_{B},\bot$} (bob.center);

\draw[sArrow] (a3) to (junc1.center) to node[pos=.8,auto,swap] {$\rho$} (eveq1.center);
\draw[sArrow] (eveq2.center) to node[pos=.264,auto,swap] {$\rho'$} (junc2.center) to (b3);

\node[yshift=\s,pnode] (a1a) at (-\u,-\v) {};
\node[yshift=\s,pnode] (a2a) at (-\u,0) {};
\node[yshift=\s,pnode] (a3a) at (-\u,\v) {};
\node[yshift=\s,protocol] (aa) at (-\u,0) {};
\node[yshift=-2,above right] at (aa.north west) {\footnotesize
  $\pi^{\qkd}_{A'}$};
\node[yshift=\s] (alicea) at (-\t,0) {};

\node[yshift=\s,pnode] (b1a) at (\u,-\v) {};
\node[yshift=\s,pnode] (b2a) at (\u,0) {};
\node[yshift=\s,pnode] (b3a) at (\u,\v) {};
\node[yshift=\s,protocol] (ba) at (\u,0) {};
\node[yshift=-2,above right] at (ba.north west) {\footnotesize $\pi^{\qkd}_{B'}$};
\node[yshift=\s] (boba) at (\t,0) {};

\node[yshift=\s,thinResource] (ccha) at (\w,-\v) {};
\node[yshift=2,below right] at (ccha.south west) {\footnotesize
  Authentic channel $\aA$};
\node[yshift=\s,thinResource] (qcha) at (-\w,\v) {};
\node[yshift=1.5,below right] at (qcha.south west) {\footnotesize
  Insecure channel $\aQ$};
\node[yshift=\s] (eveq1a) at (-\w-.4,1.75) {};
\node (junc1a) at (eveq1a |- a3a) {};
\node[yshift=\s] (eveq2a) at (-\w+.4,1.75) {};
\node (junc2a) at (eveq2a |- a3a) {};
\node[yshift=\s] (eveca) at (\w+\w,1.75) {};
\node (junc3a) at (eveca |- b1a) {};

\draw[sArrow,<->] (a1a) to node[auto,pos=.08] {$u$} node[auto,pos=.92] {$u$}  (b1a);
\draw[sArrow] (junc3a.center) to node[auto,pos=.9,swap] {$u$} (eveca.center);

\draw[sArrow] (a2a) to node[auto,pos=.75,swap] {$k_{A'},\bot$} (alicea.center);
\draw[sArrow] (b2a) to node[auto,pos=.75] {$k_{B'},\bot$} (boba.center);

\draw[sArrow] (a3a) to (junc1a.center) to node[pos=.8,auto] {$\varphi$} (eveq1a.center);
\draw[sArrow] (eveq2a.center) to node[pos=.19,auto] {$\varphi'$} (junc2a.center) to (b3a);

\end{tikzpicture}

\end{centering}
\caption[QKD protocols]{\label{fig:ex.qkd.real}Two QKD protocols and their respective
  resources run in parallel.}
\end{subfigure}

\vspace{12pt}

\begin{subfigure}[b]{\textwidth}
\begin{centering}

\begin{tikzpicture}\small

\def\t{2.368} 
\def\u{-1.85} 
\def\w{-3.5} 
\def\v{.118} 
\def\s{1.05}
\def\l{2.5cm}
\def\r{2.8cm}

\node[xshift=-\l,thinResource] (keyBox) at (0,0) {};
\node[xshift=-\l,draw] (key) at (0,\v/2-.25) {key};
\node[xshift=-\l] (junc) at (-1.4,0 |- key.center) {};
\node[yshift=-1.5,above right] at (keyBox.north west) {\footnotesize
  Secret key $\aK$};
\node[xshift=-\l] (alice) at (-\t,\v) {};
\node[xshift=\l] (bob) at (\t,\v) {};

\draw[sArrow,<->] (alice.center) to node[pos=.04,auto] {$k,\bot$} node[pos=.96,auto] {$k,\bot$} (bob.center);
\draw[thick] (junc.center |- 0,\v) to (junc.center) to node[pos=.472] (handle) {} +(160:-.8);
\draw[xshift=-\l,thick] (-.6,0 |- junc.center) to (key);

\node[xshift=-\l,simulator] (sim) at (0,\u) {};
\node[xshift=1.5,below left] at (sim.north west) {\footnotesize
  $\sigma^{\qkd}_E$};
\node[xshift=-\l,innersim,draw] (pleft) at (-\s,\u) {\footnotesize $\pi^{\qkd}_A$};
\node[xshift=-\l,innersim,draw] (pright) at (\s,\u) {\footnotesize $\pi^{\qkd}_B$};
\node[xshift=-\l,innersnode] (a1) at (-\s,\u+.35) {};
\node[xshift=-\l,innersnode] (a2) at (-\s,\u-.35) {};
\node[xshift=-\l,innersnode] (b1) at (\s,\u+.35) {};
\node[xshift=-\l,innersnode] (b2) at (\s,\u-.35) {};

\node[xshift=-\l] (evel) at (-.45,\w) {};
\node (juncl) at (evel |- a2) {};
\node[xshift=-\l] (evec) at (-.2,\w) {};
\node (juncc) at (evec |- a2) {};
\node[xshift=-\l] (ever) at (.45,\w) {};
\node (juncr) at (ever |- a1) {};

\draw[double] (a1) to node[pos=.55,auto,swap] {$0,1$} (handle.center);
\draw[sArrow,<->] (a1) to (b1);
\draw[sArrow] (juncr.center) to node[pos=.852,auto] {$t$} (ever.center);
\draw[sArrow] (a2) to (juncl.center) to node[pos=.805,auto,swap] {$\rho$} (evel.center);
\draw[sArrow] (evec.center) to node[pos=.25,auto,swap] {$\rho'$} (juncc.center) to (b2);


\node[xshift=\l,yshift=-\r,thinResource] (keyBoxa) at (0,0) {};
\node[xshift=\l,yshift=-\r,draw] (keya) at (0,-\v/2+.25) {key};
\node[xshift=\l] (junca) at (-1.4,0 |- keya.center) {};
\node[yshift=1.5,below right] at (keyBoxa.south west) {\footnotesize
  Secret key $\aK$};
\node[xshift=-\l,yshift=-\r] (alicea) at (-\t,-\v) {};
\node[xshift=\l,yshift=-\r] (boba) at (\t,-\v) {};

\draw[sArrow,<->] (alicea.center) to node[pos=.04,auto,swap] {$k',\bot$} node[pos=.97,auto,swap] {$k',\bot$} (boba.center);
\draw[thick,yshift=-\r] (junca.center |- 0,-\v) to (junca.center) to node[pos=.472] (handlea) {} +(200:-.8);
\draw[xshift=\l,thick] (-.6,0 |- junca.center) to (keya);

\node[xshift=\l,yshift=-\r,simulator] (sima) at (0,-1*\u) {};
\node[xshift=-1.5,above right] at (sima.south east) {\footnotesize
  $\sigma^{\qkd}_{E'}$};
\node[xshift=\l,yshift=-\r,innersim,draw] (plefta) at (-\s,-1*\u) {\footnotesize $\pi^{\qkd}_{A'}$};
\node[xshift=\l,yshift=-\r,innersim,draw] (prighta) at (\s,-1*\u) {\footnotesize $\pi^{\qkd}_{B'}$};
\node[xshift=\l,yshift=-\r,innersnode] (a1a) at (-\s,-1*\u-.35) {};
\node[xshift=\l,yshift=-\r,innersnode] (a2a) at (-\s,-1*\u+.35) {};
\node[xshift=\l,yshift=-\r,innersnode] (b1a) at (\s,-1*\u-.35) {};
\node[xshift=\l,yshift=-\r,innersnode] (b2a) at (\s,-1*\u+.35) {};

\node[xshift=\l,yshift=-\r] (evela) at (-.45,-1*\w) {};
\node (juncla) at (evela |- a2a) {};
\node[xshift=\l,yshift=-\r] (eveca) at (-.2,-1*\w) {};
\node (juncca) at (eveca |- a2a) {};
\node[xshift=\l,yshift=-\r] (evera) at (.45,-1*\w) {};
\node (juncra) at (evera |- a1a) {};

\draw[double] (a1a) to node[pos=.55,auto,swap] {$0,1$} (handlea.center);
\draw[sArrow,<->] (a1a) to (b1a);
\draw[sArrow] (juncra.center) to node[pos=.852,auto,swap] {$u$} (evera.center);
\draw[sArrow] (a2a) to (juncla.center) to node[pos=.805,auto] {$\varphi$} (evela.center);
\draw[sArrow] (eveca.center) to node[pos=.17,auto] {$\varphi'$} (juncca.center) to (b2a);

\end{tikzpicture}

\end{centering}
\caption[Ideal keys]{\label{fig:ex.qkd.ideal}Two secret key resources and two
  simulators run in parallel.}
\end{subfigure}

\caption[Example: parallel composition of QKD]{\label{fig:ex.qkd}Real
  and ideal systems for two QKD protocols executed in parallel.  In
  each figure, Alice and Bob have access to the left and right
  interfaces, and Eve to both the upper and lower interface.}
\end{figure}

Like for serial composition, this follows from the triangle
inequality. If the real QKD system is $\eps$\=/close to the ideal QKD
system, then two real QKD systems in parallel must be $\eps$\=/close
to an ideal and real QKD system composed in parallel, since otherwise
a distinguisher could run a real QKD system internally in parallel to
the system it is testing. Likewise, a real and ideal QKD system in
parallel must be $\eps$\=/close to two ideal QKD systems in
parallel. And hence two parallel runs of an $\eps$\=/secure QKD
protocol is $2\eps$\=/secure.

The trace distance notation does not lend itself to writing up
parallel composition of protocols. So instead of using this notation
as in the previous examples, we write up the reasoning from the
paragraph above in more detail using the resource\-/converter
formalism. If the real and ideal system of a QKD protocol are
$\eps$\=/close, then
\[\pi^{\qkd}_A \pi^{\qkd}_B \left(
  \aQ \| \aA \right) \close{\eps} \aK \sigma^{\qkd}_E \ .\] It follows
immediately from this and \eqnref{eq:axioms.nonincrease} that
\begin{align*}
  \left(\pi^{\qkd}_A \pi^{\qkd}_B \left( \aQ \| \aA \right) \right)
  \Big\| \left(\pi^{\qkd}_{A'} \pi^{\qkd}_{B'} \left( \aQ \| \aA
    \right) \right) & \close{\eps} \left( \aK \sigma^{\qkd}_E \right)
  \Big\| \left(\pi^{\qkd}_{A'} \pi^{\qkd}_{B'} \left( \aQ \| \aA
    \right)
  \right) \ , \\
  \left( \aK \sigma^{\qkd}_E \right) \Big\| \left(\pi^{\qkd}_{A'}
    \pi^{\qkd}_{B'} \left( \aQ \| \aA \right) \right) & \close{\eps}
  \left( \aK \sigma^{\qkd}_E \right) \Big\| \left( \aK
    \sigma^{\qkd}_{E'} \right) \ .
\end{align*}
From the triangle equality (\eqnref{eq:pm.tri}) we then have
\[ \left(\pi^{\qkd}_A \pi^{\qkd}_B \left( \aQ \| \aA \right) \right)
\Big\| \left(\pi^{\qkd}_{A'} \pi^{\qkd}_{B'} \left( \aQ \| \aA \right)
\right) \close{2\eps} \left( \aK \sigma^{\qkd}_E \right) \Big\| \left(
  \aK \sigma^{\qkd}_{E'} \right) \ .\] Finally, using
\eqnref{eq:axioms.order} to rearrange this expression, we get
\[ \left(\pi^{\qkd}_A \pi^{\qkd}_B \big\| \pi^{\qkd}_{A'}
  \pi^{\qkd}_{B'} \right)\left( \aQ \| \aA \| \aQ \| \aA \right)
\close{2\eps} \left(\aK \| \aK \right) \left(\sigma^{\qkd}_E \big\|
  \sigma^{\qkd}_{E'} \right) \ ,\] i.e., the parallel composition of
two runs of a QKD protocol, $\pi^{\qkd}_A \pi^{\qkd}_B \big\|
\pi^{\qkd}_{A'} \pi^{\qkd}_{B'}$, run with authentic classical and
insecure quantum channel resources, $\aQ \| \aA \| \aQ \| \aA$, is
$2\eps$\=/close to the parallel composition of two ideal key
resources, $\aK \| \aK$, and a simulator, $\sigma^{\qkd}_E \big\|
\sigma^{\qkd}_{E'}$.

\appendix
\appendixpage
\phantomsection
\label{app}
\addcontentsline{toc}{section}{Appendices}

In \appendixref{app:op} we formally define the trace distance and show
that it corresponds to the distinguishing advantage between two
quantum states. We also prove several lemmas that help interpret its
meaning and how to choose a value in a practical
implementation. In \appendixref{app:alternative} we discuss an
alternative to the secrecy criterion of \eqnref{eq:d}, which has
appeared in the literature. In \appendixref{app:moreAC} we provide
some details on technical aspects of the Abstract Cryptography
framework. In \appendixref{app:moreAC.dist} we discuss how to define a
distinguisher so that the resulting distinguishing advantage is
non\-/increasing under compositions. In
\appendixref{app:generic} we sketch a proof that the security
definition from \defref{def:security} is composable. A complete proof
of this can be found in \cite{MR11,Mau12}. And finally,
in \appendixref{app:ex.auth} we model the security of authentication
with universal hashing~\cite{WC81,Sti94}, then use this as a subprotocol
of QKD to authenticate the classical post\-/processing. Since this type
of authentication uses a short key (and an insecure classical channel)
to construct an authentic channel and QKD uses an authentic channel
(and an insecure quantum channel) to construct a long key, the
composition of the two is a key expansion protocol, which constructs a
long key from a short key (and insecure channels).


\section{Trace distance}
\label{app:op}

We have used several times in this work the well\-/known fact that the
distinguishing advantage between two systems that output states $\rho$
and $\sigma$ is equivalent to the trace distance between these
states. In this appendix, we prove this fact, along with several other
theorems that help interpret the meaning of the trace distance.

In \appendixref{app:op.definitions} we first define the trace distance
\--- as well as its classical counterpart, the total variation
distance\--- and prove some basic lemmas that can also be found in
textbooks such as \cite{NC00}. In \appendixref{app:op.distadv} we then
show the connection between trace distance and distinguishing
advantage, which was originally proven by
Helstrom~\cite{Hel76}. In \appendixref{app:op.failure} we prove that
we can alternatively think of the trace distance between a real and
ideal system as a bound on the probability that a failure occurs in
the real system. Finally, in \appendixref{app:op.local} we bound two
typical information theory notions of secrecy \--- the conditional
entropy of a key given the eavesdropper's information and her
probability of correctly guessing the key \--- in terms of the trace
distance. Although such measures of information are generally
ill-suited for defining cryptographic security, they can help
interpret the notion of a key being $\eps$\-/close to uniform.

\subsection{Metric definitions}
\label{app:op.definitions}

In the case of a classical system, statistical security is defined by
the total variation (or statistical) distance between the probability
distributions describing the real and ideal settings, which is defined
as follows.\footnote{We employ the same notation $D(\cdot,\cdot)$ for
  both the total variation and trace distance, since the former is a
  special case of the latter.}

\begin{deff}[total variation distance]
  \label{def:vdist}
  The total variation distance between two probability distributions
  $P_Z$ and $P_{\tilde{Z}}$ over an alphabet $\cZ$ is defined as
  \[D(P_Z,P_{\tilde{Z}}) \coloneqq \frac{1}{2} \sum_{z \in \cZ} \left| P_Z(z) -
    P_{\tilde{Z}}(z) \right| \ .\]
\end{deff}

Using the fact that $|a-b| = a+b-2 \min(a,b)$, the total variation
distance can also be written as
\begin{equation} \label{eq:vdist.alt} D(P_Z, P_{\tilde{Z}}) = 1-
  \sum_{z \in \cZ} \min[P_Z(z), P_{\tilde{Z}}(z)] \ .\end{equation}

In the case of quantum states instead of classical random variables,
the total variation distance generalizes to the trace distance. More
precisely, the trace distance between two density operators that are
diagonal in the same orthonormal basis is equal to the total variation
distance between the probability distributions defined by their
respective eigenvalues.

\begin{deff}[trace distance]
  \label{def:tdist}
  The trace distance between two quantum states $\rho$ and $\sigma$ is
  defined as
  \[D(\rho,\sigma) \coloneqq \frac{1}{2} \tr |\rho-\sigma| \ .\]
\end{deff}

We now introduce some technical lemmas involving the trace distance,
which help us derive the theorems in the next sections. Most of these
proofs are taken from \cite{NC00}.

\begin{lem}
  \label{lem:op.definitions.pos}
  For any two states $\rho$ and $\sigma$ and any operator $0 \leq
  M \leq I$, the two following inequalities hold:
  \begin{align}
    D(\rho,\sigma) \geq \trace{M(\rho-\sigma)} \
    , \label{eq:op.definitions.pos} \\
    \trace{M |\rho-\sigma|} \geq \left| \trace{M (\rho-\sigma)}
    \right| \ . \label{eq:op.definitions.abs}
\end{align} Furthermore, each of these inequalities is tight for some
values of $M$.
\end{lem}

The trace distance can thus alternatively be written
as \begin{equation} \label{eq:op.definitions.pos.d}
D(\rho,\sigma) = \max_M \trace{M(\rho-\sigma)} \ .\end{equation}

\begin{proof}
  We start with the proof of \eqnref{eq:op.definitions.pos}. Let
  $\{\lambda_x,\ket{\psi_x}\}_x$ be the eigenvalues and vectors of
  $\rho-\sigma$, and define \[Q_+ \coloneqq \sum_{x : \lambda_x \geq
    0} \lambda_x \proj{\psi_x} \qquad \text{and} \qquad Q_- \coloneqq
  \sum_{x : \lambda_x < 0} -\lambda_x \proj{\psi_x} \ .\] We have
  $\rho-\sigma = Q_+-Q_-$ and $|\rho-\sigma| = Q_++Q_-$. Note that
  since $\trace{Q_+-Q_-} = \trace{\rho-\sigma} = 0$, we have $\tr Q_+
  = \tr Q_-$, hence \[D(\rho,\sigma) = \frac{1}{2} \tr |\rho-\sigma|
  =\frac{1}{2}(\tr Q_+ + \tr Q_-) = \tr Q_+ \ .\] If we set $\Gamma_+
  \coloneqq \sum_{x : \lambda_x \geq 0} \proj{\psi_x}$, the projector
  on $Q_+$, we get \[\trace{\Gamma_+(\rho-\sigma)} = \tr Q_+ =
  D(\rho,\sigma) \ .\] And for any operator $0 \leq M \leq I$,
\begin{equation*} 
  \trace{M (\rho-\sigma)} = \trace{M (Q_+-Q_-)} \leq \trace{M Q_+} \leq \tr Q_+ = D(\rho,\sigma) \ . 
\end{equation*}

To prove that \eqnref{eq:op.definitions.abs} holds, note that for any
operator $0 \leq M \leq I$,
  \begin{align*}
    \left| \trace{M(\rho-\sigma)} \right| & = \left|
        \trace{M(Q_+-Q_-)} \right| \\
    & \leq \trace{M(Q_++Q_-)}  = \trace{M|\rho-\sigma|} \ .
  \end{align*}
  \eqnref{eq:op.definitions.abs} is tight for any operator $M$ which
  satisfies \[\left|\trace{M(Q_+-Q_-)} \right| = \trace{M(Q_++Q_-)} \
  , \]  i.e., any operator such that either $0 \leq M \leq \Gamma_+$ or $0
  \leq M \leq \Gamma_-$, where $\Gamma_+$ is defined as above and
  $\Gamma_- \coloneqq \sum_{x : \lambda_x \leq 0} \proj{\psi_x}$.
\end{proof}

Let $\{\Gamma_x\}_{x}$ be a positive operator\-/valued measure (POVM)
\--- a set of operators $0 \leq \Gamma_x \leq I$ such that
$\sum_x \Gamma_x = I$ \--- and let $P_X$ denote the outcome of
measuring a quantum state $\rho$ with $\{\Gamma_x\}_{x}$, i.e.,
$P_X(x) = \trace{\Gamma_x \rho}$. Our next lemma says that the trace
distance between two states $\rho$ and $\sigma$ is equal to the total
variation between the outcomes \--- $P_X$ and $Q_X$ \--- of an optimal
measurement on the two states.

\begin{lem}
  \label{lem:op.definitions.opt}
  For any two states $\rho$ and $\sigma$,
  \begin{equation} \label{eq:op.definitions.opt} D(\rho,\sigma) =
    \max_{\left\{\Gamma_x\right\}_x} D(P_X,Q_X) \ ,\end{equation}
  where $P_X$ and $Q_X$ are the probability distributions resulting
  from measuring $\rho$ and $\sigma$ with a POVM $\{\Gamma_x\}_x$,
  respectively, and the maximization is over all POVMs. Furthermore,
  if the two states $\rho_{ZB}$ and $\sigma_{ZB}$ have a classical
  subsystem $Z$, then the measurement satisfying
  \eqnref{eq:op.definitions.opt} leaves the classical subsystem
  unchanged, i.e., the maximum is reached for a POVM with
  elements \begin{equation} \label{eq:op.definitions.opt.classical}
    \Gamma_x = \sum_z \proj{z} \tensor M^z_x \ ,\end{equation} where
  $\{\ket{z}\}_z$ is the classical orthonormal basis of $Z$.
\end{lem}

\begin{proof}
  Using \eqnref{eq:op.definitions.abs} from
  \lemref{lem:op.definitions.pos} we get
  \begin{align*}D(P_X,Q_X) & = \frac{1}{2} \sum_x \left|
      \trace{\Gamma_x(\rho-\sigma)} \right| \\
      & \leq \frac{1}{2} \sum_x \trace{\Gamma_x|\rho-\sigma|} \\
      & = \frac{1}{2} \tr{|\rho-\sigma|} = D(\rho,\sigma)\ .
  \end{align*}
  The conditions for equality are given at the end of the proof of
  \lemref{lem:op.definitions.pos}, e.g., a measurement with $\Gamma_x
  = \proj{\psi_x}$, where $\{\ket{\psi_x}\}_x$ are the eigenvectors of
  $\rho-\sigma$. If $\rho_{ZB} = \sum_z p_z \proj{z} \tensor \rho_B^z$
  and $\sigma_{ZB} = \sum_z q_z \proj{z} \tensor \sigma_B^z$, then \[
  \rho-\sigma = \sum_z \proj{z} \tensor
  \left(p_z\rho^z_B-q_z\sigma^z_B\right) \ ,\] and the eigenvectors of
  $\rho-\sigma$ have the form $\ket{\psi_{z,x}} = \ket{z} \tensor
  \ket{\varphi^z_x}_B$, where $\ket{\varphi^z_x}_B$ is an eigenvector
  of $p_z\rho^z_B-q_z\sigma^z_B$. So the optimal measurement,
  $\Gamma_{z,x} = \proj{z} \tensor \proj{\varphi^z_x}_B$, satisfies
  \eqnref{eq:op.definitions.opt.classical}.
\end{proof}

\subsection{Distinguishing advantage}
\label{app:op.distadv}

Helstrom~\cite{Hel76} proved that the advantage a distinguisher has in
guessing whether it was provided with one of two states, $\rho$ or
$\sigma$, is given by the trace distance between the two,
$D(\rho,\sigma)$.\footnote{Actually, Helstrom~\cite{Hel76} solved a
  more general problem, in which the states $\rho$ and $\sigma$ are
  picked with apriori probabilities $p$ and $1-p$, respectively,
  instead of $1/2$ as in the definition of the distinguishing
  advantage.} We first sketch the classical case, then prove the
quantum version.

Let a distinguisher be given a value sampled according to probability
distributions $P_Z$ or $P_{\tilde{Z}}$, where $P_Z$ and
$P_{\tilde{Z}}$ are each chosen with probability $1/2$. Suppose the
value received by the distinguisher is $z \in \cZ$. If $P_Z(z) >
P_{\tilde{Z}}(z)$, its best guess is that the value was sampled
according to $P_Z$. Otherwise, it should guess that it was
$P_{\tilde{Z}}$. Let $\cZ'\coloneqq \{z \in \cZ : P_Z(z) >
P_{\tilde{Z}}(z)\}$ and $\cZ'' \coloneqq \{z \in \cZ : P_Z(z) \leq
P_{\tilde{Z}}(z)\}$. There are a total of $2|\cZ|$ possible events:
the sample is chosen according to $P_Z$ or $P_{\tilde{Z}}$ and takes
the value $z \in \cZ$. These events have probabilities
$\frac{P_Z(z)}{2}$ and $\frac{P_{\tilde{Z}}(z)}{2}$. Conditioned on
$P_Z$ being chosen and $z$ being the sampled value, the distinguisher
has probability $1$ of guessing correctly with the strategy outlined
above if $z \in \cZ'$, and $0$ otherwise. Likewise, if $P_{\tilde{Z}}$
was selected, it has probability $1$ of guessing correctly if $z \in
\cZ''$ and $0$ otherwise. The probability of correctly guessing
whether it was given a value sampled according to $P_Z$ or
$P_{\tilde{Z}}$, which we denote $\distinguish{P_Z,P_{\tilde{Z}}}$, is
obtained by summing over all possible events weighted by their
probabilities. Hence
\begin{align*}
  \distinguish{P_Z,P_{\tilde{Z}}} & = \sum_{z \in \cZ'} \frac{P_Z(z)}{2} + \sum_{z \in
    \cZ''} \frac{P_{\tilde{Z}}(z)}{2} \\
  & =  \frac{1}{2} \left(1- \sum_{z \in
    \cZ''} P_{Z}(z) \right) + \frac{1}{2} \left(1 - \sum_{z \in
    \cZ'} P_{\tilde{Z}}(z)\right) \\
  & = 1 - \frac{1}{2} \sum_{z \in \cZ} \min[P_Z(z), P_{\tilde{Z}}(z)] \\
  & = \frac{1}{2} + \frac{1}{2} D(P_Z,P_{\tilde{Z}}) \ , 
\end{align*}
where in the last equality we used the alternative formulation of the
total variation distance from \eqnref{eq:vdist.alt}.

We now generalize the argument above to quantum states.

\begin{thm}
  \label{thm:op.distinguishing}
  For any states $\rho$ and $\sigma$, we have \[\distinguish{\rho,\sigma} =
  \frac{1}{2}+\frac{1}{2}D(\rho,\sigma)  \ .\]
\end{thm}

\begin{proof}
  If a distinguisher is given one of two states $\rho$ or $\sigma$,
  each with probability $1/2$, its probability of guessing which one
  it holds is given by a maximization of all possible measurements it
  may do: it chooses some POVM $\{\Gamma_0,\Gamma_1\}$, where
  $\Gamma_0$ and $\Gamma_1$ are positive operators with
  $\Gamma_0+\Gamma_1 = I$, and measures the state it holds. If it gets
  the outcome $0$, it guesses that it holds $\rho$ and if it gets the
  outcome $1$, it guesses that it holds $\sigma$. The probability of
  guessing correctly is given by
\begin{align}
  \distinguish{\rho,\sigma} & = \max_{\Gamma_0,\Gamma_1} \left[
    \frac{1}{2}\trace{\Gamma_0\rho} +
    \frac{1}{2}\trace{\Gamma_1\sigma} \right] \notag \\
  & = \frac{1}{2}\max_{\Gamma_0} \left[ \trace{\Gamma_0\rho} +
    \trace{(I-\Gamma_0)\sigma} \right] \notag \\
  & = \frac{1}{2} + \frac{1}{2} \max_{\Gamma_0}
  \trace{\Gamma_0(\rho-\sigma)} \ . \label{eq:op.distinguishing.thm}
\end{align}
The proof concludes by plugging \eqnref{eq:op.definitions.pos.d} in
\eqnref{eq:op.distinguishing.thm}.
\end{proof}

\subsection{Probability of a failure}
\label{app:op.failure}

The trace distance is used as the security definition of QKD, because
the relevant measure for cryptographic security is the distinguishing
advantage (as discussed in \secref{sec:ac}), and as proven in
\thmref{thm:op.distinguishing}, the distinguishing advantage between
two quantum states corresponds to their trace distance. This
operational interpretation of the trace distance involves two worlds,
an ideal one and a real one, and the distance measure is the
(renormalized) difference between the probabilities of the
distinguisher correctly guessing to which world it is connected.

In this section we describe a different interpretation of the total
variation and trace distances. Instead of having two different worlds,
we consider one world in which the outcomes of interacting with the
real and ideal systems co-exist. And instead of these distance
measures being a difference between probability distributions, they
become the probability that any (classical) value occurring in one of
the systems does not \emph{simultaneously} occur in the other. We call
such an event a \emph{failure} \--- since one system is ideal, if the
other behaves differently, it must have failed \--- and the trace
distance becomes the probability of a failure occurring.

Given two random variables $Z$ and $\tilde{Z}$ with probability
distributions $P_Z$ and $P_{\tilde{Z}}$, any distribution
$P_{Z\tilde{Z}}$ with marginals given by $P_Z$ and $P_{\tilde{Z}}$ is
called a coupling of $P_Z$ and $P_{\tilde{Z}}$. The interpretation of
the trace distance treated in this section uses one specific coupling,
known as a \emph{maximal coupling} in probability theory~\cite{Tho00}.

\begin{thm}[maximal coupling] \label{thm:difference}
  Let $P_Z$ and $P_{\tilde{Z}}$ be two probability distributions over
  the same alphabet $\cZ$. Then there exists a probability
  distribution $P_{Z \tilde{Z}}$ on $\cZ \times \cZ$ such that
  \begin{equation} \label{eq:equalityprob}
    \Pr[Z  = \tilde{Z}] := \sum_{z} P_{Z \tilde{Z}}(z, z) \geq 1-  D(P_Z, P_{\tilde{Z}}) 
  \end{equation}
  and such that $P_Z$ and $P_{\tilde{Z}}$ are the marginals of $P_{Z
    \tilde{Z}}$, i.e., 
  \begin{align} \label{eq:marginal1}
    P_Z(z) & = \sum_{\tilde{z}} P_{Z \tilde{Z}}(z,\tilde{z})  \quad
    (\forall z \in \cZ)
    \\ \label{eq:marginal2}
    P_{\tilde{Z}}(\tilde{z}) & = \sum_z P_{Z \tilde{Z}}(z,\tilde{z})
    \quad (\forall \tilde{z} \in \cZ)\ .
  \end{align}
\end{thm}

It turns out that the inequality in \eqnref{eq:equalityprob} is tight,
i.e., one can also show that for any distribution $P_{Z\tilde{Z}}$,
$\Pr[Z = \tilde{Z}] \leq 1 - D(P_Z, P_{\tilde{Z}})$. We will however
not use this fact here.

Consider now a real system that outputs values given by $Z$ and an
ideal system that outputs values according to
$\tilde{Z}$. \thmref{thm:difference} tells us that there exists a
coupling of these distributions such that the probability of the real
system producing a different value from the ideal system is bounded by
the total variation distance between $P_Z$ and $P_{\tilde{Z}}$. Thus,
the real system behaves ideally except with probability $D(P_Z,
P_{\tilde{Z}})$.

We first prove this theorem, then in \corref{cor:difference} here
below we apply it to quantum systems.

\begin{proof}[Proof of \thmref{thm:difference}]
  Let $Q_{Z \tilde{Z}}$ be the real function on $\cZ \times \cZ$
  defined by
  \begin{align*}
    Q_{Z \tilde{Z}}(z,\tilde{z}) = \begin{cases} \min[P_Z(z),
      P_{\tilde{Z}}(\tilde{z})] & \text{if $z=\tilde{z}$} \\ 0 & \text{otherwise}\end{cases} 
  \end{align*}
  (for all $z, \tilde{z} \in \cZ$).  Furthermore, let $R_Z$ and
  $R_{\tilde{Z}}$ be the real functions on $\cZ$ defined by
  \begin{align*}
    R_Z(z) & = P_Z(z) - Q_{Z \tilde{Z}}(z, z) \ , \\
    R_{\tilde{Z}}(\tilde{z}) & = P_{\tilde{Z}}(\tilde{z}) - Q_{Z
        \tilde{Z}}(\tilde{z}, \tilde{z})  \ .
  \end{align*}
  We then define $P_{Z \tilde{Z}}$ by
  \begin{align*}
    P_{Z \tilde{Z}}(z, \tilde{z}) = Q_{Z \tilde{Z}}(z, \tilde{z}) +
    \frac{1}{D(P_Z, P_{\tilde{Z}})} R_Z(z) R_{\tilde{Z}}(\tilde{z}) \ .
  \end{align*}
   
  We now show that $P_{Z \tilde{Z}}$ satisfies the conditions of the
  theorem. For this, we note that for any $z \in \cZ$
  \begin{align*}
    R_Z(z) = P_Z(z) - \min[P_Z(z), P_{\tilde{Z}}(z)] \geq 0 \
    , 
  \end{align*}
  i.e., $R_Z$, and, likewise, $R_{\tilde{Z}}$, are nonnegative. Since
  $Q_{Z \tilde{Z}}$ is by definition also nonnegative, we have that
  $P_{Z \tilde{Z}}$ is nonnegative, too. From \eqnref{eq:marginal1}
  or~\eqref{eq:marginal2}, which we will prove below, it follows that
  $P_{Z \tilde{Z}}$ is also normalized. Hence, $P_{Z \tilde{Z}}$ is a
  valid probability distribution.

  To show \eqnref{eq:equalityprob} we use again the non\-/negativity
  of $R_{Z}$ and $R_{\tilde{Z}}$, which implies
  \begin{align*}
    \sum_z P_{Z \tilde{Z}}(z,z) \geq \sum_z Q_{Z \tilde{Z}}(z,z) =
    \sum_z \min[P_Z(z), P_{\tilde{Z}}(z)] = 1- D(P_Z, P_{\tilde{Z}}) \ ,
  \end{align*}
  where in the last equality we used the alternative formulation of the
  total variation distance from \eqnref{eq:vdist.alt}.

  To prove \eqnref{eq:marginal1}, we first note that
  \begin{align*}
    \sum_{\tilde{z}} R_{\tilde{Z}}(\tilde{z}) & = \sum_{\tilde{z}}
    P_{\tilde{Z}}(\tilde{z}) - \sum_{\tilde{z}}  Q_{Z
      \tilde{Z}}(\tilde{z},\tilde{z})  \\ & =
    1 - \sum_{\tilde{z}} \min[P_Z(\tilde{z}),
    P_{\tilde{Z}}(\tilde{z})] = D(P_Z,  P_{\tilde{Z}}) \ .
  \end{align*}
  Using this we find that for any $z \in \cZ$
  \begin{align*}
    \sum_{\tilde{z}} P_{Z \tilde{Z}}(z,\tilde{z})
  & =
    \sum_{\tilde{z}} Q_{Z \tilde{Z}}(z,\tilde{z}) + R_Z(z) \frac{1}{D(P_Z, P_{\tilde{Z}})}
    \sum_{\tilde{z}}  R_{\tilde{Z}}(\tilde{z})  \\
  & = Q_{Z \tilde{Z}}(z,z) +R_Z(z) = P_Z(z)\ .
  \end{align*}  
  By symmetry, this also proves \eqnref{eq:marginal2}.
\end{proof}

In the case of quantum states, \thmref{thm:difference} can be used to
couple the outcomes of any observable applied to the quantum systems.

\begin{cor}
\label{cor:difference}
For any states $\rho$ and $\sigma$ with trace distance $D(\rho,\sigma)
\leq \eps$, and any measurement given by its POVM operators
$\{\Gamma_w\}_w$ with outcome probabilities $P_W(w) = \trace{\Gamma_w
  \rho}$ and $P_{\tilde{W}}(w) = \trace{\Gamma_w \sigma}$, there
exists a coupling of $P_W$ and $P_{\tilde{W}}$ such that \[ \Pr[W \neq
\tilde{W}] \leq D(\rho,\sigma) \ .\]
\end{cor}

\begin{proof}
  Immediate by combining \lemref{lem:op.definitions.opt} and
  \thmref{thm:difference}.
\end{proof}

\corref{cor:difference} tells us that if two systems produce states
$\rho$ and $\sigma$, then for any observations made on those systems
there exists a coupling for which the values of each measurement will
differ with probability at most $D(\rho,\sigma)$. It is instructive to
remember that this operational meaning is not essential to the
security notion or part of the framework in any way. It is an
intuitive way of understanding the trace distance, so as to better
choose a suitable value. It allows this distance to be thought of as a
maximum failure probability, and the value for $\eps$ to be chosen
accordingly.

\subsection{Measures of uncertainty}
\label{app:op.local}

Non\-/composable security models often use measures of uncertainty to
quantify how much information an adversary might have about a secret,
e.g., entropy as used by Shannon to prove the security of the one-time
pad~\cite{Sha49}. These measures are often weaker than what one
obtains using a global distinguisher, and in general do not provide
good security definitions. They are however quite intuitive and in
order to further illustrate the quantitative value of the
distinguishing advantage, we derive bounds on two of these measures of
uncertainty in terms of the trace distance, namely on the probability
of guessing the secret key in
\appendixref{app:op.local.guessing} and on the von Neumann entropy of the
secret key in \appendixref{app:op.local.entropy}.

\subsubsection{Probability of guessing}
\label{app:op.local.guessing}

Let $\rho_{KE} = \sum_{k \in \cK} p_k \proj{k}_K \tensor \rho^k_E$ be
the joint state of a secret key in the $K$ subsystem and Eve's
information in the $E$ subsystem. To guess the value of the key, Eve
can pick a POVM $\{\Gamma_k\}_{k \in \cK}$, measure her system, and
output the result of the measurement. Given that the key is $k$, her
probability of having guessed correctly is $\trace{\Gamma_k
  \rho^k_E}$. The average probability of guessing correctly for this
measurement is then given by the sum over all $k$, weighted by their
respective probabilities $p_k$. And Eve's probability of correctly
guessing the key is defined by taking the maximum over all
measurements,
\begin{equation} \label{eq:pguess} \pguess[\rho]{K|E} \coloneqq
  \max_{\{\Gamma_k\}} \sum_{k \in \cK} p_k \trace{\Gamma_k \rho^k_E} \
  . \end{equation}

\begin{lem}
  \label{lem:pguess}
For any bipartite state $\rho_{KE}$ with classical $K$, 
\[\pguess[\rho]{K|E} \leq
\frac{1}{|\cK|} + D\left( \rho_{KE},\tau_K \tensor \rho_E\right) \
, \] where $\tau_K$ is the fully mixed state.
\end{lem}

\begin{proof}
  Note that for $M \coloneqq \sum_k \proj{k} \tensor \Gamma_k$, where
  $\{\Gamma_k\}$ maximizes \eqnref{eq:pguess}, the guessing
  probability can equivalently be written \[\pguess[\rho]{K|E} =
  \trace{M \rho_{KE}} \ . \] Furthermore,
\[\ktrace{M \left(\tau_K \tensor \rho_{E}\right)} = \frac{1}{|\cK|} \
. \] In \lemref{lem:op.definitions.pos} we proved that for any
operator $0 \leq M \leq I$,
\[\trace{M(\rho - \sigma)} \leq D(\rho,\sigma) \ . \]
Setting $\rho = \rho_{KE}$ and $\sigma = \tau_K \tensor \rho_{E}$ in
the above inequality, we finish the proof:
\begin{align*}
\trace{M\rho_{KE}} & \leq \trace{M\left(\tau_K \tensor \rho_{E}\right)}
+ D\left( \rho_{KE} , \tau_K \tensor \rho_{E} \right) \ , \\
 \implies \qquad \pguess[\rho]{K|E} & \leq \frac{1}{|\cK|} + D\left( \rho_{KE},\tau_K
   \tensor \rho_E\right) \ . \qedhere
\end{align*}
\end{proof}

\subsubsection{Entropy}
\label{app:op.local.entropy}

Let $\rho_{KE} = \sum_{k \in \cK} p_k \proj{k}_K \tensor \rho^k_E$ be
the joint state of a secret key in the $K$ subsystem and Eve's
information in the $E$ subsystem. We wish to bound the von Neumann
entropy of $K$ given $E$ \--- $\Ss(K|E)_{\rho} = \Ss(\rho_{KE}) -
\Ss(\rho_E)$, where $S(\rho) \coloneqq -\trace{\rho \log \rho}$ \---
in terms of the trace distance $D\left( \rho_{KE},\tau_K \tensor
  \rho_E\right)$. We first derive a lower bound on the von Neumann
entropy, using the following theorem from Alicki and
Fannes~\cite{AF04}.

\begin{thm}[From~\cite{AF04}]
  \label{thm:AF04}
  For any bipartite states $\rho_{AB}$ and $\sigma_{AB}$ with trace
  distance $D(\rho,\sigma) = \eps \leq 1/4$ and $\dim \hilbert_A =
  d_A$, we have \[ \left| \Ss(A|B)_\rho - \Ss(A|B)_\sigma \right| \leq
  8 \eps \log d_A + 2h(2\eps) \ , \] where $h(p) = - p \log p - (1-p)
  \log (1-p)$ is the binary entropy.
\end{thm}

\begin{cor}
  \label{cor:AF04}
 For any state $\rho_{KE}$ with $D(\rho_{KE},\tau_K \tensor
  \rho_E) = \eps \leq 1/4$, where $\tau_K$ is the fully mixed state,
  we have
  \[ \Ss(K|E)_{\rho} \geq (1- 8\eps) \log |\cK| - 2h(2\eps) . \]
\end{cor}

\begin{proof}
  Immediate by plugging $\rho_{KE}$ and $\tau_K \tensor \rho_E$ in
  \thmref{thm:AF04}.
\end{proof}

Given the von Neumann entropy of $K$ conditioned on $E$,
$\Ss(K|E)_{\rho}$, one can also upper bound the trace distance of
$\rho_{KE}$ from $\tau_K \tensor \rho_E$ by relating $\Ss(K|E)_{\rho}$
to the relative entropy of $\rho_{KE}$ to $\tau_K \tensor \rho_E$ \---
the relative entropy of $\rho$ to $\sigma$ is defined as $\Ss(\rho \|
\sigma) \coloneqq \trace{\rho \log \rho} - \trace{\rho \log \sigma}$.

\begin{lem}
  \label{lem:entropybound}
  For any quantum state $\rho_{KE}$,
  \[ D(\rho_{KE},\tau_K \tensor \rho_E) \leq
  \sqrt{\frac{1}{2}\left(\log |\cK| - \Ss(K|E)_{\rho}\right)} \ . \]
\end{lem}

\begin{proof}
  From the definitions of the relative and von Neumann entropies we
  have
  \[ \Ss \left(\rho_{KE} \middle\| \tau_K \tensor \rho_E \right) =
  \log |\cK| + \Ss \left(\rho_{KE} \middle\| \id_K \tensor \rho_E
  \right) = \log |\cK| - \Ss(K|E)_{\rho},\] where $\id_K$ is the
  identity matrix. We then use the following bound on the relative
  entropy \cite[Theorem~1.15]{OP93} to conclude the proof:
  \[ \Ss\left(\rho \middle\| \sigma \right) \geq 2
  \left(D(\rho,\sigma)\right)^2. \qedhere\]
\end{proof}

\corref{cor:AF04} and \lemref{lem:entropybound} can be written
together in one equation, upper and lower bounding the conditional von
Neumann entropy:
\[ (1- 8\eps) \log |\cK| - 2h(2\eps) \leq \Ss(K|E)_{\rho} \leq \log
|\cK| - 2\eps^2\ , \] where $\eps = D(\rho_{KE},\tau_K \tensor
\rho_E)$.

\section{Alternative secrecy criterion}
\label{app:alternative}

In \secref{sec:security} we derived two conditions \--- secrecy and
correctness \--- that together imply that a real QKD system is
indistinguishable from the ideal one. An alternative definition for
\emph{$\eps$\=/secrecy} was proposed in the
literature~\cite{TSSR10,TLGR12}:
\begin{equation}\label{eq:alternative.sec}
  (1-\pabort)\min_{\sigma_E}D\left(\rho_{KE}, \tau_K \tensor
    \sigma_E\right) \leq \eps \ .\end{equation}

This alternative notion is equivalent to the standard definition of
secrecy (\eqnref{eq:d}) up to a factor $2$, as can be seen by the
following calculation. Let $\varphi_E$ be the state for which the
minimum in \eqnref{eq:alternative.sec} is achieved. Then,
\begin{align*}
D\left(\rho_{KE}, \tau_K \tensor \rho_E\right) & \leq
D\left(\rho_{KE}, \tau_K \tensor \varphi_E\right) + D\left(\tau_K \tensor
  \varphi_{E},
  \tau_K \tensor \rho_E\right) \\
  & \leq 2 D\left(\rho_{KE}, \tau_K \tensor \varphi_E\right) \ .
\end{align*}
We thus have 
\[ D\left(\rho_{KE}, \tau_K \tensor \rho_E\right) \leq 2
\min_{\sigma_E}D\left(\rho_{KE}, \tau_K \tensor \sigma_E\right) \leq 2
D\left(\rho_{KE}, \tau_K \tensor \rho_E\right)\ .\]
This means that any QKD scheme proven secure with one definition is
still secure according to the other, with a minor adjustment of the
failure parameter $\eps$.

However, we do not know how to derive this alternative notion from a
composable framework. In particular, it is not clear if the failure
$\eps$ from \eqnref{eq:alternative.sec} is additive under parallel
composition. For example, the concatenation of two keys that each,
individually, satisfy \eqnref{eq:alternative.sec}, could possibly have
distance from uniform greater than $2\eps$. For this reason, the arXiv
version of \cite{TLGR12} was updated to use \eqnref{eq:d} instead.

\section{More Abstract Cryptography}
\label{app:moreAC}

\subsection{Distinguishing metric}
\label{app:moreAC.dist}

A distinguisher has been introduced as a single entity that has to
guess which of two systems it is holding. Mathematically, it is more
convenient to model a distinguisher as a set $\fD$. Each
element $D \in \fD$ is a system with $n+1$ interfaces. $n$ of
them connect to the $n$ interfaces of a resource $\aR$ or $\aS$ and
the last interface outputs a bit, as illustrated in
\figref{fig:distinguisher} on \pref{fig:distinguisher}. Thus, for any
$D \in \fD$ and any compatible system $\aR$, $D(\aR)$ is a
binary random variable. The distinguishing advantage can be rewritten
as
\[ d^{\fD}(\aR,\aS) = \max_{D \in \fD} \Pr[D(\aR)=1]
- \Pr[D(\aS)=1] \ .\]

For a set $\fD$ to be a valid distinguisher, it has to be
closed under composition with all resources and converters. For a
converter $\alpha$ and a resource $\aT$, define
\begin{align*}
  D(\alpha \cdot) : \ & \aR \mapsto D(\alpha\aR) \\
  D\left(\cdot \| \aT \right) : \ & \aR \mapsto D\left(\aR \| \aT \right) \ .
\end{align*}
A distinguisher $\fD$ is closed under composition with a set of
converters $\Sigma$ and a set of resources $\Gamma$,\footnote{For a
  set of converters $\Sigma$ and a set of resources $\Gamma$ to be
  valid, they also have to be closed under composition, i.e., for all
  $\alpha,\beta \in \Sigma$ and all $\aR,\aS \in \Gamma$, \[
  \alpha\beta \in \Sigma \ , \qquad \alpha\|\beta \in \Sigma \ ,
  \qquad \alpha\aR \in \Gamma \ , \quad \text{and} \qquad \aR\|\aS \in
  \Gamma \ .\]} if for all $D \in \fD$, all $\alpha \in
\Sigma$ and all $\cT \in \Gamma$,
\begin{equation} \label{eq:moreAC.dist.closure} D(\alpha \cdot) \in
  \fD \qquad \text{and} \qquad D\left(\cdot \| \aT \right)
  \in \fD \ .  \end{equation} For example, the set of all
possible distinguishers is closed under composition with the sets of
all possible converters and resources, and is used for
information\-/theoretic security. The set of all efficient
distinguishers is closed under composition with the sets of all
efficient converters and resources, and is used for computational
security. The fact that the distinguishing advantage is
non\-/increasing under composition (see \eqnref{eq:axioms.nonincrease}
on \pref{eq:axioms.nonincrease}) follows directly from the closure of
the distinguisher, \eqnref{eq:moreAC.dist.closure}.

\subsection{Generic protocol composition}
\label{app:generic}

In this section we briefly sketch why the security criteria of
\defref{def:security} guarantee that the composition of two secure
protocols is also secure. We write up the argument in the case where
an adversary is present (\eqnref{eq:def.sec} from
\defref{def:security}). The case with no adversary follows similarly
with the simulator $\sigma$ removed and a filter connected to Eve's
interface of every resource.
Proofs of this can be found in \cite{MR11,Mau12}.

\subsubsection{Sequential composition}
\label{app:generic.seq}

Let protocols $\pi$ and $\pi'$ construct $\aS_\lozenge$ from $\aR_\sharp$ and $\aT_\square$
from $\aS_\lozenge$ within $\eps$ and $\eps'$, respectively, i.e.,
\[\aR_\sharp \xrightarrow{\pi,\eps}\aS_\lozenge \qquad \text{and} \qquad \aS_\lozenge
\xrightarrow{\pi',\eps'}\aT_\square \ .\] It then follows from the
triangle inequality of the distinguishing metric that $\pi'\pi$
constructs $\aT_\square$ from $\aR_\sharp$ within $\eps+\eps'$, \[\aR_\sharp
\xrightarrow{\pi'\pi,\eps+\eps'}\aT_\square \ .\]

To see why this holds when an adversary is present, note that since
$\pi\aR$ cannot be distinguished from $\aS\sigma$ with advantage
greater than $\eps$, by \eqnref{eq:axioms.nonincrease} a distinguisher
running $\pi'$ in particular cannot distinguish them. Hence
\begin{equation*} \label{eq:composability.security.1} \pi'\pi\aR
  \close{\eps} \pi'\aS\sigma \ .\end{equation*} Likewise, a
distinguisher running $\sigma$ does not know if it is interacting with
$\pi'\aS$ or $\aT\sigma'$, i.e.,
\begin{equation*} \label{eq:composability.security.2} \pi'\aS\sigma
  \close{\eps'} \aT\sigma'\sigma \ .\end{equation*} Combining
the two equations above, we get \[\pi'\pi\aR
\close{\eps+\eps'} \aT\sigma'\sigma \ .\]

So there exists a simulator, namely $\sigma'\sigma$, such that the
real and ideal systems cannot be distinguished with advantage greater
than $\eps+\eps'$. 

\subsubsection{Parallel composition}
\label{app:generic.para}

The argument for parallel composition is similar to that of sequential
composition. Let $\pi$ and $\pi'$ construct $\aS_\lozenge$ and
$\aS'_\square$ from $\aR_\sharp$ and $\aR'_\flat$ within $\eps$ and
$\eps'$, respectively, i.e., \[\aR_\sharp
\xrightarrow{\pi,\eps}\aS_\lozenge \qquad \text{and} \qquad \aR'_\flat
\xrightarrow{\pi',\eps'} \aS'_\square \ .\] If these resources and
protocols are composed in parallel, we find that $\pi \| \pi'$
constructs $\aS_\lozenge \| \aS'_\square$ from $\aR_\sharp \|
\aR'_\flat$ within $\eps+\eps'$, \[\aR_\sharp \| \aR'_\flat
\xrightarrow{\pi \| \pi',\eps+\eps'} \aS_\lozenge \| \aS'_\square \
,\] where $\aR_\sharp \| \aR'_\flat \coloneqq \left(\aR \| \aR',\sharp
  \| \flat\right)$ is the filtered resource consisting of the parallel
composition of the resources and filters from $\aR_\sharp$ and
$\aR'_\flat$.

For the case where an adversary is present, this can be proven as
follows. From the definition of parallel composition of converters in
\eqnref{eq:axioms.order} we have \[\left( \pi \| \pi' \right) \left(
  \aR \| \aR' \right) = \left( \pi \aR \right) \| \left( \pi' \aR'
\right) \ .\] Since for some $\sigma$, $\pi\aR$ cannot be
distinguished from $\aS\sigma$ with advantage greater than $\eps$, by
\eqnref{eq:axioms.nonincrease} running $\pi'\aR'$ in parallel cannot
help the distinguisher, hence \[ \pi\aR \| \pi'\aR' \close{\eps}
\aS\sigma \| \pi'\aR' \ . \] For the same reason we also have \[
\aS\sigma \| \pi'\aR' \close{\eps'} \aS\sigma \| \aS'\sigma' \ . \]
Combining the equations above and one more use of
\eqnref{eq:axioms.order}, we obtain \[\left( \pi \| \pi' \right)
\left( \aR \| \aR' \right) \close{\eps+\eps'} \left( \aS \| \aS'
\right) \left( \sigma \| \sigma' \right) \ .\]

Thus, there exists a simulator, namely $\sigma\|\sigma'$, such that
the real and ideal systems cannot be distinguished with advantage
greater than $\eps+\eps'$. 

\section{Composition of authentication and QKD}
\label{app:ex.auth}

In this section we model recursive composition of an authentication
protocol and QKD. Starting with a short (uniform) key, we construct an
authentic channel, which is then used in a QKD protocol to obtain a
long key. Part of this new key is then consumed in another round of
authentication, which is used by QKD, resulting in more shared secret
key. This can be repeated indefinitely. From the composability of the
security definitions we immediately have that the total failure is
bounded by the sum of the individual failures of (each run of) each
protocol.

The goal of this section is to write this out explicitly: we show that
\begin{equation} \label{eq:auth.recursive} D\left(
    \rho_{A^nB^nE},\tilde{\rho}_{A^nB^nE} \right) \leq
  n (\eps^{\auth}+\eps^{\qkd}) \ ,
\end{equation} where $A^n$ and $B^n$ are registers containing all
shared, unused secret keys generated in $n$ rounds of authentication
and QKD, $E$ contains the adversary's information, $\rho_{A^nB^nE}$
and $\tilde{\rho}_{A^nB^nE}$ are the states obtained by a
distinguisher interacting with the real and ideal
systems,\footnote{Unlike the examples from \secref{sec:ex}, we cannot  
  simplify the state $\tilde{\rho}_{A^nB^nE}$ by conditioning on
  obtaining a key and writing it as $\tau_{A^nB^n} \tensor \rho_E$,
  because the $n$ repetitions of the protocol lead to $n+1$ events: it
  never aborted, aborted after one round, two rounds, etc.} and
$\eps^{\auth}$ and $\eps^{\qkd}$ are the (probabilities of)
failure of the authentication and QKD protocols in each
round.\footnote{One can also use different parameters in each round,
  e.g., so that the failure in round $i$ is half of that in round
  $i-1$, and the sum $\sum_i (\eps^{\auth}_i + \eps^{\qkd}_i)$ is
  bounded for all $n$.}

For this recursive construction it is not necessary to use the
(interactive) authentication protocol of Renner and Wolf~\cite{RW03}
which only requires a password. Instead we use the simpler universal
hashing of Wegman and Carter~\cite{WC81,Sti94}, which appends a
tag\footnote{This tag is often called a message authentication code
  (MAC) in the literature.} to the message that is sent, but requires
an (almost) uniform key. A more detailed analysis of this
authentication method, including key recycling and a proof that strong
universal hashing meets the corresponding security definition can be
found in \cite{Por14}.

In \appendixref{app:ex.auth.model} we sketch how to construct an
authentic channel from a shared secret key and insecure channel
resource, and provide a generic
simulator. In \appendixref{app:ex.auth.auth} we compose multiple
authentication protocols in parallel so that we may have multiple use
authentic channels in QKD. In \appendixref{app:ex.auth.qkd} we compose
such a construction with QKD, obtaining a key expansion protocol. And
finally in \appendixref{app:ex.auth.exp} we iteratively compose a key
expansion protocol with itself, resulting in a continuous stream of
secret key bits.

\subsection{Authentication}
\label{app:ex.auth.model}

The QKD and one-time pad protocols discussed in this work make use of
authentic channels as depicted in \figref{fig:auth.resource.simple},
which always deliver the correct message to the receiver. This is
however impossible to construct from an insecure channel, since Eve
can always cut the communication between Alice and Bob, and prevent
any message from being transmitted. What can be constructed, is a
channel which guarantees that Bob does not receive a corrupted
message. He either receives the correct message sent by Alice, or an
error, which symbolizes an attempt by Eve to change the message. This
can be modeled by giving Eve's idealized interface two controls: the
first provides her with Alice's message, the second allows her to
input one bit that prevents Alice's message from being delivered to
Bob and produces an error instead. We illustrate this in
\figref{fig:auth.resource.switch}.

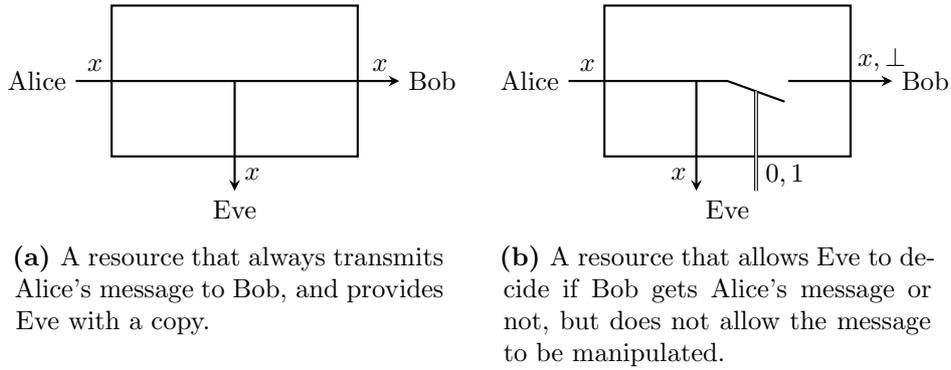
\begin{figure}[htb]
  \subcaptionbox[Simple authentic
  channel]{\label{fig:auth.resource.simple}A resource that always
    transmits Alice's message to Bob, and provides Eve with a
    copy.}[.5\textwidth][l]{
\begin{tikzpicture}\small

\node[largeResource] (keyBox) at (0,0) {};
\node (alice) at (-2.6,0) {Alice};
\node (bob) at (2.6,0) {Bob};
\node (eve) at (0,-1.7) {Eve};
\node (ajunc) at (eve.north |- alice) {};

\draw[sArrow] (alice) to node[pos=.06,auto] {$x$} node[pos=.94,auto] {$x$} (bob);
\draw[sArrow] (ajunc.center) to node[pos=.85,auto] {$x$} (eve);

\end{tikzpicture}
} \subcaptionbox[Authentic channel with
switch]{\label{fig:auth.resource.switch}A resource that allows Eve to
  decide if Bob gets Alice's message or not, but does not allow the
  message to be manipulated.}[.5\textwidth][r]{
\begin{tikzpicture}\small

\node[largeResource] (keyBox) at (0,0) {};
\node (alice) at (-2.6,0) {Alice};
\node (bob) at (2.6,0) {Bob};
\node (eve) at (0,-1.7) {Eve};
\node (ajunc) at (eve.north west |- alice) {};

\draw[thick] (alice) to node[pos=.12,auto] {$x$} (0,0) to node[pos=.5]
(ejunc) {} +(160:-.8);
\draw[sArrow] (ajunc.center) to node[pos=.85,auto,swap] {$x$} (eve.north west);
\draw[sArrow] (.8,0) to node[pos=.9,auto] {$x,\bot$} (bob);
\draw[double] (ejunc.center |- eve.north) to node[pos=.15,auto,swap] {$0,1$} (ejunc.center);

\end{tikzpicture}
} \caption[Authentic channel resources]{\label{fig:auth.resource}The
  authentic channel on the right can be constructed by an
  authentication protocol. The one on the left is a simplification as
  used in the one-time pad (\figref{fig:otp.real}) or QKD
  (\figref{fig:qkd.real}) constructions.}
\end{figure}

A construction of this authentic channel resource from an insecure
channel and a shared secret key resource is typically accomplished by
computing the hash $h_k(x)$ of the message $x$, and sending the string
$x \| h_k(x)$ to Bob, where $k$ is the shared secret key and
$\{h_k\}_{k \in \cK}$ a family of hash functions~\cite{Sti94}. Alice's
part of the authentication protocol $\pi^{\auth}_A$ thus gets a key
$k$ from an ideal key resource, a message $x$ from Alice, and sends $x
\| h_k(x)$ down the insecure channel. When Bob receives a string $x'
\| y'$, he needs to check whether $y' = h_k(x')$. His part of the
protocol gets a key $k$ from an ideal key resource, a message $x' \|
y'$ from the channel, and outputs $x'$ if $y' = h_k(x')$, otherwise an
error $\bot$. If the ideal key resource used by both players produces
an error instead of a key, Alice and Bob abort, and the protocol is
trivially secure. So for simplicity we omit this possibility in the
following, and assume that they always get a shared secret key. This
is depicted in \figref{fig:auth.real}.

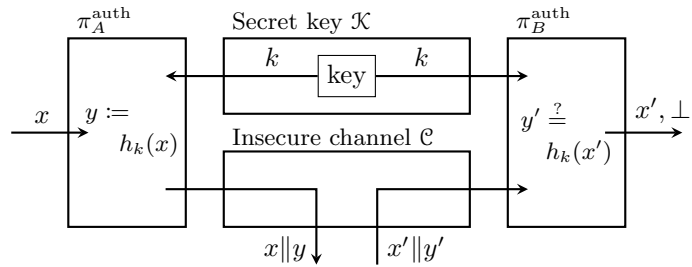
\begin{figure}[htb]
\begin{centering}

\begin{tikzpicture}\small

\def\t{4.413} 
\def\u{2.89} 
\def\v{.75}

\node[pnode] (a1) at (-\u,\v) {};
\node[pnode] (a2) at (-\u,0) {};
\node[pnode] (a3) at (-\u,-\v) {};
\node[protocol,text width=1.1cm] (a) at (-\u,0) {\footnotesize $y
  \coloneqq $\\$\quad \ h_k(x)$};
\node[yshift=-2,above right] at (a.north west) {\footnotesize
  $\pi^{\auth}_A$};
\node (alice) at (-\t,0) {};

\node[pnode] (b1) at (\u,\v) {};
\node[pnode] (b2) at (\u,0) {};
\node[pnode] (b3) at (\u,-\v) {};
\node[protocol,text width=1.2cm] (b) at (\u,0) {\footnotesize $y' \stackrel{?}{=}$\\$\quad h_k(x')$};
\node[yshift=-2,above right] at (b.north west) {\footnotesize $\pi^{\auth}_B$};
\node (bob) at (\t,0) {};

\node[thinResource] (keyBox) at (0,\v) {};
\node[draw] (key) at (0,\v) {key};
\node[yshift=-2,above right] at (keyBox.north west) {\footnotesize
  Secret key $\aK$};
\node[thinResource] (channel) at (0,-\v) {};
\node[yshift=-1.5,above right] at (channel.north west) {\footnotesize
  Insecure channel $\aC$};
\node (eveleft) at (-.4,-1.75) {};
\node (everight) at (.4,-1.75) {};
\node (ajunc) at (eveleft |- a3) {};
\node (bjunc) at (everight |- b3) {};

\draw[sArrow] (key) to node[auto,swap,pos=.3] {$k$} (a1);
\draw[sArrow] (key) to node[auto,pos=.3] {$k$} (b1);

\draw[sArrow] (alice.center) to  node[auto,pos=.4] {$x$} (a2);
\draw[sArrow] (b2) to node[auto,pos=.75] {$x',\bot$} (bob.center);

\draw[sArrow] (a3) to (ajunc.center)
to node[pos=.8,auto,swap] {$x \| y$} (eveleft.center);
\draw[sArrow] (everight.center) to node[pos=.2,auto,swap] {$x' \| y'$}
(bjunc.center) to (b3);

\end{tikzpicture}

\end{centering}
\caption[Real authentication system]{\label{fig:auth.real}The real
  authentication system \--- Alice has access to the left interface,
  Bob to the right interface and Eve to the lower interface \---
  consists of the authentication protocol
  $(\pi^{\auth}_A,\pi^{\auth}_B)$, and the secret key and insecure
  channel resources, $\aK$ and $\aC$.}
\end{figure}

In the case where no adversary is present and filters\footnote{Unlike
  the case of QKD, here we are not interested in a filter which
  introduces (honest) noise on the channel, as this can be removed by
  encoding the communication with an appropriate error correcting
  code. Therefore, the filter on the (real) insecure channel
  faithfully forwards the message, and the filter on the (ideal)
  authentic channel allows the message to be transmitted.} cover Eve's
interfaces of Figures~\ref{fig:auth.resource.switch} and
\ref{fig:auth.real}, the real and ideal systems are indistinguishable
as they are both identity channels which faithfully transmit $x$ from
Alice to Bob. So in the following we only consider the case where an
adversary is present, condition~\eqref{eq:def.sec} in
\defref{def:security}.

In the ideal setting, the authentic channel
(\figref{fig:auth.resource.switch}) has the same interface on Alice's
and Bob's sides as the real setting (\figref{fig:auth.real}): Alice
can input a message, and Bob receives either a message or an
error. However, Eve's interface looks quite different: in the real
setting she can modify the transmission on the insecure channel,
whereas in the ideal setting the adversarial interface provides only
controls to read the message and interrupt the transmission. From
\defref{def:security} we have that an authentication protocol
constructs the authentic channel if there exists a simulator
$\sigma^{\auth}_E$ that can recreate the real interface while
accessing just the idealized one. An obvious choice for the simulator
is to first generate its own key $k$ and output $x \| h_k(x)$. Then
upon receiving $x' \| y'$, it checks if $x' \| y' = x \| h_k(x)$ and
cuts the transmission on the authentic channel if this does not
hold. We illustrate this in \figref{fig:auth.ideal}.

\begin{figure}[htb]
\begin{centering}

\begin{tikzpicture}\small

\def\t{2.368} 
\def\u{-1.1}
\def\v{.75}
\def\w{-2.45} 

\node[thinResource] (channel) at (0,\v) {};
\node[yshift=-1.5,above right] at (channel.north west) {\footnotesize
  Authentic channel $\aA$};
\node (alice) at (-\t,\v) {};
\node (bob) at (\t,\v) {};

\node[simulator] (sim) at (0,\u) {};
\node[xshift=1.5,below left] at (sim.north west) {\footnotesize
  $\sigma^{\auth}_E$};
\node[snode,ellipse] (sleft) at (-.809,\u) {};
\node[snode] (sright) at (.809,\u) {};
\draw[dashed] (sim.north) to (sim.south);

\node (ajunc) at (sleft |- alice) {};
\node (bjunc) at (sright |- bob) {};

\draw[thick] (alice.center) to node[pos=.15,auto] {$x$} (.4,\v) to node[pos=.54] (ejunc) {} +(160:-.8);
\draw[sArrow] (ajunc.center) to node[pos=.63,auto,swap] {$x$} (sleft);
\draw[sArrow] (1.2,\v) to node[pos=.75,auto] {$x,\bot$} (bob.center);
\draw[double] (sright) to node[pos=.4,auto,swap] {$0,1$} (ejunc.center);

\node (sltext) at (-.809,\u) {\footnotesize $y = h_k(x)$};
\node[text width=1.2cm] (srtext) at (.809,\u) {\footnotesize $x \| y
  \stackrel{?}{=}$\\$\quad \ x' \| y'$};

\node (eveleft) at (sleft |- 0,\w) {};
\node (everight) at (sright |- 0,\w) {};
\draw[sArrow] (sleft) to node[pos=.75,auto,swap] {$x \| y$} (eveleft.center);
\draw[sArrow] (everight.center) to node[pos=.25,auto,swap] {$x' \| y'$}
(sright);

\node[draw,fill=white] (key) at (.15,\u+.8) {key};
\draw[sArrow] (key) to (sleft);

\end{tikzpicture}

\end{centering}
\caption[Ideal authentication system]{\label{fig:auth.ideal}The ideal
  authentication system \--- Alice has access to the left interface,
  Bob to the right interface and Eve to the lower interface \---
  consists of the ideal authentication resource and a simulator
  $\sigma^{\auth}_E$.}
\end{figure}
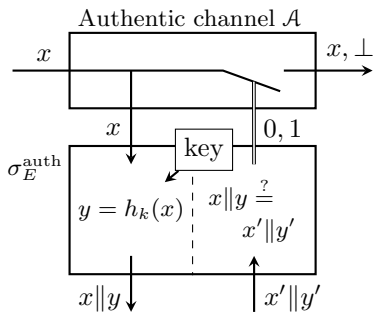


An authentication protocol is then $\eps$\=/secure if
Figures~\ref{fig:auth.real} and \ref{fig:auth.ideal} are
$\eps$\=/close, i.e.,
\begin{equation} \label{eq:security.auth} \pi^{\auth}_A\pi^{\auth}_B
  \left(\aK \| \aC \right) \close{\eps} \aA\sigma^{\auth}_E \
  .\end{equation} Portmann~\cite{Por14} showed that
\eqnref{eq:security.auth} is satisfied if the hash functions used are
$\eps$\=/almost strongly universal$_2$.\footnote{A familly of
  functions is said to be almost strongly universal$_2$ if any two
  different messages are almost uniformly mapped to all pairs of
  tags.}

\subsection{Parallel composition of authentication}
\label{app:ex.auth.auth}

In \appendixref{app:ex.auth.model} we modeled one run of an
authentication protocol, that constructs a one-time use authentic
channel. In general, QKD protocols require multiple rounds of
authenticated communication. This is achieved by running the same
protocol in parallel multiple times with new keys. It is
straightforward from \eqnsref{eq:axioms.order}, \eqref{eq:pm.tri},
\eqref{eq:axioms.nonincrease} and \eqref{eq:security.auth} that $\ell$
parallel repetitions of the authentication protocol are
$\ell\eps$\=/close to $\ell$ ideal authentic channels and simulators
in parallel,
\begin{align*}
  & \left(\pi^{\auth}_A\pi^{\auth}_B \| \dotsb \|
    \pi^{\auth}_A\pi^{\auth}_B \right) \left(\aK \| \aC \| \dotsb \|
    \aK \| \aC \right) \\ & \qquad \qquad =
  \left(\pi^{\auth}_A\pi^{\auth}_B \left(\aK \| \aC \right) \right) \|
  \dotsb \| \left(\pi^{\auth}_A\pi^{\auth}_B \left(\aK \| \aC
    \right) \right) \\
  & \qquad \qquad \close{\ell\eps} \left(\aA \sigma^{\auth}_E\right)
  \|
  \dotsb \| \left(\aA \sigma^{\auth}_E\right) = \left(\aA \| \dotsb \| \aA \right)
  \left(\sigma^{\auth}_E \| \dotsb \| \sigma^{\auth}_E \right) \ . \end{align*}

In the following, when we speak of the authentication used in QKD we
always refer to parallel repetitions of the protocol that construct
multiple use authenticated channels $\aA \| \dotsb \| \aA$. For
simplicity, we use the same notation for multiple authentic channels
as we have for single channels \--- we denote the resulting multiple
use authentic channel by $\aA$, as well as $\aC$ for the multiple use
insecure channel, $\aK$ for a key sufficiently long for authenticating
every message and $\eps^{\auth}$ for the accumulated failure of all
parallel repetitions in one round of QKD.

\subsection{Sequential composition of authentication and key distribution}
\label{app:ex.auth.qkd}

Let $\pi^{\auth}$ be an authentication protocol which constructs with
failure $\eps^{\auth}$ a (multiple use) authentic channel $\aA$ from a
short secret key of length $\ell$, $\aK^\ell$, and an insecure
classical channel $\aC$, \[ \aK^\ell \| \aC
\xrightarrow{\pi^{\auth},\eps^{\auth}} \aA \ .\] Let $\pi^{\qkd}$ be a
QKD protocol which constructs with failure $\eps^{\qkd}$ a long secret
key of length $m$, $\aK^m$, from an authentic channel $\aA$ and an
insecure quantum channel $\aQ$, \[ \aA \| \aQ
\xrightarrow{\pi^{\qkd},\eps^{\qkd}} \aK^m \ .\] By sequentially
composing the two protocols, we immediately have that
$\pi^{\qkd}\pi^{\auth}$ constructs a long secret key $\aK^m$ from a
short secret key $\aK^\ell$ and insecure channels $\aC$ and $\aQ$,
with failure $\eps^{\auth}+\eps^{\qkd}$, \[ \aK^\ell \| \aC \| \aQ
\xrightarrow{\pi^{\qkd}\pi^{\auth},\eps^{\auth}+\eps^{\qkd}} \aK^m \
.\] The generic argument for sequential composition is given
in \appendixref{app:generic.seq}. Here we illustrate it in the special
case of authentication and QKD, and draw it in \figref{fig:ex.auth}.

\begin{figure}[htbp]
\begin{subfigure}[b]{\textwidth}
\begin{centering}

\begin{tikzpicture}\small

\def\t{6.067} 
\def\e{-2.5}
\def\u{2.25}
\def\uu{4.545} 
\def\v{.75}
\def\w{.9}

\node[pnode] (a1) at (-\u,\v) {};
\node[pnode] (a2) at (-\u,0) {};
\node[pnode] (a3) at (-\u,-\v) {};
\node[protocol] (a) at (-\u,0) {};
\node[yshift=-2,above right] at (a.north west) {\footnotesize
  $\pi^{\auth}_A$};
\node (alice) at (-\t,-\v) {};

\node[pnode] (b1) at (\u,\v) {};
\node[pnode] (b2) at (\u,0) {};
\node[pnode] (b3) at (\u,-\v) {};
\node[protocol] (b) at (\u,0) {};
\node[yshift=-2,above right] at (b.north west) {\footnotesize $\pi^{\auth}_B$};
\node (bob) at (\t,-\v) {};

\node[thinResource,scale=.7] (keyBox) at (0,\v) {};
\node[lrnode,scale=.7] (key) at (0,\v) {}; 
\node[yshift=-2,above right] at (keyBox.north west) {\footnotesize
  Secret key $\aK^\ell$};

\node[pnode] (f1) at (-\uu,0) {};
\node[pnode] (f2) at (-\uu,-\v) {};
\node[pnode] (f3) at (-\uu,-2*\v) {};
\node[protocol,gray] (f) at (-\uu,-\v) {};
\node[gray,yshift=-2,above right] at (f.north west) {\footnotesize
  $\pi^{\qkd}_A$};

\node[pnode] (g1) at (\uu,0) {};
\node[pnode] (g2) at (\uu,-\v) {};
\node[pnode] (g3) at (\uu,-2*\v) {};
\node[protocol,gray] (g) at (\uu,-\v) {};
\node[gray,yshift=-2,above right] at (g.north west) {\footnotesize
  $\pi^{\qkd}_B$};

\node (eveq1) at (-2*\w,\e) {};
\node (juncq1) at (eveq1 |- f3) {};
\node (eveq2) at (-\w,\e) {};
\node (juncq2) at (eveq2 |- f3) {};
\node (evec1) at (0,\e) {};
\node (juncc1) at (evec1 |- a3) {};
\node (evec2) at (\w,\e) {};
\node (juncc2) at (evec2 |- a3) {};

\draw[sArrow] (key) to (a1);
\draw[sArrow] (key) to (b1);

\draw[sArrow,<->] (a2) to node[auto,pos=.5,swap] {\footnotesize $t,\bot$} (f1);
\draw[sArrow,<->] (b2) to node[auto,pos=.5] {\footnotesize $t',\bot$} (g1);

\draw[sArrow,<->] (a3) to (juncc1.center) to node[pos=.85,auto,swap] {$z$} (evec1.center);
\draw[sArrow,<->] (evec2.center) to node[pos=.19,auto] {$z'$} (juncc2.center) to (b3);

\draw[gray,sArrow] (f3) to (juncq1.center) to node[pos=.7,auto,swap] {$\rho$} (eveq1.center);
\draw[gray,sArrow] (eveq2.center) to node[pos=.34,auto] {$\rho'$} (juncq2.center) to (g3);

\draw[gray,sArrow] (f2) to node[auto,pos=.8,swap] {$k_A,\bot$} (alice.center);
\draw[gray,sArrow] (g2) to node[auto,pos=.8] {$k_B,\bot$} (bob.center);

\end{tikzpicture}

\end{centering}
\caption[Authentication \& QKD]{\label{fig:ex.auth.real}The
  composition of a QKD and authentication protocols. The insecure
  channels have not been depicted as boxes to simplify the figure.}
\end{subfigure}

\vspace{12pt}

\begin{subfigure}[b]{\textwidth}
\begin{centering}

\begin{tikzpicture}\small

\def\t{4.913} 
\def\u{-1.85} 
\def\w{-3.95} 
\def\v{0} 
\def\s{1.05}
\def\f{3.5}
\def\fh{-1.5}
\def\vv{-3}

\node[thinResource] (channel) at (0,\v) {};
\node[lrnode] (c0) at (0,\v) {};
\node[yshift=-1.5,above right] at (channel.north west) {\footnotesize
  Authentic channel $\aA$};
\node (alice) at (-\t,\fh) {};
\node (bob) at (\t,\fh) {};

\node[draw,thick,gray,minimum width=1.545cm,minimum height=4cm] (f) at (-\f,\fh) {};
\node[pnode] (f1) at (-\f,\v) {};
\node[pnode] (f2) at (-\f,\fh) {};
\node[pnode] (f3) at (-\f,\vv) {};
\node[gray,yshift=-2,above right] at (f.north west) {\footnotesize
  $\pi^{\qkd}_A$};

\node[draw,thick,gray,minimum width=1.545cm,minimum height=4cm] (g) at (\f,\fh) {};
\node[pnode] (g1) at (\f,\v) {};
\node[pnode] (g2) at (\f,\fh) {};
\node[pnode] (g3) at (\f,\vv) {};
\node[gray,yshift=-2,above right] at (g.north west) {\footnotesize
  $\pi^{\qkd}_B$};

\node[simulator,dashed] (sim) at (0,\u) {};
\node[xshift=-1.5,below right] at (sim.north east) {\footnotesize
  $\sigma^{\auth}_E$};
\node[snode] (sleft) at (-.809,\u) {};
\node[snode] (sright) at (.809,\u) {};
\node[snode] (srright) at (\s,\u) {};
\node[snode] (scenter) at (0,\u) {};
\node[tlrnode] (c1) at (sleft |- f1) {};
\node[tlrnode] (c2) at (sright |- g1) {};

\draw[sArrow,<->] (c0) to node[pos=.5,auto,swap] {$t,\bot$} (f1);
\draw[sArrow,<->] (c0) to node[pos=.5,auto] {$t,\bot$} (g1);
\draw[sArrow] (c1) to node[pos=.5,auto,swap] {$t$} (sleft);
\draw[double] (sright) to node[pos=.5,auto,swap] {$0,1$} (c2);

\node (eveq1) at (-2*\s,\w) {};
\node (juncq1) at (eveq1 |- f3) {};
\node (eveq2) at (-\s,\w) {};
\node (juncq2) at (eveq2 |- f3) {};
\node (evec1) at (0,\w) {};
\node (evec2) at (\s,\w) {};

\draw[dashed,sArrow,<->] (scenter) to node[pos=.8,auto,swap] {$z$} (evec1.center);
\draw[dashed,sArrow,<->] (srright) to node[pos=.78,auto,swap] {$z'$} (evec2.center);

\draw[gray,sArrow] (f2) to node[auto,pos=.8,swap] {$k_A,\bot$} (alice.center);
\draw[gray,sArrow] (g2) to node[auto,pos=.8] {$k_B,\bot$} (bob.center);
\draw[gray,sArrow] (f3) to (juncq1.center) to node[pos=.7,auto,swap] {$\rho$} (eveq1.center);
\draw[gray,sArrow] (eveq2.center) to node[pos=.34,auto] {$\rho'$} (juncq2.center) to (g3);

\end{tikzpicture}

\end{centering}
\caption[Ideal authentic channel \& QKD]{\label{fig:ex.auth.hybrid}A
  hybrid system consisting of a QKD protocol and two-way authentic
  channels with simulator.}
\end{subfigure}

\vspace{12pt}

\begin{subfigure}[b]{\textwidth}
\begin{centering}

\begin{tikzpicture}\small

\def\b{5.336cm} 
\def\t{3.418} 
\def\u{-1.3} 
\def\uu{-2.4} 
\def\s{1.05}
\def\w{-3.45} 

\node[draw,thick,minimum width=\b,minimum height=1cm] (channel) at (0,0) {};
\node[minimum width=\b-.5cm,minimum height=.2cm] (c0) at (0,0) {};
\node[yshift=-1.5,above right] at (channel.north west) {\footnotesize
  Secret key $\aK^m$};
\node (alice) at (-\t,0) {};
\node (bob) at (\t,0) {};

\node[draw,thick,minimum width=1.618*2cm,minimum height=.6cm] (simqkd) at (-\s,\u) {};
\node[xshift=1.5,left] at (simqkd.west) {\footnotesize
  $\sigma^{\qkd}_E$};
\node[inner sep=0,minimum height=.2cm,minimum width=.2cm] (pleft) at (-2*\s,\u) {};
\node[inner sep=0,minimum height=.2cm,minimum width=.2cm] (pcenter) at (-\s,\u) {};
\node[inner sep=0,minimum height=.2cm,minimum width=.2cm] (pright) at (0,\u) {};

\node[draw,thick,minimum width=1.618*2cm,minimum height=.6cm,dashed] (simauth) at (\s,\uu) {};
\node[xshift=-1.5,right] at (simauth.east) {\footnotesize
  $\sigma^{\auth}_E$};
\node[inner sep=0,minimum height=.2cm,minimum width=.2cm] (qleft) at (0,\uu) {};
\node[inner sep=0,minimum height=.2cm,minimum width=.2cm] (qcenter) at (\s,\uu) {};
\node[inner sep=0,minimum height=.2cm,minimum width=.2cm] (qright) at (2*\s,\uu) {};

\node[tlrnode] (ajunc) at (pcenter |- alice) {};
\node[tlrnode] (bjunc) at (qcenter |- bob) {};
\draw[sArrow] (c0) to node[pos=.8,auto,swap] {$k,\bot$} (alice);
\draw[sArrow] (c0) to node[pos=.8,auto] {$k,\bot$} (bob);
\draw[double] (pcenter) to node[pos=.45,auto,swap] {$0,1$} (ajunc);
\draw[double] (qcenter) to node[pos=.48,auto,swap] {$0,1$} (bjunc);

\node (peve1) at (-2*\s,\w) {};
\node (peve2) at (-\s,\w) {};
\node (qeve1) at (\s,\w) {};
\node (qeve2) at (2*\s,\w) {};

\draw[sArrow] (pleft) to node[pos=.84,auto,swap] {$\rho$} (peve1);
\draw[sArrow] (peve2) to node[pos=.18,auto] {$\rho'$} (pcenter);
\draw[sArrow] (pright) to node[pos=.5,auto] {$t$} (qleft);
\draw[sArrow,<->,dashed] (qcenter) to node[pos=.6,auto,swap] {$z$} (qeve1);
\draw[sArrow,<->,dashed] (qright) to node[pos=.55,auto,swap] {$z'$} (qeve2);

\end{tikzpicture}

\end{centering}
\caption[Ideal key]{\label{fig:ex.auth.ideal}The ideal secret key and
  corresponding composed simulator
  $\sigma^{\qkd}_E\sigma^{\auth}_E$. This ideal key resource has two
  switches preventing the key from being generated: one to capture an abort
  from the authentication protocol and one from the QKD protocol.}
\end{subfigure}

\caption[Example: serial composition of authentication and
QKD]{\label{fig:ex.auth}Steps in the security proof of the composition
  of QKD and an authentication protocol. If we remove the gray parts
  from Figures~\ref{fig:ex.auth.real} and \ref{fig:ex.auth.hybrid}, we
  recover the real and ideal systems of authentication. If we remove
  the dashed parts from Figures~\ref{fig:ex.auth.hybrid} and
  \ref{fig:ex.auth.ideal} we recover the real and ideal systems of
  QKD.}
\end{figure}
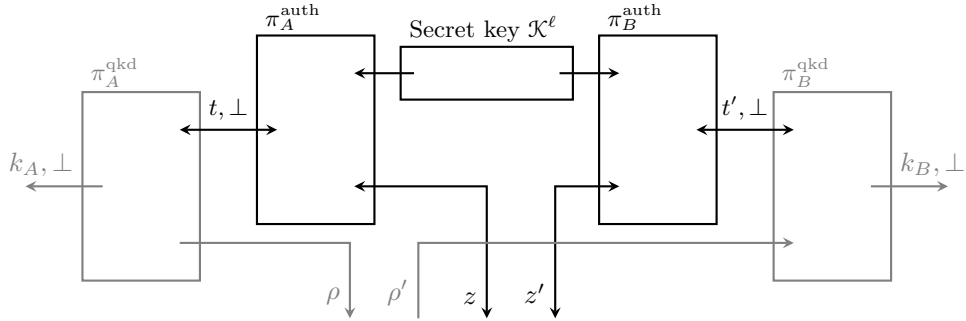
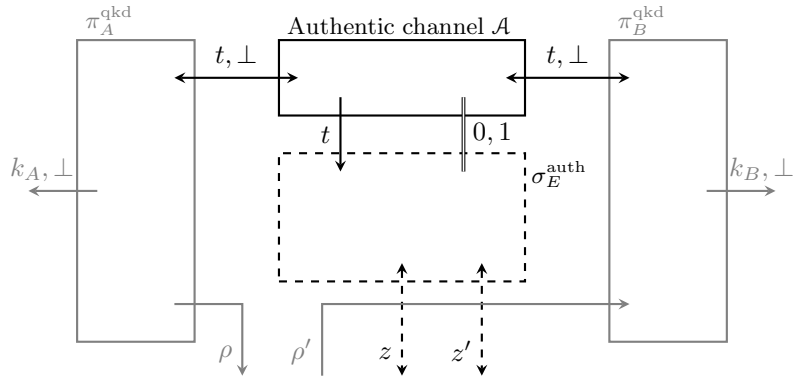
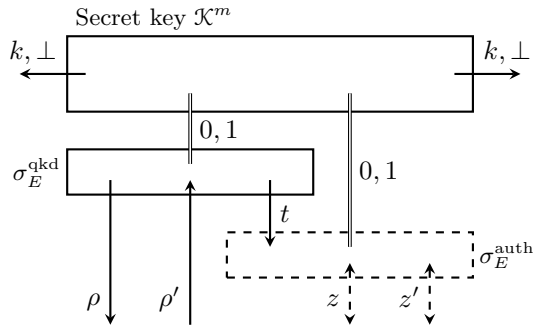

\figref{fig:ex.auth.real} depicts the real world: the two protocols
are composed in sequence and run using the short key and insecure
channel resources. In \figref{fig:ex.auth.hybrid} we have a system
consisting of the real QKD protocol and the ideal authentic channel
and simulator. We know that the black parts of
Figures~\ref{fig:ex.auth.real} and \ref{fig:ex.auth.hybrid} are
$\eps^{\auth}$\=/close, so adding the QKD protocol in gray can only
reduce the distance. \figref{fig:ex.auth.ideal} depicts the ideal
secret key resource and simulators. By removing the dashed simulator
$\sigma^{\auth}_E$ from Figures~\ref{fig:ex.auth.hybrid} and
\ref{fig:ex.auth.ideal}, we recover the real and ideal QKD systems
from Figures~\ref{fig:qkd.real.adv} and \ref{fig:qkd.ideal} \--- with
an extra switch on the authentic channel in the real system and on the
secret key resource in the ideal system. Since these are
$\eps^{\qkd}$\=/close, so are Figures~\ref{fig:ex.auth.hybrid} and
\ref{fig:ex.auth.ideal}. Putting the two statements together with the
triangle inequality finishes the argument.

This reasoning is summed up in the following equation, which can be
directly derived from \eqnsref{eq:axioms.order}, \eqref{eq:pm.tri} and
\eqref{eq:axioms.nonincrease}:
\begin{multline*} \left.\begin{aligned}
\pi^{\auth}_A \pi^{\auth}_B \left( \aK^\ell \| \aC \right) & \close{\eps^{\auth}}
\aA \sigma^{\auth}_E \\
\pi^{\qkd}_A \pi^{\qkd}_B \left( \aA \| \aQ \right) & \close{\eps^{\qkd}} \aK^m \sigma^{\qkd}_E
\end{aligned}\right\}
\implies \\ \pi^{\qkd}_A \pi^{\qkd}_B\pi^{\auth}_A \pi^{\auth}_B
\left( \aK^\ell \| \aC \| \aQ \right) \close{(\eps^{\auth}+\eps^{\qkd})}
\aK^m \sigma^{\qkd}_E\sigma^{\auth}_E \ . \end{multline*}

Let $\rho_{ABE}$ be the state gathered by a distinguisher interacting
with the real system from \figref{fig:ex.auth.real} and
$\tilde{\rho}_{ABE}$ be the state gathered by the distinguisher
interacting with the ideal system from \figref{fig:ex.auth.ideal}. By
the argument above we have
\[D\left( \rho_{ABE},\tilde{\rho}_{ABE}\right) \leq \eps^{\auth} +
\eps^{\qkd} \ .\]

\subsection{Iterated key expansion}
\label{app:ex.auth.exp}

In \appendixref{app:ex.auth.qkd} we show that the composition of QKD
and authentication \--- i.e., key expansion \--- constructs a long key
from a short key and insecure channels. To show that this can be done
recursively, we need to argue that part of the long key can be kept
for the next round of key expansion. So far the secret keys have been
treated as blocks, entirely consumed by a protocol, which is not
convenient for the analysis of a protocol that uses only part of a
key. Instead, we should think of these key resources \--- e.g.,
$\aK^\ell$ and $\aK^m$ in Figures~\ref{fig:ex.auth.real} and
\ref{fig:ex.auth.ideal} \--- as a parallel composition of resources
that produce a single bit of key, i.e., $\aK^\ell = \aK^1_1 \| \dotsb
\| \aK^1_\ell$, where $\aK^1_i$ is the $\ith{i}$ instance of a
resource $\aK^1$ that produces one bit of key (or an error message) at
Alice and Bob's interfaces, and has a switch at Eve's interface that
decides if it produces the key or error.

Then, a proof that
\[ \aK^\ell \| \aC \| \aQ
\xrightarrow{\pi^{\qkd}\pi^{\auth},\eps^{\auth}+\eps^{\qkd}} \aK^m \
,\]
is immediately also a proof that
\[ \aK^{\ell'} \| \aC \| \aQ
\xrightarrow{\pi^{\qkd}\pi^{\auth},\eps^{\auth}+\eps^{\qkd}}
\aK^{m+\ell'-\ell} \ ,\] for any $\ell' \geq \ell$.
Iterating the protocol $n$ times we get
\[ \aK^{\ell} \| \aC^n \| \aQ^n
\xrightarrow{\left(\pi^{\qkd}\pi^{\auth}\right)^n,n(\eps^{\auth}+\eps^{\qkd})}
\aK^{nm-(n-1)\ell} \ ,\] where $\aC^n$ and $\aQ^n$ are $n$ instances
of the resources $\aC$ and $\aQ$ in parallel, and
$\left(\pi^{\qkd}\pi^{\auth}\right)^n$ is $n$ times the sequential
composition of $\pi^{\qkd}\pi^{\auth}$. \eqnref{eq:auth.recursive}
follows immediately from this.

\section*{Acknowledgements}
\addcontentsline{toc}{section}{Acknowledgements}

We are greatly indebted to the following people for having proofread
an initial draft and provided us with invaluable feedback and
comments: Rotem Arnon Friedman, Normand Beaudry, Vedran Dunjko, Felipe
Lacerda, Charles Ci Wen Lim, Christoph Pacher, Joseph Renes, Marco
Tomamichel, and Nino Walenta.

This work has been funded by the Swiss National Science Foundation
(via grant No.~200020-135048 and the National Centre of Competence in
Research `Quantum Science and Technology'), the European Research
Council \--- ERC (grant No.~258932) \--- and by the Vienna Science and
Technology Fund (WWTF) through project ICT10-067 (HiPANQ).



\newcommand{\etalchar}[1]{$^{#1}$}
\providecommand{\bibhead}[1]{}
\expandafter\ifx\csname pdfbookmark\endcsname\relax%
  \providecommand{\tocrefpdfbookmark}{}
\else\providecommand{\tocrefpdfbookmark}{%
   \phantomsection%
   \addcontentsline{toc}{section}{\refname}}%
\fi

\tocrefpdfbookmark

\end{document}